\documentclass[journal]{IEEEtran}
\usepackage{amsbsy}%,color,multicol}
\usepackage{floatflt} % text flows around figures

\usepackage{amsmath}
\usepackage{amssymb}
\usepackage{times}
\usepackage{graphics}
\usepackage{graphicx}
\usepackage{xspace}
\usepackage{paralist} % compact and in-paragraph* lists
\usepackage{setspace} % buggy somehow
\usepackage{xypic}
\xyoption{curve}
\usepackage{latexsym}
\usepackage{theorem}
\usepackage{ifthen}
\usepackage{subfigure}
\usepackage{booktabs}
\usepackage{algorithm}
\usepackage{algorithmic}
\usepackage{color}
\ifCLASSINFOpdf
\else
\fi

{\theoremheaderfont{\it} \theorembodyfont{\rmfamily}
\newtheorem{theorem}{Theorem}
\newtheorem{lemma}[theorem]{Lemma}
\newtheorem{definition}[theorem]{Definition}

}

\newcounter{brojac}

{\theoremheaderfont{\it} \theorembodyfont{\rmfamily}
\newtheorem{assumption}[brojac]{Assumption}
}

\begin{document}
\title{Fast Distributed Gradient Methods}

\author{Du$\check{\mbox{s}}$an Jakoveti\'c,~\IEEEmembership{Student~Member,~IEEE,}
        Jo\~ao Xavier,~\IEEEmembership{Member,~IEEE,}
        and~Jos\'e M.~F.~Moura,~\IEEEmembership{Fellow,~IEEE}% <-this % stops a space
        \thanks{The work of the first and second authors was supported by: the Carnegie Mellon$|$Portugal Program under a grant from the Funda\c{c}\~ao de Ci$\hat{\mbox{e}}$ncia e Tecnologia~(FCT) from Portugal; by FCT grants CMU-PT/SIA/0026/2009 and SFRH/BD/33518/2008 (through the Carnegie Mellon$|$Portugal Program managed by ICTI); by ISR/IST plurianual funding (POSC program, FEDER), and the work of the first and third authors was funded by AFOSR grant~FA95501010291 and by NSF grant~CCF1011903, while the first author was a doctoral or postdoctoral student within the Carnegie Mellon$|$Portugal Program. D. Jakoveti\'c is with University of Novi Sad, BioSense Center, 21000 Novi Sad, Serbia. J. Xavier is with Instituto de Sistemas e Rob\'otica~(ISR), Instituto Superior T\'ecnico~(IST), University of Lisbon, 1049-001 Lisbon, Portugal.  J.~M.~F.~Moura is with Department of Electrical and Computer Engineering, Carnegie Mellon University, Pittsburgh, PA 15213, USA. Authors e-mails: djakovet@uns.ac.rs,
jxavier@isr.ist.utl.pt, moura@ece.cmu.edu}}

\maketitle

\begin{abstract}
We study distributed optimization problems when~$N$ nodes minimize the sum of their individual costs subject to a common vector variable. The costs are convex, have Lipschitz continuous gradient (with constant~$L$), and bounded gradient. We propose two fast distributed gradient algorithms based on the centralized Nesterov gradient algorithm and establish their convergence rates in terms of the per-node communications~$\mathcal{K}$ and the per-node gradient evaluations~$k$.
%, and the network topology defined by the $N \times N$ network-wide doubly stochastic weight matrix.
 Our first method,
 Distributed Nesterov Gradient, achieves
rates~$O\left( {\log \mathcal{K}}/{\mathcal{K}}\right)$ and~$O\left({\log k}/{k}\right)$.
%  when the nodes lack knowledge of the global parameters~$L$ and the spectral gap $1-\mu(W)$. When the nodes know $L$ and $\mu(W)$, the rates are~$O\left(\frac{1}{(1-\mu(W))^{1+\xi}} \frac{\log \mathcal{K}}{\mathcal{K}}\right)$ and~$O\left( \frac{1}{(1-\mu(W))^{1+\xi}} \frac{\log k}{k}\right)$, for the optimized step size.
 Our second method, Distributed Nesterov gradient with Consensus iterations, assumes at all nodes knowledge of~$L$ and
 $\mu(W)$ -- the second largest singular value of the $N \times N$ doubly stochastic weight matrix~$W$. It achieves
rates~$O\left( {1}/{\mathcal{K}^{2-\xi}}\right)$ and~$O\left( {1}/{k^2}\right)$ ($\xi>0$ arbitrarily small).
 Further, we give with both methods explicit dependence of the convergence constants on $N$ and $W$. %Furthermore, focusing on the dependence on~$k$ and~$\mathcal{K}$ only (with a fixed~$N$),
% we show that, while maintaining computationally simple iterations, both our methods achieve strictly faster rates than an
   %While involving only computationally simple iterations,
  %Our methods have strictly faster rates than existing distributed (sub)gradient methods, which have rates at most~$\Omega(1/\mathcal{K}^{2/3})$ and~$\Omega(1/k^{2/3})$.
 %Further, we extend
% the $O\left( \frac{1}{(1-\mu(W))^2}\frac{1}{{\mathcal K}^{2-\xi}}\right)$ and $O\left( \frac{1}{k^2}\right)$ convergence rate results of
% a (projected) distributed Nesterov gradient with consensus iterations
% to constrained problems with a compact constraint
% set known by all nodes, when
% the $f_i$'s have Lipschitz continuous gradient.
 Simulation examples
  %with the logistic and Huber losses demonstrate that our algorithms outperform existing distributed algorithms.
illustrate our findings.

\end{abstract}


% \fi
\IEEEpeerreviewmaketitle

\section{Introduction}
\label{section-introduction}
\IEEEPARstart{D}{istributed} computation and optimization have been
studied for a long time, e.g.,~\cite{Tsitsiklis_Distr_Opt,tsitsiklisThesis84}, and have received renewed interest, motivated by applications in sensor~\cite{Rabbat}, multi-robot~\cite{johansson-negotiation}, or cognitive networks~\cite{bazerque_lasso}, as well as in distributed control~\cite{NecoaraMPC} and learning~\cite{BoydADMoM}.
%
%
%
%Distributed optimizat methods, focus on
% the minimization of a smooth function $f$
%
%The existing literature Distributed gradient type methods have been developed~
This paper focuses on the problem where~$N$ nodes (sensors, processors, agents)
 minimize a sum of convex functions $f(x):=\sum_{i=1}^N f_i(x)$ subject to a common variable $x \in {\mathbb R}^d$. Each function $f_i: {\mathbb R}^d \rightarrow {\mathbb R}$ is convex and known only to node $i$. The underlying network is generic and connected.

To solve this and related problems, the literature proposes
several distributed gradient like methods, including:~\cite{nedic_T-AC} (see also~\cite{nedic_novo,asu-random,Matei});~\cite{SayedConf} (see also \cite{SayedOptim});~\cite{duchi} (see also~\cite{Rabbat,Rabbat-Consensus-Dual-Avg}); and~\cite{Sonia-Martinez}.
%The (sub)gradient type algorithms are attractive due to their applicability to generic problems within a certain class of $f_i$'s and simple, inexpensive low cost iterations that do not require sophisticated processors
% -- a feature relevant, e.g., with inexpensive sensor network motes.
%
%Most of the existing work analyzes the algorithms under a wide class of (possibly) non-differentiable, convex $f_i$'s, with bounded gradients -- for unconstrained minimization, or
%Lipschitz continuous $f_i$'s over the constraint set -- for constrained minimization. Under these
%function classes,
%
%
When the nodes lack global knowledge of the network parameters, reference~\cite{duchi} establishes, for the distributed dual averaging algorithm therein, rate~$O\left(\frac{1}{(1-\mu(W))}\frac{\log (N k)} {k^{1/2}}\right)$, where~$k$ is the number of communicated $d$-dimensional vectors per node, which also equals the number of iterations (gradient evaluations per node,) and $\mu(W)$ is
 the second largest singular value of the underlying $N \times N$ doubly stochastic
 weight matrix~$W$. Further, when $\mu(W)$ is known to the nodes, and after optimizing the step-size, \cite{duchi} shows the convergence rate to be~$O\left(\frac{1}{(1-\mu(W))^{1/2}}\frac{\log ( N k) }{k^{1/2}}\right)$.

%
%dual averaging algorithm therein achieves
%the optimality gap in the cost function~$O\left(\frac{1}{(1-\mu(W))^{1/2}}\frac{\log k}{\sqrt{k}}\right)$ at arbitrary node,
%where $k$ is the number of iterations -- the number of communicated vectors per node,
%and $\mu(W) \in [0,1]$ is the modulus of the second-largest-in-modulus eigenvalue of the network averaging matrix.
%

    \textbf{Setup}. The class of functions usually considered in the references above are more general than we consider here, namely, they assume that the $f_i$'s are (possibly) non-differentiable and convex, and:
    \begin{inparaenum}[1)]
    \item for unconstrained minimization, the $f_i$'s have bounded gradients, while
    \item for constrained minimization, they are Lipschitz continuous over the constraint set.
    \end{inparaenum}
    In contrast, we assume the class $\mathcal{F}$ of convex $f_i$'s that have Lipschitz continuous and bounded gradients.

    It is well established in centralized optimization,~\cite{Nesterov-Gradient}, that one expects faster convergence rates on classes of more structured functions; e.g., for convex, non-smooth functions, the best achievable rate for centralized (sub)gradient methods is~$O(1/\sqrt{k})$, while, for convex functions with Lipschitz continuous gradient, the best rate is $O(1/k^2)$, achieved, e.g., by the Nesterov gradient method~\cite{Nesterov-Gradient}. Here $k$ is the number of iterations, i.e., the number of gradient evaluations.
    %In this paper, we ask whether we can
   % The centralized Nesterov gradient method does not require bounded gradients -- the assumption that we impose for our
%    distributed methods.

    \textbf{Contributions}. Building from the centralized Nesterov gradient method, we develop for the class $\mathcal F$ two distributed gradient methods and prove their convergence rates, in terms of the number of per-node communications~$\mathcal{K}$, the per-node gradient evaluations~$k$, and the network topology. Our first method, the Distributed Nesterov Gradient (D--NG), uses
    one communication per~$k$ (it has $k=\mathcal K$) and achieves convergence
    rate~$O\left(\frac{1}{(1-\mu)^{p+\xi}} \left[ \,\frac{\log k}{k} + \frac{\sqrt{N} \log^{1/2} k}{k^{3/2}} + \frac{N}{k^2}\,\right] \right)$, where $p=3$ and $\xi>0$ is an arbitrarily small quantity, when the nodes have no global knowledge of the parameters underlying the optimization problem and the network: $L$ and~$G$ the $f_i$'s gradient's Lispchitz constant and the gradient bound, respectively, $\mu:=\mu(W)$ the second largest singular value of $W$, and $R$ a bound on the distance to a solution. When $L$ and $\mu$ are known by all, D--NG with optimized step-size achieves the same rate with $p$ reduced to~$1$.

    Our second method, Distributed Nesterov gradient with Consensus iterations (D--NC), assumes global knowledge on~$\mu$ and $L$ and achieves rates~$O\left(
    \frac{1}{\left[(1-\mu) \mathcal{K}^{1-\xi}\right]^2}+
    \frac{\sqrt{N}}{\left[(1-\mu) \mathcal{K}^{1-\xi}\right]^3}
    +
    \frac{N}{\left[(1-\mu) \mathcal{K}^{1-\xi}\right]^4}\right)$  and $O\left(\frac{1}{k^2} + \frac{\sqrt{N}}{k^3} + \frac{N}{k^4}\right)$. Further, we establish that, for the class~$\mathcal{F}$, both our methods (achieving at least~$O(\log k/k)$) are strictly better than the distributed (sub)gradient method~\cite{nedic_T-AC} and the distributed dual averaging  in~\cite{duchi},
    even when these algorithms are restricted to functions in~$\mathcal F$. We show analytically that~\cite{nedic_T-AC} cannot be better than~$\Omega\left( 1/k^{2/3}\right)$ and $\Omega\left( 1/{\mathcal K}^{2/3}\right)$ (see Subsection~\ref{subsection-comparisons} for details), and by simulation examples that~\cite{nedic_T-AC} and~\cite{duchi} perform similarly.

    \textbf{Distributed versus centralized Nesterov gradient methods}.
    The centralized Nesterov gradient method does not require bounded gradients -- an assumption that we make for our
     distributed methods. We prove here that if we drop the bounded gradients assumption,
     the convergence rates that we establish do not hold for either of our algorithms. (It may be possible to \emph{replace} the bounded gradients assumption with a weaker requirement.)
     In fact, the worst case convergence rates
     of D--NG and D--NC become arbitrarily slow. (See Subsection~\ref{subsection-discussion-bounded-gradients} for details.)
     %Thus, the convergence rates of the distributed methods D--NG and D--NC
     % hold under a strictly smaller class than with the corresponding centralized method.
     This important result illustrates a distinction between the allowed function classes by the
     centralized and distributed methods. The result is not specific
     to our accelerated methods; it can be shown that
     the standard distributed gradient method in~\cite{nedic_T-AC}
     is also arbitrarily slow when the assumption of bounded gradients is dropped~(while convexity
     and Lipschitz continuous gradient~hold)~\cite{arxivVersion}.

    % Such
%     Similar
%     We point out that the latter is not a drawback of our accelerated distributed methods;
%      without the bounded gradients Assumption, the same slow convergence can be shown
%       for the standard distributed gradient in~\cite{nedic_T-AC} also.
%

     %Finally, we remark that our candidate functions
     %that disobey bounded gradients assumption

     \textbf{Remark}. Since we make use here of the bounded gradients assumption,
      an interesting research direction is to look for a weaker requirement, e.g.,
      boundedness of all $x_i^\star \in \mathrm{arg\,min}_{x \in {\mathbb R}^d}f_i(x)$ ($\|x_i^\star\| \leq C <\infty$, $\forall x_i^\star$, $\forall i$.)
       In fact, with both D--NG and D--NC,
        we prove elsewhere that we can assume \emph{different}
     setups (corresponding to broad classes of functions) and still
      achieve the same convergence rates in terms of $k$ and $\mathcal K$. With D--NG,
     we can replace the bounded gradients assumption
     with the following: there exists $b,B>0$ such that,
      $\forall i$, $f_i(x) \geq b\|x\|$ whenever $\|x\|\geq B$.
      For a natural extension of D--NC, we can replace the unconstrained problems with
      Lipschitz continuous and bounded gradients assumed here
      by a constrained optimization problem (compact, convex constraint set $\mathcal X$)
      where the $f_i$'s have Lipschitz continuous gradient on
      a certain compact set that includes~$\mathcal X$.
%      Then, for a natural extension of D--NC, we can show $O(1/k^2)$ and $O(1/\mathcal{K}^{2-\xi})$
%     convergence rates.
Due to lack of space,
      these alternatives are pursued elsewhere.

\textbf{Remark}. We comment on references~\cite{AnnieChen}
and~\cite{AnnieChenThesis}~(see also Subsection~\ref{subsection-comparisons} and~\cite{arxivVersion}).
They develop accelerated proximal methods for time varying networks that resemble D--NC.
The methods in~\cite{AnnieChen} and~\cite{AnnieChenThesis}
use only one consensus algorithm per outer iteration~$k$, while we use two with D--NC.
%The total number of consensus steps (communication rounds) per~$k$ is similar.
%This makes the analysis of the disagreement across nodes more challenging compared to D--NC.
Adapting the results in~\cite{AnnieChen,AnnieChenThesis} to our framework, it can be shown
that the optimality gap bounds in~\cite{AnnieChen,AnnieChenThesis} expressed in terms of $N,1-\mu(W),$ and $\mathcal K$
have the same or worse (depending on the variant of their methods) dependence on $\mathcal K$
 and $\mu(W)$ than the one we show for D--NC, and a worse dependence on $N$. (See~Subsection~\ref{subsection-comparisons}
 and~\cite{arxivVersion}.)

%\textbf{Remark}.
In addition to distributed gradient methods, literature also proposes distributed augmented Lagrangian dual or ordinary dual methods~\cite{bazerque_lasso,cooperative-convex,Joao-Mota,Joao-Mota-2,Uday,Terelius,johansson-accelerated-grad-conf,johansson-accelerated-grad-journ}. These are based on the augmented Lagrangian (or ordinary) dual of the original problem. They in general have significantly more complex iterations than the gradient type methods that we consider in this paper, due to solving local optimization problems at each node, at each iteration, but may have a lower total communication cost. %Although~\cite{bazerque_lasso,cooperative-convex,Joao-Mota,Joao-Mota-2} do not present convergence rates for their distributed methods, a result from recent reference~\cite{Rate-ADMoM} on the centralized ADMoM may
%   be possible to extend to distributed variants of ADMoM~\cite{bazerque_lasso,bazerque_sensing,Joao-Mota-2},
%   so that they would achieve~$O(1/\mathcal{K})$; but the dependence on network type and size, e.g., $\mu(W)$, remains unknown.
%
% %~\cite{bazerque_lasso,cooperative-convex,Joao-Mota,Joao-Mota-2}.
    Reference~\cite{Joao-Mota} uses the Nesterov gradient method to propose an \emph{augmented Lagrangian dual} algorithm but does not analyze its convergence rate. In contrast, ours are \emph{primal gradient algorithms}, with no notion of Lagrangian dual variables, and we establish the convergence rates of our algorithms.
    References~\cite{johansson-accelerated-grad-conf,johansson-accelerated-grad-journ}
    study both the resource allocation and the problems that we consider (see~\eqref{eqn-opt-prob-original}).
    For~\eqref{eqn-opt-prob-original},~\cite{johansson-accelerated-grad-conf,johansson-accelerated-grad-journ}
    apply certain accelerated gradient methods on the~\emph{dual~problem}, in contrast with our primal gradient methods.
     Finally,~\cite{NecoaraMPC} uses the Nesterov gradient algorithm to propose a decomposition method based on a smoothing technique,
    for a problem formulation different than ours and on the~\emph{Lagrangian dual problem}.
    %, apply them to the dual variables
%   Hence, dual Lagrangian methods may have a smaller overall communication cost than gradient-like methods but have increased overall computational cost.
 % \end{inparaenum}

%
%
%
%
%
%
%

\textbf{Paper organization}. The next paragraph introduces notation. Section~\ref{section-problem-model} describes the network and optimization models that we assume. Section~\ref{section-algorithms} presents our algorithms, the distributed Nesterov gradient and the distributed Nesterov gradient with consensus iterations, D--NG and D--NC for short.
Section~\ref{section-inexact-oracle} explains the framework of the (centralized) inexact Nesterov gradient method; we use this framework to establish the convergence rate results for D--NG and D--NC. Sections~\ref{section-converg-rate-D-NG} and~\ref{section-converg-rate-D-NC} prove convergence rate results for the algorithms D--NG and D--NC, respectively.
%Section~\ref{section-projected-D-NC}
%presents the results on the projected D--NC algorithm.
Section~\ref{section-comparisons} compares our algorithms D--NG and D--NC with existing distributed gradient type methods, discusses the algorithms' implementation, and discusses the need for our Assumptions. Section~\ref{section-simulations} provides simulation examples.
Finally, we conclude in~Section~\ref{section-conclusion}. Proofs of certain lengthy arguments are relegated to Appendix.% proves several

\textbf{Notation}. %We deal with both real and complex scalars, vectors, and matrices, and the notation we define here refers to both, unless explicitly stated otherwise.
We index by a subscript $i$ a (possibly vector) quantity assigned to node $i$; e.g., $x_i(k)$ is node $i$'s estimate at iteration $k$. Further, we denote by: ${\mathbb R}^d$ the $d$-dimensional real coordinate space; $\mathbf{j}$ the imaginary unit ($\mathbf{j}^2 = -1$); $A_{lm}$ or $[A]_{lm}$ the entry in the $l$-th row and $m$-th column of a matrix $A$; $a^{(l)}$ the $l$-th entry of vector $a$; $(\cdot)^\top$ the transpose and $(\cdot)^H$ the conjugate transpose; $I$, $0$, $\mathbf{1}$, and $e_i$, respectively, the identity matrix, the zero matrix, the column vector with unit entries, and the $i$-th column of $I$;
%$J$ the $N \times N$ ideal consensus matrix $J:=(1/N)\mathbf{1}\mathbf{1}^\top$;
$\oplus$ and $\otimes $ the direct sum and Kronecker product of matrices, respectively; $\| \cdot \|_l$ the vector (respectively, matrix) $l$-norm of its vector (respectively, matrix) argument; $\|\cdot\|=\|\cdot\|_2$ the Euclidean (respectively, spectral) norm of its vector (respectively, matrix) argument ($\|\cdot\|$ also denotes the modulus of a scalar); $\lambda_i(\cdot)$ the $i$-th smallest \emph{in modulus} eigenvalue; $A \succeq 0$ means that a Hermitian matrix $A$ is positive semi-definite; $\lceil a \rceil$ the smallest integer
not smaller than a real scalar $a$; $\nabla \phi(x)$ and $\nabla^2 \phi(x)$ the gradient and Hessian at $x$ of a twice differentiable function $\phi: {\mathbb R}^d \rightarrow {\mathbb R}$, $d \geq 1$.
%$\zeta(s)=\sum_{t=1}^{\infty} \frac{1}{t^s}$ the Riemann zeta function; $h_k=\sum_{t=1}^k \frac{1}{t}$ the $k$-th harmonic number; and $C_h = 0.577215...$ the Euler--Mascheroni constant.
For two positive sequences $a_k$ and $b_k$, the following is the standard notation: $b_k = O(a_k)$ if $\limsup_{k \rightarrow \infty}\frac{b_k}{a_k}<\infty$; $b_k = \Omega (a_k)$ if $\liminf_{k \rightarrow \infty}\frac{b_k}{a_k}>0$; and $b_k=\Theta(a_k)$ if $b_k = O(a_k)$ and $b_k = \Omega(a_k)$.
\vspace{-0mm}
\section{Problem model}
\label{section-problem-model}
This section introduces the network and optimization models that we assume.
% and presents our
%distributed Nesterov based algorithm.
%\subsection{Problem model}
%\label{subsec:networkmodel}
%\subsection{Network model}
%\label{subsection-network-model}

\textbf{Network model}. We consider a (sparse) network $\mathcal{N}$ of~$N$ nodes (sensors, processors, agents,) each communicating only locally, i.e., with a subset of the remaining nodes. The communication pattern is captured by the graph $\mathcal{G} = (\mathcal{N} ,E),$ where $E \subset \mathcal{N} \times \mathcal{N}$ is the set of links. The graph~$\mathcal{G}$ is connected, undirected and simple (no~self/multiple~links.)

\textbf{Weight matrix.} We associate to the graph~$\mathcal{G}$
 a symmetric, doubly stochastic (rows and columns sum to one and all the entries are non-negative), $N \times N$ weight matrix $W$,
 with, for $i \neq j$, $W_{ij}>0$ if and only if, $\{i,j\} \in E,$
  and $W_{ii}=1-\sum_{j \neq i}W_{ij}.$ Denote by $\widetilde{W}=W-J,$
   where $J:=\frac{1}{N} {\mathbf 1} {\mathbf 1}^\top$ is the ideal consensus matrix. We let $\widetilde{W}=Q \widetilde{\Lambda} Q^\top$, where $\widetilde{\Lambda}$ is the diagonal matrix with $\widetilde{\Lambda}_{ii}=\lambda_i(\widetilde{W})$,
     and $Q=[q_1,...,q_N]$ is the matrix of the eigenvectors of~$\widetilde{W}$.
%  We also define:
%%
%\begin{equation}
%\label{eqn-tilde-W-def}
%\widetilde{W} = W-J,
%\end{equation}
%a matrix that represents the deviation of $W$ from the ideal consensus matrix $J=\frac{1}{N} \mathbf{1}\mathbf{1}^\top.$
%
%
With D--NC, we impose Assumption~\ref{assumption-network}~(a) below; with D--NG, we require both
Assumptions~\ref{assumption-network}~(a)~and~(b).
\begin{assumption}[Weight matrix]
\label{assumption-network}
We assume that (a) $\mu(W)<1$; and (b)~$W\succeq {\eta}\,I,$ where ${\eta}<1$ is an arbitrarily small positive quantity.
\end{assumption}
Note that Assumption~\ref{assumption-network}~(a) can be fulfilled only by a connected network.
%
% We require with both D--NG and D--NC that $\mu(W):=\|W-J\|=\|\widetilde{W}\|<1.$
%  In addition, with D--NG, we require:
%\begin{equation}
%\label{eqn-weight-cond}
%W\succeq {\eta}\,I,
%\end{equation}
%
%
%First condition in~\eqref{eqn-weight-cond} is standard; the second condition
%of $W$ being positive definite is what we additionally
%require with respect to existing work, e.g.,~\cite{nedic_T-AC,duchi}.
%
%
%
Assumption~\ref{assumption-network}~(a) is standard and is also needed with the existing
algorithms in~\cite{nedic_T-AC,duchi}. For a connected network,
nodes can assign the weights $W$ and fulfill Assumption~\ref{assumption-network}~(a), e.g., through the Metropolis weights~\cite{BoydFusion};
to set the Metropolis weights, each node needs to know its own degree and its neighbors' degrees.
 Assumption~\ref{assumption-network}~(b) required by D--NG is not common in the literature.
  We discuss the impact of Assumption~\ref{assumption-network}~(b) in Subsection~\ref{subsection-comparisons}.

\textbf{Distributed optimization model}.
%\label{subsection-optimization-model}
The nodes solve the unconstrained problem:
\vspace{-0mm}
\begin{equation}
\label{eqn-opt-prob-original}
%\begin{array}[+]{ll}
\mbox{minimize} \:\: \sum_{i=1}^N f_i(x) =:f(x).
%\end{array}
\end{equation}
The function $f_i: {\mathbb R}^d \rightarrow {\mathbb R}$ is known only to node $i$. We impose Assumptions~\ref{assumption-f-i-s} and~\ref{assumption-bdd-gradients}.
\begin{assumption}[Solvability; Lipschitz continuous gradient]
\label{assumption-f-i-s}
\begin{enumerate}[(a)]
\item There exists a solution $x^\star \in {\mathbb R}^d$ with $f(x^\star)=\inf_{x \in {\mathbb R}^d}f(x)=:f^\star$.
\item $\forall i$, $f_i$ is convex, differentiable, with Lipschitz continuous derivative with constant~$L\in [0,\infty)$:
$
\|\nabla f_i(x) - \nabla f_i(y)\| \leq L \|x-y\|,\:\:\:\forall x,y \in {\mathbb R}^d.
$
\end{enumerate}
\end{assumption}
%
%We denote by ${{x^\star}}$ a solution to~\eqref{eqn-opt-prob-original} and the optimal value $f^\star:=f({{x^\star}})$.
%
%
%In addition to Assumptions \ref{assumption-network} and \ref{assumption-f-i-s},
%we assume exclusively one of the following two assumptions (Assumption \ref{assumption-bdd-gradients} or Assumption \ref{assumption-linear-growth}.)
%
%
%
\begin{assumption}[Bounded gradients]
\label{assumption-bdd-gradients}
$\exists G \in [0,\infty)$ such that, $\forall i$, $\|\nabla f_i(x)\|\leq G$, $\forall x \in {\mathbb R}^d.$
\end{assumption}
%
%
%
%
%
%
%We will later additionally impose
%either bounded gradients of $f_i$'s or
%a certain linear growth condition (see ahead Assumptions \ref{assumption-bdd-gradients} and \ref{assumption-linear-growth}.)
%
%\begin{assumption}[Linear growth]
%\label{assumption-linear-growth}
%There exists a pair $(b,B)$, $b>0,B>0$, such that, for all $i$:
%\[
%f_i(x) \geq b \|x\|\,\,\mathrm{whenever}\,\,\|x\|\geq B.
%\]
%\end{assumption}
%Note that, under Assumption~\ref{assumption-linear-growth}, the function $f_i$ is coercive, i.e.,
%$f_i(x) \rightarrow \infty$ whenever $\|x\| \rightarrow \infty.$

Examples of $f_i$'s that satisfy Assumptions~\ref{assumption-f-i-s}--\ref{assumption-bdd-gradients} include the logistic and Huber losses (See Section~\ref{section-simulations}), or
the ``fair'' loss in robust statistics, $\phi: {\mathbb R} \mapsto {\mathbb R}$, $\phi(x) =
b_0^2 \left( {\frac{|x|}{b_0}} - \log \left( 1+\frac{|x|}{b_0}\right) \right)$, where $b_0$ is a positive parameter,
e.g.,~\cite{Blatt-Hero-Gauchman}.
Assumption~\ref{assumption-f-i-s} is precisely the assumption required by~\cite{Nesterov-Gradient} in the convergence
analysis of the (centralized) Nesterov gradient method.
With respect to the centralized Nesterov gradient method~\cite{Nesterov-Gradient},
we additionally require bounded gradients as given by Assumption~\ref{assumption-bdd-gradients}.
We explain the need for Assumption~\ref{assumption-bdd-gradients} in Subsection~\ref{subsection-discussion-bounded-gradients}.
\vspace{-0mm}
\section{Distributed Nesterov based algorithms}
\label{section-algorithms}
We now consider our two proposed algorithms.
Subsection~\ref{subsection-D-NG} presents algorithm D--NG, while
subsection~\ref{subsection-D-NC} presents algorithm D--NC.
\vspace{-0mm}
\subsection{Distributed Nesterov gradient algorithm~(D--NG)}
\label{subsection-D-NG}

%\textbf{Average consensus algorithm}. Our algorithm has as its ingredient the
%standard average consensus algorithm; we now briefly
%recall the average consensus algorithm.
%Given a network $\mathcal{G}=(\mathcal{N},E)$ that satisfies Assumption~\ref{assumption-network},
% the average consensus algorithm calculates
% the average of the initial values $z_i(0)$ at all nodes: $\frac{1}{N}\sum_{i=1}^N z_i(0)$.
% Denote by $z_i(t)$ the node $i$'s estimate at iteration $t$, and
% $z(t)=(z_1(t),...,z_N(t))^\top$. Then, the update in matrix form is:
% %
% %
% \begin{equation}
% \label{eqn-consensus-alg}
% z(t) = W z(t-1),
% \end{equation}
%%
%%
%where $W$ is the weight matrix in~\eqref{eqn-weight-matrix}.
%Consider
%$z(t)-\frac{1}{N}\sum_{i=1}^N z_i(0)\mathbf{1}$. It is well known that, after $t$ iterations:
%\begin{equation}
%\label{eqn-consensus-estimate}
%\left\| z(t)-\frac{1}{N}\sum_{i=1}^N z_i(0)\mathbf{1} \right\| \leq \left( \mu(W)\right)^t \,\left\| z(0)-\frac{1}{N}\sum_{i=1}^N z_i(0)\mathbf{1} \right\|.
%\end{equation}

%We now present the D--NG algorithm to solve~\eqref{eqn-opt-prob-original}.
Algorithm D--NG generates the sequence
$\left(x_i(k),y_i(k)\right)$, $k=0,1,2,...,$ at each node $i$, where $y_i(k)$ is an auxiliary variable.
D--NG is initialized by $x_i(0)=y_i(0) \in {\mathbb R}^d$, for all $i$. The update at node $i$ and $k=1,2,...$~is:
%
%
%\begin{eqnarray}
%\label{eqn-psi-grad}
%x(k) &=& y(k-1) - \alpha_{k-1} \nabla \Psi_{k-1} (y(k-1)) \\
%y(k) &=& x(k) + \beta_{k-1} \left( x(k) -x(k-1)\right),
%\end{eqnarray}
%with $\alpha_k = \frac{c}{k}$, and $\beta_{k} = \frac{k}{k+3}$. By
%rearranging the terms, the algorithm becomes:
%
\vspace{-0mm}
\begin{eqnarray}
\label{eqn-our-alg-scalar}
x_i(k) &=& \sum_{j \in O_i} W_{ij}\,y_j(k-1)- \alpha_{k-1} \nabla f_i (y_i(k-1)) \\
\label{eqn-our-alg-druga-scalar}
y_i(k) &=& x_i(k) + \beta_{k-1} \left( x_i(k) -x_i(k-1)\right).
\end{eqnarray}
\vspace{-0mm}
Here, $W_{ij}$ are the averaging weights (the entries of $W$), and
$O_i$ is the neighborhood set of node $i$ (including $i$). The step-size $\alpha_k$
 and the sequence $\beta_k$ are:
\begin{eqnarray}
\label{eqn-alpha-k}
\alpha_k = \frac{c}{k+1},\:\:\:\:\:\:c>0;\:\:\:\:\:\:
%\label{eqn-beta-k}
\beta_k = \frac{k}{k+3},\:\:\:\:\:\:k=0,1,...
\end{eqnarray}
%The choice and the meaning of the constants $c, $, and $w_0$ will become clearer further ahead.
With
 algorithm~\eqref{eqn-our-alg-scalar}--\eqref{eqn-our-alg-druga-scalar}, each node $i$, at each iteration $k$, performs
 the following: 1) broadcasts its
 variable $y_i(k-1)$ to all its neighbors $j \in O_i$; 2) receives $y_j(k-1)$
  from all its neighbors $j \in O_i$; 3) updates
  $x_i(k)$ by weight-averaging its own $y_i(k-1)$ and its neighbors variables $y_j(k-1)$, and
  performs a negative gradient step with respect to $f_i$; and 4) updates $y_i(k)$ via
   the inexpensive update in~\eqref{eqn-our-alg-druga-scalar}. To avoid notation explosion
   in the analysis further ahead, we assume throughout the paper, with both D--NG and D--NC, equal initial estimates $x_i(0)=y_i(0)=x_j(0)=y_j(0)$
    for all $i,j;$ e.g., nodes can set them to~zero.

We adopt the sequence $\beta_k$ as in the centralized fast gradient method by Nesterov~\cite{Nesterov-Gradient}; see also~\cite{TsengFastGradient,Vandenberghe}.
With the centralized Nesterov gradient, $\alpha_k=\alpha$ is constant along the iterations.
 However, under a constant step-size, algorithm~\eqref{eqn-our-alg-scalar}--\eqref{eqn-our-alg-druga-scalar} does not converge to the exact solution, but only to a solution neighborhood. More precisely, in general,
 $f(x_i(k))$ does not converge to $f^\star$ (See~\cite{cdc-submitted} for details.) We force $f(x_i(k))$
  to converge to $f^\star$ with~\eqref{eqn-our-alg-scalar}--\eqref{eqn-our-alg-druga-scalar} by adopting a diminishing step-size $\alpha_k$,
  as in~\eqref{eqn-alpha-k}. The constant $c>0$ in~\eqref{eqn-alpha-k} can be arbitrary (See also ahead Theorem~\ref{theorem-basic-convergence-result}.)

%
%\textbf{Step size and initialization}. We assume that the step-size constant $c$ satisfies:
%\begin{equation}
%\label{eqn-cond-step-size}
%c \leq \frac{1}{2 L}.
%\end{equation}
%Requirement~\eqref{eqn-cond-step-size}
%is not needed for the algorithm to converge ar rate $O(\log k/k)$. (See a comment in****.)
% However,~\eqref{eqn-cond-step-size} is needed to guarantee
% the hidden constant (See $\mathcal{C}(N)$ in Theorem~\ref{theorem-basic-convergence-result}.)
%is similar to the requirements with standard (centralized) gradient methods, namely,
% that the step size be less than the inverse of the gradient's Lipschitz constant.
% The need for~\eqref{eqn-cond-step-size} will become completely clear in the next Section. We
% comment on the practical implementation of our algorithm \eqref{eqn-our-alg}
% in Subsection \ref{subsection-theorem-bdd-grad}.

\textbf{Vector form}.
Let $x(k)=(x_1(k)^\top, x_2(k)^\top,...,x_N(k)^\top)^\top$, $y(k)=(y_1(k)^\top, y_2(k)^\top,...,y_N(k)^\top)^\top$,
and introduce $F: {\mathbb R}^{N d} \rightarrow {\mathbb R}^N$ as:
$
F(x) = F(x_1,x_2,...,x_N) = (f_1(x_1),f_2(x_2),...,f_N(x_N))^\top.
$
Then, given initialization $x(0)=y(0),$ D--NG in vector form is:
\vspace{-0mm}
\begin{eqnarray}
\label{eqn-our-alg}
x(k) &=& (W \otimes I) y(k-1) - \alpha_{k-1} \nabla F (y(k-1)) \\
\label{eqn-our-alg-druga}
y(k) \hspace{-2mm}&=& \hspace{-2mm} x(k) + \beta_{k-1} \left( x(k) -x(k-1)\right),\, k=1,2,...,%\nonumber
\end{eqnarray}
\vspace{-0mm}
where the identity matrix is of size $d$ -- the dimension of the optimization variable in~\eqref{eqn-opt-prob-original}.
%For simplicity and to avoid notation explosion, we assume throughout this paper that
% $x_i(0)=x_j(0)$, for all $i,j$; e.g., nodes can set $x_i(0)=x_j(0)=0.$

\subsection{Algorithm D--NC}
\label{subsection-D-NC}
%
%
%
%We now present a distributed Nesterov gradient method
%with the convergence rate $O(\frac{1}{\mathcal{K}^{2-\xi}})$,
%where $\epsilon$ is an arbitrarily small positive number, and
%$\mathcal{K}$ is the total number of communication rounds.
%We let Assumptions~\ref{assumption-network}--\ref{assumption-bdd-gradients} hold.

%\textbf{The distributed Nesterov based algorithm}.
Algorithm D--NC uses a~\emph{constant step-size} $\alpha \leq 1/(2L)$ and operates in two time scales.
In the outer (slow time scale) iterations~$k$, each node $i$
updates its solution estimate $x_i(k)$, and updates an auxiliary variable $y_i(k)$ (as with the D--NG);
%\footnote{To avoid notation explosion, we use the same letters to denote the iterates of D--NG and D--NC.}
 in the inner iterations~$s$, nodes
 perform two rounds of consensus with the number of inner iterations given in~\eqref{eqn-consensus} and~\eqref{eqn-consensus-2} below, respectively. D--NC is Summarized in Algorithm~1.
 %
 %
%We now detail D--NC. It starts by $x_i(0)=y_i(0) \in {\mathbb R}^d$. At iteration $k$,
% node $i$ has available~$x_i(k-1)$ and~$y_i(k-1)$.
% Node $i$ calculates $\nabla f_i(y_i(k-1))$. Then, it computes:
  %
%  nodes run $\tau_{k-1}$ iterations of~\eqref{eqn-consensus} with the initialization $z_i(s=0,k-1) = \nabla f_i(y_i(k-1))$,
%  so that each node obtains $g_i(k-1)$ -- an inexact version of $\frac{1}{N}\sum_{i=1}^N \nabla f_i(y_i(k-1))$.
%   Subsequently,
%node $i$ performs the update:
   %
   %
   %\begin{equation}
%   \label{en-update-x-temp}
%   x_i^{(a)}(k) = y_i(k-1) - \alpha \nabla f_i(y_i(k-1)).
%   \end{equation}
   %
   %
   %
%   In the next step, nodes jointly run first average-consensus with the initialization $z_i(s=0,k) = x_i^{(a)}(k)$:
%%
% \begin{equation}
% \label{eqn-consensus}
% z_i(s,k) = \sum_{j \in O_i} W_{ij} z_j(s-1,k),\:\:s=1,2,...,\tau_{k-1},
% \end{equation}
% and obtain $x_i(k):=z_i(s=\tau_{k-1},k)$.
%    %
%    Subsequently, node $i$ calculates $y_i^{(a)}(k)$ via:
%    %
%    %
%   \begin{equation}
%   \label{en-update-y-temp}
%   y_i^{(a)}(k) = x_i(k) + \beta_{k-1} \left( x_i(k) - x_i(k-1)\right),
%   \end{equation}
%   %
%   where $\beta_k$ is in~\eqref{eqn-alpha-k}.
%   %
%   %
%   %
%   %
%   Finally, the nodes run a second average-consensus like in~\eqref{eqn-consensus} with the initialization $z_i(s=0,k)=y_i^{(a)}(k)$ and $\tau_{k-1}$ iterations, so that node
%   $i$ obtains $y_i(k)=z_i(s=\tau_{k-1},k)$.
%   The algorithm D--NC is summarized in Algorithm~1.

%
%
\begin{algorithm}
\label{algorithm-D-NC}
\caption{  Algorithm D--NC }
\begin{algorithmic}[1]
{\small{
    \STATE Initialization: Node $i$ sets: $x_i(0)=y_i(0) \in {\mathbb R}^d$; and $k=1.$
    %\STATE
        %\STATE Node $i$ calculates~$\nabla f_i(y_i(k-1))$.
        %\STATE (First consensus) Nodes run~\eqref{eqn-consensus}
%        for $s=1,2,...\tau_{k-1}$, with $z_i(s=0,k-1)=\nabla f_i(y_i(k-1))$, so
%        that node $i$ obtains~$g_i(k-1):=z_i(s=\tau_{k-1},k-1)$.
        \STATE Node $i$ calculates:
        $
        x_i^{(a)}(k) = y_i(k-1) - \alpha \nabla f_i(y_i(k-1)).
        $
        \STATE (First consensus) Nodes run average consensus initialized by~$x_i^{(c)}(s=0,k)=x_i^{(a)}(k)$:
        \begin{eqnarray}
        x_i^{(c)}(s,k) &=& \sum_{j \in O_i} W_{ij} x_j^{(c)}(s-1,k),\:\:s=1,2,...,\tau_x(k) \nonumber \\%,\:\:\:\:\:\:\:\:\:
        %\label{eqn-consensus}
        \label{eqn-consensus}
        \tau_x(k) &=& \left\lceil \frac{2 \log k }{-\log \mu(W)} \right\rceil,
        \end{eqnarray}
        and set~$x_i(k):=x_i^{(c)}(s=\tau_x(k),k)$.
        \STATE Node $i$ calculates $
        y_i^{(a)}(k) = x_i(k) + \beta_{k-1} \left( x_i(k) - x_i(k-1)\right).
        $
        \STATE (Second consensus) Nodes run average consensus initialized by~$y_i^{(c)}(s=0,k)=y_i^{(a)}(k)$:
        \begin{eqnarray}
        y_i^{(c)}(s,k) &=& \sum_{j \in O_i} W_{ij} y_j^{(c)}(s-1,k),\:\:s=1,2,...,\tau_y(k) \nonumber \\%,\:\:\:\:\:\:\:
        %\label{eqn-consensus}
        \label{eqn-consensus-2}
        \tau_y(k) &=& \left\lceil \frac{\log 3}{-\log \mu(W)} + \frac{2 \log k }{-\log \mu(W)} \right\rceil,
        \end{eqnarray}
        and set~$y_i(k):=y_i^{(c)}(s=\tau_y(k),k)$.
        \STATE Set $k \mapsto k+1$ and go to step 2.}}
\end{algorithmic}
\vspace{0mm}
%\label{algorithm_SNG}
\end{algorithm}
The number of inner consensus iterations in~\eqref{eqn-consensus} increases as $\log k$ and
depends on the underlying network through~$\mu(W)$. Note an important difference between D--NC and D--NG. D--NC uses explicitly a number of consensus steps at each~$k$. In contrast, D--NG does not explicitly
use multi-step consensus at each~$k$; consensus occurs implicitly, similarly to~\cite{nedic_T-AC,duchi}.

%\textbf{Step-size and the number of inner (consensus) iterations}.
%We require the step-size $\alpha$ to satisfy $\alpha \leq 1/(2L)$.
% With D--NC, this condition
% is critical for convergence.% (See also Subsection~\ref{subsection-logistic-loss} and Figure~1, left.)
%%
%%%
%%%
%%\begin{equation}
%%\label{eqn-step-size-cond-Consensus-alg}
%%\alpha \leq \frac{1}{2 L}.
%%\end{equation}
%%
%%
%We let the number of the inner (consensus) iterations $\tau_{k-1}$
% at the outer iteration $k$ to equal:
 %
 %

%
%
% As with D--NG,
% we assume~$x_i(0)=y_i(0)=x_j(0)=y_j(0)$, for all~$i,j.$

\textbf{Vector form}. %We write D--NC in vector form.
Using the same compact notation for $x(k)$, $y(k)$, and $\nabla F(y(k))$
 as with D--NG, D--NC in vector form is:
%Denote, as before,
%{
%\small{
%\[
%x(k)=(x_1(k),...,x_N(k))^\top,\:\:\:\: y(k)=(y_1(k),...,y_N(k))^\top,\:\: \:\:  \nabla F(y(k)) = (\nabla f_1(y_1(k)),...,\nabla f_N(y_N(k)))^\top.
%\]
%}
%}
%Then:
%
%
\vspace{-0mm}
\begin{eqnarray}
\label{eqn-our-alg-cons}
x(k) &=& (W \otimes I)^{\tau_x(k)} \left[\,\, y(k-1) - \alpha\, \nabla F(y(k-1)) \,\,\right] \\
\label{eqn-our-alg-cons-druga}
y(k) \hspace{-3.5mm}&=& \hspace{-3.5mm}(W \otimes I) ^{\tau_y(k)} \left[ x(k) + \beta_{k-1} (x(k)-x(k-1))\right].
\vspace{0mm}
\end{eqnarray}
The power $(W \otimes I)^{\tau_x(k)}$ in~\eqref{eqn-our-alg-cons}
corresponds to the first consensus in~\eqref{eqn-consensus}, % on the %$\nabla f_i(y_i(k-1))$'s; the left matrix power $(W \otimes I)^{\tau_{k-1}} $
% in~\eqref{eqn-our-alg-cons} corresponds to the second consensus
%$x_i^{(a)}(k)$'s;
 and the power~$(W \otimes I)^{\tau_y(k)}$ in~\eqref{eqn-our-alg-cons-druga}
  corresponds to the second consensus in~\eqref{eqn-consensus-2}.
   The connection between D--NC and the (centralized) Nesterov gradient
    method becomes clearer in Subsection~\ref{subsection-algs-in-framework}.
     The matrix powers~\eqref{eqn-our-alg-cons}--\eqref{eqn-our-alg-cons-druga}
     are implemented in a distributed way through multiple iterative steps --
     they require respectively $\tau_x(k)$ and $\tau_y(k)$ iterative (distributed) consensus steps.
      This is clear from the representation in~Algorithm~1.

\vspace{-0mm}
\section{Intermediate results: Inexact Nesterov gradient method}
\label{section-inexact-oracle}
%
%\textbf{Motivation}.
We will analyze the convergence rates of D--NG and D--NC
by considering the evolution of the global averages $\overline{x}(k):=\frac{1}{N}\sum_{i=1}^N x_i(k)$
 and $\overline{y}(k):=\frac{1}{N}\sum_{i=1}^N y_i(k)$.
 We will show that, with both distributed methods,
 the evolution of $\overline{x}(k)$ and $\overline{y}(k)$
 can be studied through the framework of
 the inexact (centralized) Nesterov gradient method,
 essentially like the one in~\cite{inexact-oracle}.
 Subsection~\ref{subsection-inexact-oracle} introduces this framework and gives the
relation for the progress in one iteration. Subsection~\ref{subsection-algs-in-framework}
then demonstrates that we can cast our algorithms D--NG and D--NC
in this framework.

\subsection{Inexact Nesterov gradient method}
\label{subsection-inexact-oracle}
%
%
%
%Throughout this subsection, we consider an arbitrary convex function ${f}: {\mathbb R}^d \rightarrow \mathbb R$ with Lipschitz continuous gradient with constant $L_{{f}}$.
We next introduce the definition of a (pointwise) inexact first order oracle.
% ofSubsequently, we explain how the inexact first order oracle framework is useful
%for analyzing the convergence of~D--NG.
%Then, we relate
%this definition with the evolution of $(\overline{x}(k), \overline{y}(k))$ in \eqref{eqn-our-alg-CM}.

\begin{definition}[Pointwise inexact first order oracle]
\label{def-inexact-oracle}
Consider a function ${f}: {\mathbb R}^d \rightarrow \mathbb R$ that is convex and has Lipschitz continuous gradient with constant $L_{{f}}$. We say that
a pair $\left(\widehat{f}_y, \widehat{g}_y\right) \in  {\mathbb R} \times {\mathbb R}^d$ is a $\left( L_y,\delta_y\right)$ inexact oracle
of ${f}$ at point $y$ if:
\begin{eqnarray}
\label{eqn-inexact-oracle-prop-1}
\widehat{f}_y &+& \widehat{g}_y^\top \left( x - y\right) \leq {f}(x)
\leq \widehat{f}_y \\
&+& \widehat{g}_y^\top \left( x - y\right)
+ \frac{L_y}{2}\|x-y\|^2 + \delta_y,\: \forall x \in {\mathbb R}^d. \nonumber
\end{eqnarray}
\end{definition}

For any $y \in {\mathbb R}^d$, the pair $({f}(y), \nabla {f}(y))$ satisfies Definition~\ref{def-inexact-oracle} with $\left( L_y = L_{{f}},\delta_y = 0\right)$. If $\left(\widehat{f}_y, \widehat{g}_y\right) $ is a $\left( L_y,\delta_y\right)$ inexact oracle at $y$, then it is also
 a $\left( L_y^\prime,\delta_y\right)$ inexact oracle at $y$, with $L_y^\prime \geq L_y.$

\textbf{Remark}. The prefix pointwise
in Definition~\ref{def-inexact-oracle} emphasizes that we are concerned with finding $\left(\widehat{f}_y, \widehat{g}_y\right)$
 that satisfy~\eqref{eqn-inexact-oracle-prop-1} with
 $(L_y,\delta_y)$ \emph{at a fixed point $y$}. This differs from the conventional definition (Definition~1) in~\cite{inexact-oracle}.
Throughout, we always refer to the inexact oracle in the sense
of Definition~\ref{def-inexact-oracle} here and drop the prefix pointwise.

\textbf{Inexact Nesterov gradient method}. Lemma~\ref{lemma-progress-one-iteration} gives the progress in one iteration of the inexact (centralized) Nesterov gradient method for the unconstrained minimization of ${f}$. Consider a point $(\overline{x}(k-1),\overline{y}(k-1)) \in {\mathbb R}^d \times {\mathbb R}^d$, for some
fixed $k=1,2,...$
%, and define
%$\overline{v}(k-1)$ by:
%\begin{equation}
%\label{eqn-v-k-def}
%\overline{v}(k-1) = \frac{\overline{y}(k-1) - (1-\theta_{k-1})\overline{x}(k-1)}{\theta_{k-1}},
%\end{equation}
%where $\theta_k=\frac{2}{k+2}$.
Let $\left( \widehat{f}_{k-1}, \widehat{g}_{k-1}\right)$ be
a $(L_{k-1},\delta_{k-1})$ inexact oracle of the function ${f}$ at point $\overline{y}(k-1)$ and:
%\footnote{For convenience,
%we denote the iterates in~\eqref{eqn-lemma-updates} by $\overline{x}(k),\overline{y}(k)$ -- the same notation as
%we use for the global averages with D--NG and D--NC, as we later
%apply Lemma~\ref{lemma-progress-one-iteration} to these global averages.}
%
%
%
\begin{eqnarray}
\label{eqn-lemma-updates}
\overline{x}(k) &=& \overline{y}(k-1) - \frac{1}{L_{k-1}} \widehat{g}_{k-1}\\
%\\
%\label{eqn-lemma-updates-2}
\overline{y}(k) &=& \overline{x}(k) + \beta_{k-1} \left( \overline{x}(k) - \overline{x}(k-1) \right).
\nonumber
%\overline{v}(k) &=& \frac{\overline{y}(k) - (1-\theta_{k})\overline{x}(k)}{\theta_{k}} . \nonumber
\end{eqnarray}
%
%
%We have the following Lemma.
%
\begin{lemma}[Progress per iteration]
\label{lemma-progress-one-iteration}
Consider the update rule~\eqref{eqn-lemma-updates} for some $k=1,2,...$ Then:
\begin{eqnarray}
&(k+1)^2 & \left( {f}(\overline{x}(k)) - {f}({x^\bullet})\right)
+
2 L_{k-1} \|\overline{v}(k) - {x^\bullet}\|^2 \nonumber \\
&\leq&
(k^2-1)\, \left( {f}(\overline{x}(k-1)) - {f}({x^\bullet})\right) \nonumber \\
\label{eqn-consensus-2}
&+&\hspace{-1mm} 2 L_{k-1} \|\overline{v}(k-1) - {x^\bullet}\|^2 + (k+1)^2 \delta_{k-1},
\end{eqnarray}
for any ${x^\bullet} \in {\mathbb R}^d$,
where $\gamma_k=2/(k+2)$ and
$
\overline{v}(k) = \frac{\overline{y}(k) - (1-\gamma_{k})\overline{x}(k)}{\gamma_{k}}.
$%\end{equation}
\end{lemma}
Lemma~\ref{lemma-progress-one-iteration} is similar to~[\cite{inexact-oracle}, Theorem~5], although~\cite{inexact-oracle}
 considers a different accelerated Nesterov method.
 It is intuitive: the progress per iteration is the same
 as with the exact Nesterov gradient algorithm,
  except that it is
  deteriorated by the ``gradient direction inexactness'' ($(k+1)^2 \delta_{k-1}$).
    The proof follows the arguments of~\cite{inexact-oracle} and~\cite{Vandenberghe,Nesterov-Gradient,TsengFastGradient}
     and is in~\cite{arxivVersion}.
%
%
%
%
%
%
%
%
%\textbf{Nesterov gradient under inexact oracle}.
%
\vspace{-0mm}
\subsection{Algorithms D--NG and D--NC in the inexact oracle framework}
\label{subsection-algs-in-framework}
We now cast algorithms D--NG and D--NC in the inexact oracle framework.

\textbf{Algorithm D--NG}. Recall the global averages $\overline{x}(k):=\frac{1}{N}\sum_{i=1}^N x_i(k)$ and $\overline{y}(k):=\frac{1}{N}\sum_{i=1}^N y_i(k)$, and define:
\begin{eqnarray}
\label{eqn-f-k-g-k}
\widehat{f}_k \hspace{-2mm}&=&\hspace{-2mm} \sum_{i=1}^N \left\{f_i(y_i(k)) + \nabla f_i(y_i(k))^\top (\overline{y}(k)-y_i(k)) \right\}\\
\widehat{g}_k &=& \sum_{i=1}^N \nabla f_i(y_i(k)). \nonumber
\end{eqnarray}
Multiplying~\eqref{eqn-our-alg}--\eqref{eqn-our-alg-druga} from the left by $(1/N)(\mathbf{1}^\top \otimes I)$, using $(\mathbf{1}^\top \otimes I) (W \otimes I)=\mathbf{1}^\top \otimes I$,
letting $L_{k-1}^\prime:=\frac{N}{\alpha_{k-1}}$, and using~$\widehat{g}_k$ in~\eqref{eqn-f-k-g-k}, we obtain
 that $\overline{x}(k)$, $\overline{y}(k)$ evolve according to:
% , we get that
% the sequence $(\overline{x}(k), \overline{y}(k))$ evolves according to:
 %
 %
\begin{eqnarray}
\label{eqn-our-alg-CM}
\overline{x}(k) &=& \overline{y}(k-1) - \frac{1}{L_{k-1}^\prime} \widehat{g}_{k-1} \\
%\label{eqn-our-alg-druga-CM}
\overline{y}(k) &=& \overline{x}(k) + \beta_{k-1} \left( \overline{x}(k) -\overline{x}(k-1)\right), \nonumber
\end{eqnarray}
%
%
%for $k=1,2,...$, and~$\overline{x}(0)=\overline{y}(0)$.
%Thus, by evolving $\overline{x}(k), \overline{y}(k)$ according to
%an inexact centralized Nesterov gradient method with step-size $\frac{\alpha_{k-1}}{N}$, D--NG minimizes $f:=\sum_{i=1}^N f_i$, where the exact gradient $ \nabla f(\overline{y}(k-1)) = \sum_{i=1}^N \nabla f_i (\overline{y}(k-1))$ is replaced
%by $\sum_{i=1}^N \nabla f_i (y_i(k-1)).$
%Let:
%We will show that D--NG
% is actually a Nesterov gradient under inexact oracle,
% and we show that the amount of inexactness is controlled by
% the differences $\|y_i(k)-\overline{y}(k)\|$.
%
%
%
%Then, the algorithm~\eqref{eqn-our-alg-CM} is rewritten as~\eqref{eqn-lemma-updates}
% with $L_{k-1}$ replaced by~$L_{k-1}^\pri:=\frac{N}{\alpha_{k-1}}$.
 %: %rewritten as:
%%
%\begin{eqnarray}
%\label{eqn-our-alg-CM-2}
%\overline{x}(k) = \overline{y}(k-1) - \frac{1}{L^\prime_{k-1}} \widehat{g}_{k-1},\:\:\:\:\:\:\:\:\:\:% \\
%%\label{eqn-our-alg-druga-CM-2}
%\overline{y}(k) = \overline{x}(k) + \beta_{k-1} \left( \overline{x}(k) -\overline{x}(k-1)\right),%\:\:\:k=1,2,... %\nonumber
%\end{eqnarray}
%%
%
%
%
%for $k=1,2,...$
The following Lemma shows how we can analyze convergence of~\eqref{eqn-our-alg-CM} in the inexact oracle framework.
 Define~$\widetilde{y}_i(k):= y_i(k) - \overline{y}(k)$
and $\widetilde{y}(k):=(\widetilde{y}_1(k)^\top,...,\widetilde{y}_N(k))^\top.$
 Define analogously $\widetilde{x}_i(k)$ and $\widetilde{x}(k)$.
 We refer to $\widetilde{x}(k)$ and $\widetilde{y}(k)$ as the disagreement vectors,
 as they indicate how mutually apart the estimates of different nodes are.% different nodes' estimates.
\begin{lemma}
\label{lemma-relate-inexact-oracle-with-our-alg}
Let Assumption~\ref{assumption-f-i-s} hold.
Then,~$(\widehat{f}_k, \widehat{g}_k)$ in~\eqref{eqn-f-k-g-k} is a~$(L_k,\delta_k)$ inexact oracle
of~$f=\sum_{i=1}^N f_i$ at point~$\overline{y}(k)$ with constants $L_k = 2 N L$
and~$\delta_k=L \|\widetilde{y}(k)\|^2.$
\end{lemma}
%
%
%\textbf{Remark}. If $c \leq 1/(2L)$ such that $\alpha_{k-1} \leq \frac{1}{2 L k}$,
%then $L_{k-1}^\prime =\frac{N}{\alpha_{k-1}}=2 N L k \geq L_{k-1} = 2 NL$,
%and Lemma~\ref{lemma-progress-one-iteration}
% applies to $\overline{x}(k),\overline{y}(k)$ with D--NG.
% %, by setting
%% $\phi \equiv f$, $\widehat{x}(k)\equiv\overline{x}(k)$,
%% $\widehat{y}(k) \equiv \overline{y}(k)$, and
%$x^\bullet \equiv {{x^\star}}$.
%  Hence, we can establish the optimality gap $f(\overline{x}(k))-f^\star$
%   by using Lemma~\ref{lemma-progress-one-iteration}, setting $x^\bullet \equiv x^\star$, and,
%   additionally, by finding $\delta_k$, i.e., $\|\overline{y}(k)-y_i(k)\|$.
%    This is detailed in Section~\ref{section-converg-rate-D-NG}.
%
%Lemma~\ref{lemma-relate-inexact-oracle-with-our-alg} says that \eqref{eqn-our-alg-CM} is
% the Nesterov gradient algorithm to minimize $f$, with inexact oracle $(\widehat{f}_k,\widehat{g}_k)$ in \eqref{eqn-f-k-g-k}.
%Note that, as the (centralized) Nesterov gradient requires $\alpha_k \leq $
%
%
%
%
%
Lemma~\ref{lemma-relate-inexact-oracle-with-our-alg} implies that, if~$L_{k-1}^\prime = \frac{N k}{c} \geq 2 N L$, i.e.,
 if~$c \leq \frac{k}{2 L}$, then
 the progress per iteration in Lemma~\ref{lemma-progress-one-iteration}
 holds for~\eqref{eqn-our-alg-CM} with~$\delta_{k-1}:=L\|\widetilde{y}(k-1)\|^2$.
 If $c\leq 1/(2 L)$, Lemma~\ref{lemma-progress-one-iteration}
 applies for \emph{all iterations}~$k=1,2,...$; otherwise,
 it holds for all~$k\geq 2 c L$.
\begin{IEEEproof}[Proof of Lemma~\ref{lemma-relate-inexact-oracle-with-our-alg}]
For notation simplicity, we re-write $y(k)$ and $\overline{y}(k)$ as $y$ and $\overline{y}$, and $\widehat{f}_k, \widehat{g}_k, L_k, \delta_k$ as
$\widehat{f}_y, \widehat{g}_y, L_y, \delta_y$. In view of Definition~\ref{def-inexact-oracle}, we need to show inequalities \eqref{eqn-inexact-oracle-prop-1}. We first show the left one.
By convexity of $f_i(\cdot)$:
$
f_i(x) \geq f_i(y_i) + \nabla f_i(y_i)^\top (x-y_i), \:\:\: \forall x;
$
summing over $i=1,...,N$, using $f(x)=\sum_{i=1}^N f_i(x)$, and expressing $x-y_i=x-\overline{y}+\overline{y}-y_i$:
\begin{eqnarray*}
f(x) &\geq& \sum_{i=1}^N \left( f_i(y_i)+ \nabla f_i(y_i)^\top (\overline{y}-y_i)\right)
\\
&+& \left( \sum_{i=1}^N \nabla f_i(y_i)\right)^\top (x - \overline{y})
= \widehat{f}_y + \widehat{g}_y^\top (x - \overline{y}).
\end{eqnarray*}
We now prove the right inequality in~\eqref{eqn-inexact-oracle-prop-1}. As $f_i(\cdot)$
 is convex and has Lipschitz continuous derivative with constant $L$, we have:
$
 f_i(x) \leq f_i(y_i) + \nabla f_i(y_i)^\top (x-y_i) + \frac{L}{2}\|x-y_i\|^2,
$
 which, after summation over $i=1,...,N$, expressing
 $x-y_i=(x-\overline{y})+(\overline{y}-y_i)$, and using the inequality
  $\|x-y_i\|^2  = \|(x-\overline{y})+(\overline{y}-y_i)\|^2 \leq 2\|x-\overline{y}\|^2+ 2\|\overline{y}-y_i\|^2$, gives:
 {\allowdisplaybreaks{
\begin{eqnarray*}
 f(x) &\leq & \sum_{i=1}^N \left( f_i(y_i) + \nabla f_i(y_i)^\top (\overline{y}-y_i) \right) \\
  &+& \left( \sum_{i=1}^N \nabla f_i(y_i)\right)^\top (x - \overline{y}) \\
  &+&
  N L \|x-\overline{y}\|^2 + L \sum_{i=1}^N \|\overline{y}-y_i\|^2 \\
  &=& \widehat{f}_y + \widehat{g}_y^\top (x-\overline{y}) + \frac{2 N L}{2} \|x-\overline{y}\|^2 + \delta_y,
\end{eqnarray*}}}
and so~$(\widehat{f}_y,\widehat{g}_y)$ satisfy the right inequality in~\eqref{eqn-inexact-oracle-prop-1} with
$L_y=2 N L$ and $\delta_y=L \sum_{i=1}^N \|\overline{y}-y_i\|^2.$% which completes the proof of the Lemma.
\end{IEEEproof}

\textbf{Algorithm D--NC}. Consider algorithm D--NC in~\eqref{eqn-our-alg-cons}--\eqref{eqn-our-alg-cons-druga}.
To avoid notational clutter, use the same notation as with D--NG for the global averages:
 $\overline{x}(k):=\frac{1}{N}\sum_{i=1}^N x_i(k)$,
 and $\overline{y}(k):=\frac{1}{N}\sum_{i=1}^N y_i(k)$,
re-define~$\widehat{f}_k,\widehat{g}_k$ for D--NC as in~\eqref{eqn-f-k-g-k},
and let $L_{k-1}^\prime:=\frac{N}{\alpha}$.
 Multiplying~\eqref{eqn-our-alg-cons}--\eqref{eqn-our-alg-cons-druga} from the left by $(1/N)\mathbf{1}^\top \otimes I$, and using
  $(\mathbf{1}^\top \otimes I) (W \otimes I)=\mathbf{1}^\top \otimes I$, we get that
  $\overline{x}(k),\overline{y}(k)$ satisfy~\eqref{eqn-our-alg-CM}.
  As~$\alpha \leq 1/(2L)$, we have
$L_{k-1}^\prime \geq 2NL$, and so, by Lemma~\ref{lemma-relate-inexact-oracle-with-our-alg},
the progress per iteration in Lemma~\ref{lemma-progress-one-iteration}
 applies to~$\overline{x}(k),\overline{y}(k)$ of D--NC for all~$k$, with~$\delta_{k-1}=L\|\widetilde{y}(k-1)\|^2$.

In summary, the analysis of convergence rates of both D--NG and D--NC boils down to finding the disagreements~$\|\widetilde{y}(k)\|$ and
   then applying Lemma~\ref{lemma-progress-one-iteration}.
   %As we will see, their disagreements are different. %This is considered in Section~\ref{section-converg-rate-D-NC}.
   % Thus, $\overline{x}(k)$

%
%
%
%
%
%

\vspace{-0mm}
%\subsection{Intermediate results}
%Our convergence analysis of both D--NG and D--NC
%consists of two main parts: 1)
%analyzing the convergence of $(\overline{x}(k),\overline{y}(k))$ from the perspective of
%the centralized Nesterov gradient with inexact oracle;
%and 2) bounding the differences $\|y_i(k)-\overline{y}(k)\|$.
%For the former, with both D--NG and D--NC, we apply Lemma~\ref{lemma-progress-one-iteration} by making
%specific choices for the iterates $\widehat{x}(k) \equiv \overline{x}(k)$ and $\widehat{y}(k)
%\equiv \overline{y}(k)$, the function $\phi \equiv f$, the point $x^\bullet={{x^\star}}$,
%and a certain inexact oracle parameters $(L_k,\delta_k)$. For the disagreement estimate,
%each algorithm requires different analysis.
%
%Subsection \ref{subsection-progress-one-iter-thm} deals with the former problem.
%,
%while Section \ref{section-distance-to-cons} deals with the latter problem. To derive the convergence results of our algorithm D--NG,
%(Sections~\ref{section-converg-rate-bdd-grad} and~\ref{section-linear-growth}),
% we will use the results of Subsection~\ref{subsection-progress-one-iter-thm} and Section~\ref{section-distance-to-cons} repeatedly.
%
%
%
%
%
%
%
%
\section{Algorithm D--NG: Convergence analysis}
\label{section-converg-rate-D-NG}
This Section studies the convergence of D--NG.
Subsection~\ref{subsection-D-NG-disagreement}
bounds the disagreements~$\|\widetilde{x}(k)\|$ and~$\|\widetilde{y}(k)\|$
 with D--NG; Subsection~\ref{subsection-D-NG-rate}
  combines these bounds with
  Lemma~\ref{lemma-progress-one-iteration} to derive the convergence rate of D--NG
   and its dependence on the underlying network.
\subsection{Algorithm D--NG: Disagreement estimate}
\label{subsection-D-NG-disagreement}
%
%
%\subsection{Statement of the result}
%\label{subsection-consensus-estimate-statement}
%
This subsection
shows that $\|\widetilde{x}(k)\|$ and $\|\widetilde{y}(k)\|$
 are $O(1/k)$, hence establishing asymptotic consensus --
the differences of the nodes' estimates
$x_i(k)$ (and $y_i(k)$) converge to zero.
 Recall the step-size constant $c>0$ in~\eqref{eqn-alpha-k} and the gradient bound~$G$ in Assumption~\ref{assumption-bdd-gradients}.
%, and we show at which rate this occurs.
%Consider $\widetilde{x}(k):=(\widetilde{x}_1(k)^\top,...,\widetilde{x}_N(k)^\top)$,
% where we recall $\widetilde{x}_i(k):=x_i(k)-\overline{x}(k).$
%  (We analogously define $\widetilde{y}(k)$.)
%subsection, to avoid notational clutter, we let $d=1$, but all the results hold for a generic $d$.
 %
%  \mathbf{1} = (I-J)x(k)$. Denote also
% by $\widetilde{y}(k):=y(k)-\overline{y}(k) \mathbf{1} = (I-J)y(k)$, and
%Recall $\widetilde{W}:=W-J$,
% $\mu(W)=\|\widetilde{W}\|$.
 %
 %
 %
 %
\begin{theorem}[Consensus with D--NG]
\label{theorem-bound-distance-to-consensus}
For D--NG in~\eqref{eqn-our-alg-scalar}--\eqref{eqn-alpha-k} under Assumptions~\ref{assumption-network} and~\ref{assumption-bdd-gradients}:
%
% and step-size $\alpha_k=c/(k+1)$,~$c>0$:
%
%
\begin{eqnarray}
\|\widetilde{x}(k)\| &\leq& \sqrt{N}\,c\,{G}\,C_{\mathrm{cons}}\,\frac{1}{k} \\
\|\widetilde{y}(k)\| &\leq&  4\,\sqrt{N}\, c\,{G}\,C_{\mathrm{cons}}\,\frac{1}{k},\:k=1,2,..., \nonumber \\
\label{eqn-C-cons}
C_{\mathrm{cons}}\hspace{-2mm}& =&\hspace{-2mm}
\frac{8\left\{ 2\,\mathcal{B}\left(\sqrt{\mu(W)}\right) + \frac{7}{1-\mu(W)}\right\} }{\sqrt{{\eta}(1-\mu(W))}},
\end{eqnarray}
%
%
%where $C_{\mathrm{cons}}$ is a constant that depends on the matrix ${W}$ and equals:
%
%\vspace{-0mm}
%
%
%\begin{eqnarray}
% %\nonumber
%\end{eqnarray}
%
%
with $\mathcal{B}(r) := \sup_{z \geq 1/2} \left( z r^{z}\log(1+z)\right) \in (0,\infty),$ $r \in (0,1).$
%and $C_h = 0.577215\cdots $ is the Euler--Mascheroni constant.
%
\end{theorem}
For notational simplicity, we prove Theorem~\ref{theorem-bound-distance-to-consensus} for $d=1$,
but the proof extends to a generic $d>1.$ We model the dynamics of
the augmented state~$(\widetilde{x}(k)^\top, \widetilde{x}(k-1)^\top)^\top$
 as a linear time varying system
 with inputs~$(I-J)\nabla F(y(k))$.
 We present here the linear system and solve it in the Appendix.
 %For simplicity, we take $d=1$ but the
 %same arguments apply for a generic $d$.
 Substitute the expression for $y(k-1)$ in~\eqref{eqn-our-alg}; multiply the resulting equation from the left by $(I-J)$; use $(I-J)W=\widetilde{W}=\widetilde{W}(I-J)$;
 and set $\widetilde{x}(0)=0$~by~assumption. We~obtain:
 %
%%
%%
%\begin{eqnarray*}
%\widetilde{x}(k) &=& (1+\beta_{k-2})\widetilde{W}\widetilde{x}(k-1)  -\beta_{k-2}\widetilde{W}\widetilde{x}(k-2) -\alpha_{k-1} (I-J) \nabla F(y(k-1)),\:k=1,...\\
%\widetilde{x}(0) &=& (I-J)x(0),\:\:\widetilde{x}(-1)=0, \nonumber
%\end{eqnarray*}
%%
%and $\beta_k$, for $k=0,1,...$, is in~\eqref{eqn-alpha-k}, and
%$\beta_{-1}=0.$ Recall that we assumed $x_i(0)=x_j(0)$ for all $i,j$, so that $\widetilde{x}(0)=0.$
%%
%Next, the recursion for the $2N \times 1$ augmented state $(\widetilde{x}(k)^\top, \widetilde{x}(k-1)^\top)^\top$ becomes,
%for $k=1,2,...$:
{
\small{
\begin{eqnarray}
\left[ \begin{array}{cc} \widetilde{x}(k) \\
\widetilde{x}(k-1) \end{array} \right] \hspace{-2mm}&=&\hspace{-2mm}
\left[ \begin{array}{cc}
(1+\beta_{k-2}) \widetilde{W} & -\beta_{k-2} \widetilde{W} \nonumber \\
I & 0
 \end{array} \right]
 \, \left[ \begin{array}{cc} \widetilde{x}(k-1) \\
\widetilde{x}(k-2) \end{array} \right] \\
 &-& \hspace{-2mm}
 \label{eqn-recursion}
 \alpha_{k-1} \left[ \begin{array}{cc} (I-J)\nabla F(y(k-1))\\
0 \end{array} \right] ,
\end{eqnarray}}}
for all $k=1,2,...$, where $\beta_k$, for $k=0,1,...$, is in~\eqref{eqn-alpha-k}, $\beta_{-1}=0,$
and $(\widetilde{x}(0)^\top,\widetilde{x}(-1)^\top)^\top=0$.
We emphasize that system~\eqref{eqn-recursion}
is more complex than the corresponding
systems in, e.g.,~\cite{nedic_T-AC,duchi}, which
involve only a single state $\widetilde{x}(k)$;
the upper bound on $\|\widetilde{x}(k)\|$ from~\eqref{eqn-recursion}
is an important technical contribution of this paper; see Theorem~\ref{theorem-bound-distance-to-consensus}
 and Appendix~A.

\subsection{Convergence rate and network scaling}
\label{subsection-D-NG-rate}
% In addition to Assumptions \ref{assumption-network} and \ref{assumption-f-i-s},
% we require the gradients of the functions $f_i$ be bounded.
%
%
%
%
%
%
%
%\subsection{Basic convergence result}
%
%We state our main result
%
%
%
%We now state the convergence rate result for algorithm D--NG.
 Theorem~\ref{theorem-basic-convergence-result}~(a)
states the $O\left( \log k/k\right)$ convergence
rate result for D--NG when the step-size constant
$c\leq 1/(2L)$; Theorem~\ref{theorem-basic-convergence-result}~(b) (proved in~\cite{arxivVersion})
 demonstrates that the $O\left( \log k/k\right)$
 convergence rate still holds if $c>1/(2L)$,
 with a deterioration in the convergence constant. Part~(b) assumes
 $x_i(0)=y_i(0)=0$, $\forall i$, to avoid notational clutter.
\begin{theorem}
\label{theorem-basic-convergence-result}
Consider D--NG under Assumptions~\ref{assumption-network}--\ref{assumption-bdd-gradients}.
 Let $\|\overline{x}(0)-{{x^\star}}\| \leq R$, $R \geq 0$. Then:
\begin{enumerate}[(a)]
\item If $c \leq 1/(2L)$, we have, $\forall i$,~$\forall k=1,2,...$:
\begin{eqnarray}
&\,&\frac{f(x_i(k))-f^\star}{N}
\leq  \frac{2\, R^2}{c}  \left(\frac{1}{k}\right)
 + 16\,c^2\, \,L \,C_{\mathrm{cons}}^2 \,G^2 \nonumber \\
 &\times& \left( \frac{1}{k} \sum_{t=1}^{k-1} \frac{(t+2)^2}{(t+1)t^2} \right)+
\, c\, \sqrt{N}\,G^2  C_{\mathrm{cons}}\, \left(\frac{1}{k}\right)  \nonumber \\
% \begin{eqnarray}
% \label{eqn-opt-gap-for-scaling}
% \frac{1}{N} \left( f(x_i(k))-f^\star\right)
&\leq& \hspace{-.5mm}
\mathcal{C}\, \left( \frac{1}{k}\sum_{t=1}^{k} \frac{(t+2)^2}{(t+1)t^2}\right) \nonumber \\
\mathcal{C} &=&  \frac{2R^2}{c} + 16 c^2 L C_{\mathrm{cons}}^2 G^2 +  \,c \,\sqrt{N}\,G^2 C_{\mathrm{cons}}.
\label{eqn-constant-scaling}
 \end{eqnarray}
 %\end{eqnarray}
% \end{enumerate}
%
%
\item Let $x_i(0)=y_i(0)=0,\forall i.$
If $c > 1/(2L)$,~\eqref{eqn-constant-scaling} holds $\forall i$,
$\forall k \geq 2 \,c L$,
with $\mathcal{C}$ replaced with
$\mathcal{C}^\prime  = \mathcal{C}^{\prime \prime}(L,G,R,c) + 16 c^2 L C_{\mathrm{cons}}^2 G^2 +  \,c \,\sqrt{N}\,G^2 C_{\mathrm{cons}}$, and
$\mathcal{C}^{\prime \prime}(L,G,R,c) \in [0,\infty)$
is a constant that depends on $L,G,R,c$, and \emph{is independent of}
 $N$~and~$W$.
 \end{enumerate}
\end{theorem}
We prove here Theorem~\ref{theorem-basic-convergence-result}~(a); for part~(b), see~\cite{arxivVersion}.
\begin{IEEEproof}[Proof of Theorem~\ref{theorem-basic-convergence-result}~(a)] The proof consists of two parts.
In the Step~1 of the proof, we estimate the optimality gap $\frac{1}{N}(f(\overline{x}(k))-f^\star)$ at the point $\overline{x}(k)=
\frac{1}{N}\sum_{i=1}^N x_i(k)$ using
Lemma~\ref{lemma-progress-one-iteration} and the inexact oracle machinery. In the Step~2, we
 estimate the optimality gap $\frac{1}{N}(f(x_i(k))-f^\star)$ at any node $i$ using convexity of the $f_i$'s
 and the bound on $\|\widetilde{x}(k)\|$ from Theorem~\ref{theorem-bound-distance-to-consensus}.

\textbf{Step 1. Optimality gap $(f(\overline{x}(k))-f^\star)$.} Recall that, for $k=1,2,...,$ $(\widehat{f}_k, \widehat{g}_k)$ in \eqref{eqn-f-k-g-k}
 is a $(L_k,\delta_k)$ inexact oracle of $f$ at point $\overline{y}(k)$ with
  $L_k=2 N L$ and $\delta_k = L \|\widetilde{y}(k)\|^2$. Note that $(\widehat{f}_k, \widehat{g}_k)$
  is also a $(L^\prime_k,\delta_k)$ inexact oracle of $f$ at point $\overline{y}(k)$ with
  $L^\prime_k= N \frac{1}{c} (k+1) = \frac{N}{\alpha_k}$, because $\frac{1}{c} \geq 2 L$, and so $L^\prime_k \geq L_k$.
  Now, we apply Lemma~\ref{lemma-progress-one-iteration}
  to~\eqref{eqn-our-alg-CM}, with
%  $\phi \equiv f$, $\widehat{x}(k) \equiv \overline{x}(k)$, $\widehat{y}(k) \equiv \overline{y}(k)$,
${x^\bullet} = {{x^\star}}$, and the Lipschitz constant $L_k^\prime = 1/(\alpha_k/N).$
   %Note also that
   %$\delta_k = L\sum_{i=1}^N \|y_i(k)-\overline{y}(k)\|^2 = L\|\widetilde{y}(k)\|^2$.
   Recall that $\overline{v}(k) = \frac{\overline{y}(k) - (1-\gamma_{k})\overline{x}(k)}{\gamma_{k}}.$~We~get:
  {
  \small{
\begin{eqnarray}
\label{eqn-progress-eqn-proof}
&\,&\frac{(k+1)^2}{k} \left( f(\overline{x}(k)) - f^\star \right) + \frac{2 N}{c} \|\overline{v}(k)-{{x^\star}}\|^2 \\
&\leq&
\frac{k^2-1}{k} \left( f(\overline{x}(k-1)) - f^\star \right) + \frac{2 N}{c}  \|\overline{v}(k-1)-{{x^\star}}\|^2 \nonumber \\
&+& L \|\widetilde{y}(k-1)\|^2 \frac{(k+1)^2}{k}. \nonumber
\end{eqnarray}}}
%
%
%Define $\gamma_{k-1}:=\frac{k^2-1}{k}$.
Because
 $\frac{(k+1)^2}{k}
 \geq \frac{(k+1)^2-1}{k+1}$, and
 $ \left( f(\overline{x}(k)) - f^\star \right) \geq 0$, we have:
 {
 \small{
\begin{eqnarray*}
&\,&\frac{(k+1)^2-1}{k+1} \left( f(\overline{x}(k)) - f^\star \right) + \frac{2 N}{c}  \|\overline{v}(k)-{{x^\star}}\|^2 \\
&\leq&
\frac{k^2-1}{k} \left( f(\overline{x}(k-1)) - f^\star \right) + \frac{2 N}{c}  \|\overline{v}(k-1)-{{x^\star}}\|^2 \\
&+& L \|\widetilde{y}(k-1)\|^2 \frac{(k+1)^2}{k}.
\end{eqnarray*}}}
By unwinding the above recursion, and using $\overline{v}(0)=\overline{x}(0)$, gives:
$
 \frac{(k+1)^2-1}{k+1} \left( f(\overline{x}(k)) - f^\star \right) \leq
 \frac{2 N}{c}  \|\overline{x}(0)-{{x^\star}}\|^2
+ L \sum_{t=1}^k \|\widetilde{y}(t-1)\|^2 \frac{(t+1)^2}{t}.
$
Applying Theorem~\ref{theorem-bound-distance-to-consensus} to the last equation,
and using $\frac{k+1}{(k+1)^2-1}=\frac{k+1}{k(k+2)} \leq \frac{k+2}{k(k+2)} = \frac{1}{k}$,
and the assumption $\|\widetilde{y}(0)\|=0$, leads to, as desired:
\begin{eqnarray}
\label{eqn-opt-gap-overline-x}
 \left( f(\overline{x}(k)) - f^\star \right) &\leq&
 \frac{1}{k} \frac{2 N}{c}  \|\overline{x}(0)-{{x^\star}}\|^2 \\
&+& \frac{16\,c^2\,N}{k} L\, C_{\mathrm{cons}}^2 G^2 \,\sum_{t=2}^k  \frac{(t+1)^2}{t (t-1)^2}. \nonumber
\end{eqnarray}

\textbf{Step 2. Optimality gap $(f(x_i(k))-f^\star)$}. Fix an arbitrary node $i$; then, by convexity of $f_j$, $j=1,2,...,N$:
$
f_j(\overline{x}(k)) \geq f_j(x_i(k)) + \nabla f_j(x_i(k))^\top (\overline{x}(k)-x_i(k)),
$
and so:
$
f_j(x_i(k)) \leq f_j(\overline{x}(k)) + G \|\overline{x}(k)-x_i(k)\|.
$
Summing the inequalities for $j=1,...,N$, using $\|\overline{x}(k)-x_i(k)\| \leq \|\widetilde{x}(k)\|$,
%using $\sum_{i=1}^N \|\overline{x}(k)-x_i(k)\| = \sum_{i=1}^N \|\widetilde{x}_i(k)\| \leq \sqrt{N} \|\widetilde{x}(k)\|$,
 subtracting $f^\star$ from both sides, from Theorem~\ref{theorem-bound-distance-to-consensus}:
\begin{eqnarray}
\label{eqn-relate-gaps}
f(x_i(k)) - f^\star &\leq& f(\overline{x}(k)) - f^\star + G N \|\widetilde{x}(k)\| \\
&\leq& f(\overline{x}(k)) - f^\star +c\, N \sqrt{N} \, C_{\mathrm{cons}} G^2 \frac{1}{k}, \nonumber
\end{eqnarray}
which, with~\eqref{eqn-opt-gap-overline-x}~where the summation variable $t$
is replaced by $t+1$, completes~the~proof.
% Proof of (b) follows by replacing ${G}$ with $G.$
\end{IEEEproof}

%
%
%
%
%\ssuection{Network scaling and Comparison with Existing Work}
%\label{section-network-scaling}

\textbf{Network Scaling}. %We examine how the convergence rate depends on $N$ and the network topology.
%We assume that $L$, $G$, and $R$ do
% not depend on $N$. Just to formally set up the network scaling, suppose
% that we have a given sequence of the weight matrices
% $W^{(1)}$, $W^{(2)}$,..., $W^{(N)}$,...,
% where the size of $W^{(N)}$ is $N \times N$. (The matrices
% increase in size $N$ to emulate the increase of the network.) Then, we consider how $\mathcal{C}(N)$ in Theorem~5
%  depends on the matrix $W=W^{(N)}$.
%  Recall that $\mu(W)$ equals
%  the second largest eigenvalue of $W^{(N)}$.
%   For simplicity, write simply $W$ instead of $W^{(N)}$; likewise,
%   we write $\widetilde{W}=W^{(N)}-J$.
% With many network models, $\mu(W) \rightarrow 1$ as $N \rightarrow \infty$,
% i.e., the network speed of consensus deteriorates with the increase of $N$.
% The expander graphs are exceptions, for which the weight
% matrix $W$ can be chosen such that $1-\mu(W)=\Omega(1)$.
  Using Theorem~\ref{theorem-basic-convergence-result},
  Theorem~\ref{theorem-scaling} studies the dependence of the convergence rate on
  the underlying network --
  $N$ and $W$, when: 1) nodes do not know $L$ and $\mu(W)$ before the algorithm run, and they
  set the step-size constant $c$ to a constant independent of $N,L,W$, e.g., $c=1$; and 2) nodes
  know $L,\mu(W)$, and they set $c=\frac{1-\mu(W)}{2 L}.$ See~\cite{duchi} for dependence of~$1/(1-\mu(W))$ on $N$ for commonly used models, e.g., expanders or geometric graphs.
%
 %In the following,
%  we assume $\mu(W) \rightarrow 1$ as $N \rightarrow \infty$; afterwards,
%  in Corollary \ref{corollary-specific-topologies} (d), we treat expanders separately.
% %
 %
 %
 %
\begin{theorem}
\label{theorem-scaling}
Consider the algorithm D--NG in~\eqref{eqn-our-alg-scalar}--\eqref{eqn-alpha-k} under Assumptions~\ref{assumption-network}--\ref{assumption-bdd-gradients}.~Then,
 $\frac{1}{N}\left(f(x_i(k))-f^\star \right)$ is:
\[
O\left( \frac{1}{(1-\mu)^{p+\xi}}\left[
\frac{\log k}{k}+\frac{N^{1/2}\log^{1/2} k}{k^{3/2}}+\frac{N}{k^2}\right]\right),\]
where:
% Also, suppose
%that $\lambda_2(\widetilde{W})=\Omega(1)$ as $N \rightarrow \infty$.\footnote{This is true, e.g.,
%with the weights example based on the neighbors' degrees in Section~\ref{section-problem-model}.}
 (a)~$p=3$ for arbitrary $c=\mathrm{const}>0$; and (b)~$p=1$ for $c=\frac{1-\mu(W)}{2 L}$.
 %:
%$
%\frac{1}{N}\left(f(x_i(k))-f^\star \right) =
% O \left(
% \frac{1}  { \left( 1-\mu \right)^{1+\xi}}\left(
%\frac{\log k}{k}+\frac{\sqrt{N}\log^{1/2} k}{k^{3/2}}+\frac{N \log k}{k^2}\right)\right)
%.$ %\:\:\:\:\mathrm{as}\:\|(N,k)\| \rightarrow \infty.
% %\end{eqnarray}
% \end{enumerate}
%
\end{theorem}

 \begin{IEEEproof}[Proof of Theorem~\ref{theorem-scaling}]
  Fix ${\eta} \in (0,1)$ and $\xi \in (0,1)$ (two arbitrarily small positive constants).
  By~Assumption~\ref{assumption-network}~(b), $\mu=\mu(W) \in [{\eta},1]$.
   We show that for $C_{\mathrm{cons}}$ in~\eqref{eqn-C-cons}:
   \begin{equation}
   \label{eqn-c-cons-proof}
   C_{\mathrm{cons}} \leq A(\xi,{\eta})\,\,\frac{1}{(1-\mu)^{3/2+\xi}},\:\:\forall \mu \in [{\eta},1],
      \end{equation}
where $A(\xi,{\eta}) \in (0,\infty)$
 depends only on $\xi,{\eta}$.
Consider $\mathcal{B}(r) = \sup_{z\geq 1/2} \left\{ z \,r^z \log(1+z) \right\},$ $r \in (0,1);$
there exists~$K_B(\xi) \in (0,\infty)$ such that:
 $ \log(1+z) \leq K_B(\xi) z^{\xi}$, $\forall z \geq 1/2.$
 %(The logarithm $\log(1+z)$ grows slower than
 Thus:
 \begin{eqnarray*}
 %\label{eqn-ovaj}
 &\,&\mathcal{B}(r) \leq K_B(\xi) \, \sup_{z \geq 1/2} \left\{  z^{1+\xi} r^z  \right\} \\
                &=&
                 \frac{K_{B}(\xi) \,e^{-(1+\xi)} (1+\xi)^{(1+\xi)}}{(-\log r)^{1+\xi}}
                 =:\frac{A^\prime(\xi)}{(-\log r)^{1+\xi}},
 \end{eqnarray*}
for all $r\in(0,1)$.
From the above equation, and using $1/(-\log \sqrt{u}) \leq 2/(1-u)$, $\forall u \in [0,1)$,
 we have $\mathcal{B}\left(\sqrt{\mu}\right) \leq 2 A^\prime(\xi)/(1-\mu)^{1+\xi}$.
 The latter, applied to~\eqref{eqn-C-cons}, yields~\eqref{eqn-c-cons-proof}, with
 $A(\xi,{\eta}):=\frac{8}{\sqrt{{\eta}}} \max\left\{3 A^\prime(\xi),7 \right\}.$

A scaling result $O\left(\frac{N^{1/2}}{(1-\mu)^{p+\xi}}\frac{\log k}{k}\right)$, $p=3,1$, readily follows by substitution of~\eqref{eqn-c-cons-proof} in
Theorem~\ref{theorem-basic-convergence-result}~(a) and~(b), respectively.
 To prove Theorem~6,
 we modify the argument of~\eqref{eqn-relate-gaps}.
 We first prove claim~(b).
 Namely, at any node~$i$, using Lipschitz
 continuity of $\nabla f$ (with constant $N L$),
  $ f(x_i(k)) \leq f(\overline{x}(k)) + \nabla f(\overline{x}(k))^\top (x_i(k)-\overline{x}(k))
 +\frac{N L}{2}\|x_i(k)-\overline{x}(k)\|^2$, and thus:
 \begin{eqnarray}
 \label{eqn-nova-jednacina}
f(x_i(k)) \hspace{-1mm}\leq \hspace{-1mm}
 f(\overline{x}(k))\hspace{-1mm} + \hspace{-1mm}\|\nabla f(\overline{x}(k))\|  \|\widetilde{x}(k)\|
 +\hspace{-1mm}\frac{N L \|\widetilde{x}(k)\|^2}{2},
 \end{eqnarray}
 where we use $\|x_i(k)-\overline{x}(k)\|\leq \|\widetilde{x}(k)\|$.
 From~\eqref{eqn-opt-gap-overline-x},
 $f(\overline{x}(k))-f^\star=O\left(\frac{N}{c k}
 +\frac{N c^2 C_{\mathrm{cons}}^2 \log k}{k} \right)$.
  Using again Lipschitz
 continuity of~$\nabla f$ (with constant~$N L$):
 $\|\nabla f(\overline{x}(k))\| \leq \sqrt{2 N L}\sqrt{f(\overline{x}(k))-f^\star}
 = O\left(\frac{N}{\sqrt{c k}}
 +\frac{N c C_{\mathrm{cons}} \log^{1/2} k}{\sqrt{k}} \right)$.
 Consider~\eqref{eqn-nova-jednacina}. Subtracting $f^\star$ from both sides,
 dividing by $N$, and substituting the above bound on $\|\nabla f(\overline{x}(k))\|$ while using Theorem~5~(a),
 we obtain:
 \begin{eqnarray}
 \label{eqn-nova-nova}
 &\,&\frac{f(x_i(k))-f^\star}{N} =
 O(\frac{1}{c\,k}+\frac{c^2 C_{\mathrm{cons}}^2 \log k}{k}\\
& + &
 \left(\frac{1}{\sqrt{c\,k}}+\frac{c C_{\mathrm{cons}} \log^{1/2} k}{\sqrt{k}}  \right)
 \frac{ \sqrt{N} c C_{\mathrm{cons}}}{k} +\frac{N c^2 C_{\mathrm{cons}}^2}{k^2} ). \nonumber
 \end{eqnarray}
  We now apply~\eqref{eqn-c-cons-proof} to~\eqref{eqn-nova-nova}. Claim~(b) is proved after setting~$c=(1-\mu)/2L$.
   The proof for claim~(a) is completely analogous; the argument only replaces the term $\frac{1}{k} \frac{2 N}{c} R^2$ in~\eqref{eqn-opt-gap-overline-x} with $\frac{1}{k}\mathcal{C}^{\prime \prime}(L,G,R,C)$, see also~\cite{arxivVersion},
   and sets $c=\Theta(1)$.
  %
 %
 %
 %Claim~(a) now follows by taking arbitrary $c=\Theta(1)$ and applying~\eqref{eqn-c-cons-proof} to Theorem~\ref{theorem-basic-convergence-result}~(b);
%  and claim~(b) follows by taking $c=\frac{1-\mu}{2 L}$ and applying~\eqref{eqn-c-cons-proof} to Theorem~\ref{theorem-basic-convergence-result}~(a).
 \end{IEEEproof}
\vspace{-0mm}
\section{Algorithm D--NC: Convergence Analysis}
\label{section-converg-rate-D-NC}
%
%%
%In this Section, we analyze the convergence rate of
%algorithm D--NC.
We now consider the D--NC algorithm.
Subsection~\ref{subsection-D-NC-disagreement}
provides the disagreement estimate, while
Subsection~\ref{subsection-D-NC-disagreement} gives the convergence rate and network scaling.
\vspace{-0mm}
\subsection{Disagreement estimate}
\label{subsection-D-NC-disagreement}
%
%
%
%
%
%\textbf{Inexact oracle framework}.

%We carry out the convergence analysis of~D--NC
%by: 1) estimating the distance to consensus $\|\overline{y}(k)-y_i(k)\|$ and $\|\overline{x}(k)-x_i(k)\|$;
%and 2) estimating the optimality gap $f(\overline{x}(k))-f^\star$ using
%the inexact oracle framework.
%Clearly, evaluation of $f(\overline{x}(k))-f^\star$ follows from Lemma~\ref{lemma-progress-one-iteration}
% and requires little additional effort. Consensus estimate, however,
% requires new analysis.
%
%
%
%
%
%
%
%

%\textbf{Consensus Estimate}.
%\label{section-consensus-estimate}
%
%
%
%For notational simplicity, throughout this subsection, we set $d=1$,
%but all results hold for generic $d$.
We estimate the disagreements
$\widetilde{x}(k)$,
and $\widetilde{y}(k)$
 with~D--NC.
 %We have the following result.
 %
 %
 %
 %
 \begin{theorem}[Consensus with D--NC]
 \label{theorem-distance-cons-consensus-alg}
 Let Assumptions~\ref{assumption-network}~(a) and~\ref{assumption-bdd-gradients} hold, and consider
 the algorithm D--NC.
 %Further, set
  %the number of the inner consensus iterations at the $k$-outer iteration $k$ as in~\eqref{eqn-consensus}.
  Then,
   for $k=1,2,...$:
  $
  \|\widetilde{x}(k)\| \leq 2 \alpha \sqrt{N} {G} \frac{1}{k^2}$, and
  $\|\widetilde{y}(k)\| \leq 2 \alpha \sqrt{N} {G} \frac{1}{k^2}.$
 \end{theorem}
%
%
%
%
%None of the Assumptions~\ref{assumption-network}~(a) nor~\ref{assumption-bdd-gradients} can be dropped;
%regarding Assumption~\ref{assumption-network}~(a), the same counterexample as with Theorem~\ref{theorem-bound-distance-to-consensus} can be taken here; regarding Assumption~\ref{assumption-bdd-gradients}, the example is similar to Theorem~\ref{theorem-bound-distance-to-consensus}; see Subsection~\ref{subsection-discussion-bounded-gradients} and~Appendix~\ref{subsection-lower-bound-without-bdd-grad}.
%
%
\begin{IEEEproof} For notational simplicity, we perform the proof for $d=1$,
but it extends to a generic~$d>1$.
Denote by $B_{t-1}:=\max \left\{  \|\widetilde{x}(t-1)\|,\,\,\|\widetilde{y}(t-1)\|\right\}$,
and fix $t-1$. We want to upper bound~$B_t$.
Multiplying~\eqref{eqn-our-alg-cons}--\eqref{eqn-our-alg-cons-druga}
by $(I-J)$ from the left, using~$(I-J)W =\widetilde{W}(I-J)$:
\vspace{0mm}
\begin{eqnarray}
\label{eqn-our-alg-cons-tilde}
\widetilde{x}(t) &=& \widetilde{W}^{\tau_x(t)} \, \widetilde{y}(t-1) \\
                 &-& \alpha \widetilde{W}^{\tau_x(t)} (I-J)\nabla F(y(t-1)) \nonumber \\
\label{eqn-our-alg-cons--tilde-druga}
\widetilde{y}(t) &=& \widetilde{W}^{\tau_y(t)} \left[ \,\,\widetilde{x}(t) + \beta_{t-1} (\widetilde{x}(t)-
\widetilde{x}(t-1))\,\,\right].
\end{eqnarray}
\vspace{-0mm}
We upper bound $\|\widetilde{x}(t)\|$ and $\|\widetilde{y}(t)\|$ from~\eqref{eqn-our-alg-cons-tilde},~\eqref{eqn-our-alg-cons--tilde-druga}.
Recall $\|\widetilde{W}\|=\mu(W) :=\mu \in (0,1)$;
from~\eqref{eqn-consensus} and~\eqref{eqn-consensus-2}, we have
 $\mu^{\tau_x(t)} \leq \frac{1}{t^2}$ and $\mu^{\tau_y(t)} \leq \frac{1}{3 t^2}.$
From~\eqref{eqn-our-alg-cons-tilde}, using
the sub-additive and sub-multiplicative properties of norms, and using
$\|\widetilde{y}(t-1)\|\leq B_{t-1}$, $\mu \in (0,1)$,
 $\|(I-J)\nabla F(y(t-1))\| \leq \|\nabla F(y(t-1)) \| \leq \sqrt{N} G$, $\beta_{t-1}\leq 1$:
\begin{eqnarray}
\|\widetilde{x}(t)\|
&\leq&
%
%
%(\mu(W))^{\tau_{t-1}} \, \|\widetilde{y}(t-1)\| + \alpha (\mu(W))^{2\, \tau_{t-1}} \sqrt{N} {G} %
% \nonumber \\
%
%
%&\leq&
\mu^{\tau_x(t)} \, B_{t-1} + \alpha \mu^{\tau_x(t)} \sqrt{N} {G} \nonumber \\
\label{eqn-111}
&\leq& \frac{1}{t^2}\,B_{t-1}+\alpha \sqrt{N} G \frac{1}{t^2}%
 \\
\|\widetilde{y}(t)\| &\leq& 2\, \mu^{\tau_y(t)} \|\widetilde{x}(t) \| + \mu^{\tau_y(t)} \|\widetilde{x}(t-1)\|
\nonumber \\
&\leq&
2\, \mu^{\tau_x(t)+\tau_y(t)} B_{t-1} + 2 \alpha \sqrt{N} G\,\mu^{\tau_x(t)+\tau_y(t)} \nonumber\\
&+& \mu^{\tau_y(t)} B_{t-1} \nonumber \\
&\leq&
3 \,\mu^{\tau_y(t)} B_{t-1} +  2 \alpha \sqrt{N} G\mu^{\tau_y(t)} \nonumber \\
\label{eqn-222}
&\leq& \frac{1}{t^2} B_{t-1} + \alpha \sqrt{N} G \frac{1}{t^2}.
\end{eqnarray}
%
%
%where~\eqref{eqn-222} uses $\mu^{\tau_x(t)} \leq 1.$

%\end{IEEEproof}
   %
   %
 %  Using $(\mu(W))^{2 \tau_{t-1}} \leq (\mu(W))^{\tau_{t-1}} $
%    and $(\mu(W))^{3 \tau_{t-1}} \leq (\mu(W))^{\tau_{t-1}} $ to further
%    upper bound~\eqref{eqn-111} and~\eqref{eqn-222}, and taking the maximum
%    over the two resulting inequalities:
%   %
   %%
%   \begin{equation}
%   \label{eqn-b-t-update}
%   B_t \leq
%   3\, (\mu(W))^{\tau_{t-1}}  B_{t-1} + 2 \alpha \sqrt{N} (\mu(W))^{ \tau_{t-1}} {G}.
%   %+ (\mu(W))^{\tau_{t-1}} B_{t-1}.
%   \end{equation}
  %
  %
  %
  %
  %
Clearly, from~\eqref{eqn-111} and~\eqref{eqn-222}:%
%%
%%
%%
%$
%%(\mu(W))^{\tau_{t-1}} + 2\, (\mu(W))^{2\,\tau_{t-1}} \leq
%3\, (\mu(W))^{\tau_{t-1}}
%%
%%
%\leq
%3 e^{-\frac{\log 3}{\log \mu(W)} \log \mu(W)} e^{-\frac{log t^3}{\log \mu(W)} \log \mu(W)}
%=
%\frac{1}{t^3}.
%$
%%
%%
%%
%Applying the latter to~\eqref{eqn-b-t-update}, and
%using $(\mu(W))^{ \tau_{t-1}}\leq  \frac{1}{t^3}$:
%%
%%
%%
%%
$
   B_t \leq
   \frac{1}{t^2} B_{t-1} + \frac{1}{t^2} \alpha \sqrt{N}  {G}.
   %+ (\mu(W))^{\tau_{t-1}} B_{t-1}.
$
Next, using $B_0=0$, unwind the latter recursion
for $k=1,2$, to obtain, respectively: $B_1 \leq \alpha \sqrt{N} G$
 and $B_2 \leq \alpha \sqrt{N} G/2$, and so the bound in Theorem~7 holds
 for $k=1,2.$ Further, for $k \geq 3$ unwinding the same recursion for $t=k,k-1,...,1$:
\begin{eqnarray*}
%\label{eqn-unwinded}
B_k %&\leq& \alpha \sqrt{N}  {G} \left( \frac{1}{k^2}+\sum_{t=2}^{k-1}\frac{1}{k^2(k-1)^2...t^2}+\frac{1}{k^2(k-1)^2...2^2}    \right) \\
&\leq&
\frac{\alpha \sqrt{N} G}{k^2}  ( 1+\sum_{t=2}^{k-1}\frac{1}{(k-1)^2(k-2)^2...t^2}\\
&+& \frac{1}{(k-1)^2(k-2)^2...2^2 })\\
&\leq&
\frac{\alpha \sqrt{N} G}{k^2}  \left( 1+\sum_{t=2}^{k-1}\frac{1}{t^2}+\frac{1}{2^2}\right)\\
&\leq&
\frac{\alpha \sqrt{N} G}{k^2}  \left( \frac{\pi^2}{6}+\frac{1}{4}\right)
\leq \frac{2\,\alpha \sqrt{N} G}{k^2},%\:\:k=1,2,...,
%
%
%\left( \frac{1}{k^2}+\frac{1}{k^2(k-1)^2}+\frac{1}{k^2(k-1)^2(k-2)^2}+...+\frac{1}{k^2(k-1)^2...3^2} + 2 \frac{1}{k^2(k-1)^2...3^2 2^2} \right)\\
%&\leq&
% \alpha \sqrt{N}  {G}\,\frac{1}{k^2} \left( 1+ \frac{1}{(k-1)^2}+\frac{1}{(k-2)^2} +...+  \frac{1}{2^2} +\frac{1}{2^2}  \right)\\
% & = &\alpha \sqrt{N} G\, \frac{1}{k^2} \left( \frac{1}{4}+ 1+\frac{1}{2^2}+...+\frac{1}{(k-1)^2} \right) \leq
%  \left(\frac{\pi^2}{6}+\frac{1}{4} \right) \frac{\alpha \sqrt{N} G}{k^2},\:k=1,2,...
\end{eqnarray*}
where we use $1+\sum_{t=2}^{k-1}\frac{1}{t^2} \leq \pi^2/6$,~$\forall k \geq 3.$
%
%
%Thus, the Lemma.
\end{IEEEproof}
\vspace{-0mm}
\subsection{Convergence rate and network scaling}
\label{subsection-D-NC-rate}
%
%
%
%\subsubsection{Convergence rate result}
We are now ready to state the Theorem on the convergence rate of~D--NC.
\begin{theorem}
\label{theorem-conv-rate-cons-alg-new}
Consider the algorithm D--NC under Assumptions~\ref{assumption-network}~(a),~\ref{assumption-f-i-s}, and~\ref{assumption-bdd-gradients}. %,
 %with the constant step size $\alpha \leq 1/(2L)$.
  Let $\|\overline{x}(0)-{{x^\star}}\| \leq R$, $R \geq 0$. Then,
 after
$
 \mathcal{K} = \sum_{t=1}^{k}\left(\tau_x(t)+\tau_y(t)\right) \leq \frac{2}{-\log \mu(W)} \left( k \log 3 + 2(k+1)\log (k+1) \right)
 = O\left(k \log k\right)
$
% \end{equation}
 %
 %
 communication rounds, i.e., after~$k$ outer iterations, at any node~$i$:
\begin{eqnarray}
\label{eqn-theorem-alg-new-cons}
 &\,& \frac{1}{N}\left(f(x_i(k)) - f^\star\right) \\
 &\leq&
\frac{1}{k^2} \left( \frac{2}{\alpha}  R^2
+ 11 \,\alpha^2 L G^2 + \alpha \sqrt{N} G^2 \right),\:\:k=1,2,...\nonumber
\end{eqnarray}
\end{theorem}
%
%
%
%
%Neither Assumption~\ref{assumption-network}~(a) nor~\ref{assumption-bdd-gradients} can be dropped.
%The counterexamples are the same as with Theorem~\ref{theorem-distance-cons-consensus-alg}.
%As noted, Assumption~\ref{assumption-f-i-s} is standard with gradient methods.
%It is even required
%for the D--NC algorithm to be well-defined -- to set the step-size $\alpha \leq 1/(2 L).$
%Nonetheless, we explain what may occur if we relax the gradient's Lipschitz continuity.
%We take $N=2$ and the example functions $f_i$, $i=1,2,$ in~\cite{AnnieChenThesis},~pages~{29--31}.
%  With $W_{11}=W_{22}=W_{12}=W_{21}$, $\alpha=0.1,$ and $x(0)=y(0)=(-1,1)^\top$,
%  simulations show that $f(x_i(k))-f^\star$, $i=1,2$ stays above a value of~$0.1$.
 %
%
%
%
\begin{IEEEproof}[Proof outline]
%As with the proof of Theorem~\ref{theorem-basic-convergence-result},
%we divide the proof in two steps.
%In the first step, we upper bound the optimality
%gap $f(\overline{x}(k))-f^\star$; in the second step,
%we upper bound the optimality gap at each node
%$f(x_i(k))-f^\star.$
%
%\textbf{Step~1: Upper bounding $f(\overline{x}(k))-f^\star$.}
The proof is very similar to the proof of Theorem~\ref{theorem-basic-convergence-result}~(a) (for details see~\cite{arxivVersion}, second version~v2);
 first upper bound $f(\overline{x}(k))-f^\star$, and then~$f(x_i(k))-f^\star$.
%Hence, we detail only the differences. When proving
%
To upper bound $f(\overline{x}(k))-f^\star$, recall that the evolution~\eqref{eqn-our-alg-CM}
  with $\alpha_k=\alpha$ for $(\overline{x}(k),\overline{y}(k))$ is the
  inexact Nesterov gradient with the inexact oracle $(\widehat{f}_k,\widehat{g}_k)$
   in~\eqref{eqn-f-k-g-k}, and
   $(L_k = 2 N L, \,\delta_k = L\|\widetilde{y}(k)\|^2)$.
  % Also, from Lemma~\ref{theorem-distance-cons-consensus-alg},
   % $\delta_k \leq 4 \alpha^2 L N G^2 \frac{1}{k^4}.$
    Then, apply Lemma~\ref{lemma-progress-one-iteration} with
%     $\phi \equiv f$, $\widehat{x}(k) \equiv \overline{x}(k)$,
%      $\widehat{y}(k) \equiv \overline{y}(k)$, $\widehat{v}(k) \equiv \overline{v}(k)$, and
      $x^\bullet \equiv {{x^\star}}$ and $L_{k-1}^\prime=N/\alpha$, %to obtain:
%%
%%
%%
%{
%    \small{
%    \begin{eqnarray*}
%    {(k+1)^2} \left( f(\overline{x}(k)) - f^\star \right) &+& \frac{2 N}{\alpha} \|\overline{v}(k)-{{x^\star}}\|^2 \\
%    &\leq&
%    ({k^2-1}) \left( f(\overline{x}(k-1)) - f^\star \right) + \frac{2 N}{\alpha}  \|\overline{v}(k-1)-{{x^\star}}\|^2
%    + L \|\widetilde{y}(k-1)\|^2 {(k+1)^2}.
%    \end{eqnarray*}
%}
%}
and use Theorem~\ref{theorem-distance-cons-consensus-alg},
to obtain:
\begin{eqnarray}
\label{eqn-opt-gap-f-x-bar-cons}
 f(\overline{x}(k)) - f^\star \hspace{-1mm}&\leq& \hspace{-1mm}
\frac{1}{k^2} \left( \frac{2 N  R^2}{\alpha}
+ 11 \,\alpha^2 L N G^2  \right)
\end{eqnarray}
 Finally, find the bound on $f({x}_i(k))-f^\star$
 analogously to the proof of Theorem~\ref{theorem-basic-convergence-result}~(a).
\end{IEEEproof}
%
%
%

%Before proving Theorem~\ref{theorem-conv-rate-cons-alg-new}, we
%interpret the Theorem.

%\subsubsection{Interpretations}
\vspace{-0mm}
\textbf{Network scaling}. We now give the network scaling for algorithm D--NC in Theorem~\ref{corollary-spec-top-cons-alg}.
%As before with algorithm D--NG, we assume that $L,G,$ and $R$ do note depend on $N$.
 We assume that nodes know $L$ and $\mu(W)$ before the algorithm run.
%Proof of Theorem~\ref{corollary-spec-top-cons-alg} is omitted and is provided in Appendix.
%
%Thus, when $N \rightarrow \infty$:
%%
%%
%%
%$
%\mathcal{C}(N) = O\left( \frac{1}{\left( -\log \mu(W)\right)^{2-\xi}} \right),
%$
%%
%%
%which gives the following result.
%
%
%
%
\begin{theorem}
 \label{corollary-spec-top-cons-alg}
Consider D--NC under Assumptions~\ref{assumption-network}~(a),~\ref{assumption-f-i-s},~and~\ref{assumption-bdd-gradients} with step-size $\alpha\leq 1/(2L)$.
 %and the step-size $\alpha \leq 1/(2L)$.
 Then, after $k$ outer iterations and~$\mathcal{K}$
  communication rounds, at any node $i$, %the optimality gap
  $\frac{1}{N} \left(f(x_i)-f^\star \right)$ is~$
O\left( \frac{1}{\left(\left( 1-\mu \right)\mathcal{K}^{1-\xi}\right)^2}
+
\frac{\sqrt{N}}{\left(\left( 1-\mu \right)\mathcal{K}^{1-\xi}\right)^3}
+
\frac{N}{\left(\left( 1-\mu \right)\mathcal{K}^{1-\xi}\right)^4}
 \right)$ and $O\left(\frac{1}{k^2}+\frac{N^{1/2}}{k^3}+\frac{N}{k^4}\right).$
\end{theorem}
\begin{IEEEproof}
Fix $\xi \in (0,1)$,
and let $\mathcal K$ be the number of elapsed
communication rounds after $k$ outer iterations.
 There exists $C_0(\xi) \in (1,\infty)$,
 such that, $2 \left(k \log 3 + 2(k+1)\log(k+1)\right) \leq C_0(\xi)k^{1+\xi},$ $\forall k \geq 1.$
  The latter, combined with~$1/(-\log \mu(W)) \leq 1/(1-\mu(W))$, $\mu(W) \in [0,1)$,
   and the upper bound bound on~$\mathcal K$ in Theorem~\ref{theorem-conv-rate-cons-alg-new},
   gives: $1/k \leq \left( C_0(\xi)\right) \frac{1}{(1-\mu) {\mathcal K}^{1- \xi}}$.
   Plugging the latter in the optimality gap bound in Theorem~\ref{theorem-conv-rate-cons-alg-new}
   gives a scaling result $O(N^{1/2}/[(1-\mu){\mathcal K}]^2)$
    and $O(N^{1/2}/k^2)$. To prove Theorem~9,
    we proceed analogously to the proof of Theorem~6.
    From Theorem~8 and
    $\|\nabla f(\overline{x}(k))\|\leq \sqrt{2 N L}\sqrt{f(\overline{x}(k))-f^\star}$,
    $\|\nabla f(\overline{x}(k))\| =O(N/k)$. Consider~\eqref{eqn-nova-jednacina}.
    Subtracting $f^\star$, dividing by $N$, and
    using $\|\nabla f(\overline{x}(k))\| = O(N/k)$ and~\eqref{eqn-opt-gap-f-x-bar-cons},
    we obtain~$\frac{1}{N}(f(x_i(k))-f^\star)=O(1/k^2+N^{1/2}/k^3+N/k^4)$.
     Finally, substitute
    $1/k \leq \left( C_0(\xi)\right) \frac{1}{(1-\mu) {\mathcal K}^{1- \xi}}$ in the last bound.
\end{IEEEproof}
\section{Comparisons with the literature and Discussion of the Assumptions}
\label{section-comparisons}
Subsection~\ref{subsection-comparisons} compares
 D--NG,~D--NC, and the distributed (sub)gradient algorithms in~\cite{nedic_T-AC,duchi,AnnieChen}, from the aspects of implementation and
convergence rate; Subsection~\ref{subsection-discussion-bounded-gradients} gives a detailed
discussion on~Assumptions~\ref{assumption-network}--\ref{assumption-bdd-gradients}.
\subsection{Comparisons of D--NG and D--NC with the literature}
\label{subsection-comparisons}
%
%
%
%
%
%
%
%\label{subsection-comparisons}
%We compare~D--NG,~D--NC, and
%existing distributed (sub)gradient algorithms in~\cite{nedic_T-AC,duchi,AnnieChen}, from the aspects of implementation and
% convergence rate.
 %the required number of communication rounds $\mathcal{K}(N,W;\epsilon)$ for $\epsilon$-accuracy.
% We define $\mathcal{K}(N,W;\epsilon)$ as the smallest number of elapsed communication rounds $\mathcal K$ such that $\frac{1}{N}(f(x_i)-f^\star)\leq \epsilon$, $\forall i$.
%%
%
%
%
%
We first set up the comparisons by explaining how to account for Assumption~\ref{assumption-network}~(b) and by
adapting the results in~\cite{AnnieChen,AnnieChenThesis} to our framework.

\textbf{Assumption~\ref{assumption-network}~(b)}. To be fair, we account for Assumption~\ref{assumption-network}~(b) with D--NG as follows.
Suppose that the nodes are given arbitrary symmetric, doubly stochastic weights~$W$
with $\mu(W)<1$ -- the matrix required by D--NC and~\cite{nedic_T-AC,duchi,AnnieChen}.
 (For example, the Metropolis weights~$W$.)
 As the nodes may not be allowed to check whether the given $W$ obeys Assumption~\ref{assumption-network}~(b) or not,
 they modify the weights to $W^\prime:=\frac{1+{\eta}}{2}I+\frac{1-{\eta}}{2}W$, where~$\eta \in (0,1)$
  can be taken arbitrarily small. The matrix~$W^\prime$ obeys Assumption~\ref{assumption-network}~(b), whether $W$ obeys it or not.
  %\footnote{Note that $\lambda_1(W^\prime) \geq \frac{1+{\eta}}{2}
%   - \frac{1-{\eta}}{2}|\lambda_1(W)| \geq {\eta}$,
%because $|\lambda_1(W)|<1$, and so $W^\prime \succeq {\eta}\,I$ whether $W \succeq {\eta}\,I$ or not.}
The modification is done without any required knowledge of the system parameters nor inter-node communication;
node $i$ sets: 1) $W^\prime_{ij}=\frac{1-{\eta}}{2}W_{ij}$, for $\{i,j\} \in E$, $i \neq j$;
  2) $W^\prime_{ij}=0$, for $\{i,j\} \notin E$, $i \neq j$;
  and 3) $W^\prime_{ii}:=1-\sum_{j \neq i} W^\prime_{ij}$. To be fair, when we compare D--NG with other methods (either
   theoretically as we do here or numerically as done in Section~\ref{section-simulations}), we set its weights to~$W^\prime$.
  For theoretical comparisons, from Theorem~\ref{theorem-basic-convergence-result}, the
 convergence rate of D--NG depends on $W^\prime$ through
 the inverse spectral gap~$1/(1-\mu(W^\prime)).$
 It can be shown that $\frac{1}{1-\mu(W^\prime)}=\frac{2}{1-{\eta}}\frac{1}{1-\mu(W)},$
  i.e., the spectral gaps of $W$ and $W^\prime$
   differ only by a constant factor and the weight modification does not affect the convergence rate~(up to a numerical constant); henceforth, we express the theoretical rate for D--NG in terms~of~$W$.

\textbf{References~\cite{AnnieChen,AnnieChenThesis}} develop and analyze non-accelerated and accelerated
 distributed gradient and proximal gradient methods for time-varying networks and convex $f_i$'s that have a differentiable component with Lipschitz
 continuous and bounded gradient and a non-differentiable component with bounded gradient.
 To compare with~\cite{AnnieChenThesis}, we adapt it to our framework
 of static networks and differentiable~$f_i$'s. (We set the non-differentiable components of the $f_i$'s to zero.)
 %The accelerated methods in~\cite{AnnieChenThesis} achieve
  %faster rates than the non-accelerated ones; we focus only on the former.
   References~\cite{AnnieChen,AnnieChenThesis} assume deterministic time-varying networks.
  To adapt their results to our static network setup in a fair way,
  we replace the parameter~$\gamma$ in~\cite{AnnieChen}~(see equation~(7) in~\cite{AnnieChen}) with~$\mu(W)$.
   %\footnote{Using~$\gamma$ gives poor bounds; e.g., $-1/\log \gamma = -1/\log (1-{\underline w}^{N-1})^{1/(N-1)} \approx N (1/\underline{w})^{N}$
   %grows exponentially with $N$. (Here $\underline{w} \in (0,1)$ is a lower bound on the positive entries of $W$.)}
   The references propose two variants of the accelerated algorithm:
              the first (see~(6a)--(6d) in~\cite{AnnieChen}) has $k$ inner consensus iterations at the outer
               iteration~$k$, while the second
               one has~$\lceil  4 \log (k+1)/(-\log \mu) \rceil$~(See Subsection~{III-C} in~\cite{AnnieChen}.)
The bounds established in~\cite{AnnieChen} for the second variant give its
                rate: 1)~$O\left( \frac{N^2}{(1-\mu(W))^2\mathcal{K}^{2-\xi}} \right)$,
                when nodes know $\mu(W)$ and $L$.
                %; and~2)~$O\left( \frac{N^{1/2}}{(1-\mu(W))^2\mathcal{K}^{2-\xi}} \right)$,
%                when they in addition know~$N$ and set the step-size $\alpha =\Theta(1/\sqrt{N}).$
                  The first variant has a slower rate~\cite{arxivVersion}.

\textbf{Algorithm implementation and convergence rate}.
Table~1 compares D--NG, D--NC, the algorithm in~\cite{duchi} and
  the second algorithm in~\cite{AnnieChen} with respect to implementation and the number of communications~$\mathcal{K}(\epsilon;N,W)$ to achieve $\epsilon$-accuracy. Here~$\mathcal{K}(\epsilon;N,W)$
    is the smallest number of communication rounds~$\mathcal{K}$ after which
    $\frac{1}{N}(f(x_i)-f^\star) \leq \epsilon$, $\forall i.$
     Regarding implementation, we discuss the knowledge required a priori by all nodes for: 1) convergence (row~1); and 2) both stopping and optimizing the step-size (row~2). By stopping,
we mean determining a priori the (outer)~iteration~$k_0$ such that $\frac{1}{N}(f(x_i(k))-f^\star) \leq \epsilon$, $\forall k \geq k_0$, $\forall i$.
Optimizing the step size here means finding the
step-size that minimizes the established upper bound (in the reference of interest) on the optimality gap (e.g.,
the bound for D--NG in Theorem~\ref{theorem-basic-convergence-result}~(a).) We assume, with all methods, that~$W$ is already given (e.g., Metropolis.) Regarding~$\mathcal{K}(\epsilon;N,W)$, we neglect the logarithmic
and $\xi$-small factors and distinguish two cases: 1) the nodes have no
 global knowledge (row~3); and 2) the nodes know $L,\mu(W)=:\mu.$
We can see from Table~1 that, without global knowledge (row~3),
 D--NG has better dependence on~$\epsilon$ than~\cite{duchi} and worse dependence on~$N,\mu$.
 Under global knowledge~(row~4), D--NC has better complexity than~\cite{AnnieChen} and
 has better dependence on $\epsilon,\mu$ than~\cite{duchi} and a worse dependence on~$N$.
   % When with the second method in~\cite{AnnieChen} nodes in addition
  %know~$N$, their bound improves to~$O\left(\frac{N^{1/4}}{(1-\mu)\sqrt{\epsilon}}\right)$~\cite{arxivVersion}.
  Further, while D--NG and~\cite{duchi} require no knowledge of any global parameters for convergence~(row~1), D--NC
 and the second algorithm in~\cite{AnnieChen} need $L$ and $\mu(W)$.
 The first variant in~\cite{AnnieChen} requires only $L.$ Also, Table~1 for~\cite{duchi} holds for a wider class of functions, and in row 4, only
 $\mu$ is~needed~\cite{duchi}.
{\begin{tiny}{
\begin{table*}\centering
%\ra{1.3}
\begin{tabular}{@{}lllllll@{}}
\toprule
%& \multicolumn{3}{c}{$w = 8$} & \phantom{abc}& \multicolumn{3}{c}{$w = 16$} &
%\phantom{abc} & \multicolumn{3}{c}{$w = 32$}\\
%\cmidrule{2-4} \cmidrule{6-8} \cmidrule{10-12}
& D--NG & D--NC & \cite{duchi} & \cite{AnnieChen} &
%\\
%\midrule
%$Dir=1$\\
%\midrule $\,$
\\
{\small{Kn. for conver.}} & {\small{none}} & $L,\mu$ & {\small{none}} & $L$, $\mu$ \\
{\small{Kn. for stop.; s.s.}} & $\mu,R,G,L,N$ & $\mu,R,G,L,N$ & $\mu,R,G,N$ & $\mu,R,G,L,N$ \\
{\small{$\mathcal{K}(\epsilon;N,W)$: No kn.}} & $\frac{1}{(1-\mu)^3 \epsilon } \hspace{-.7mm}+\hspace{-.7mm}
\frac{N^{1/3}}{(1-\mu)^2 \epsilon^{2/3} } \hspace{-.7mm}+\hspace{-.7mm} \frac{N^{1/2}}{(1-\mu)^{3/2} \epsilon^{1/2} }$ & {\small{not guarant.}} &
   $ \frac{1}{(1-\mu)^2 \epsilon^2} $     & {\small{not studied}}  \\
{\small{$\mathcal{K}(\epsilon;N,W)$: $L,\mu$}} & $ \frac{1}{(1-\mu)\epsilon}\hspace{-.7mm}
+\hspace{-.7mm}\frac{N^{1/3}}{(1-\mu)^{2/3}\epsilon^{2/3}}\hspace{-.7mm}+\hspace{-.7mm}\frac{N^{1/2}}{(1-\mu)^{1/2}\epsilon^{1/2}} $ &
$ \frac{1}{(1-\mu){\epsilon^{1/2}}} + \frac{N^{1/6}}{(1-\mu){\epsilon^{1/3}}} + \frac{N^{1/4}}{(1-\mu){\epsilon^{1/4}}}$ &
$\frac{1}{(1-\mu)\epsilon^2} $  &
$ \frac{N}{(1-\mu){\epsilon^{1/2}}}$  \\
%\vspace{-1cm}
\bottomrule
\end{tabular}
\vspace{1mm}
\caption{Comparisons of algorithms D--NG, D--NC,~\cite{duchi}, and~\cite{AnnieChen}~(algorithms~1 and~2).}
\vspace{-0mm}
\end{table*}}
\end{tiny}}

\textbf{Global knowledge~$\mu(W),L,G,R$} (as needed, e.g., by D--NG for stopping) can be obtained as follows. Consider~$L$
and suppose each node knows a Lipschitz constant $L_i$ of its own~$f_i$. Then, $L$ can be taken as $L=\max_{i=1,...,N}L_i.$
 Thus, each node can compute $L$ if nodes run a distributed
 algorithm for maximum computation, e.g.,~(\cite{johansson-max-cons},~(1));
 all nodes get $L$ after~$O(\mathrm{Diam})$ per-node communicated scalars, where
 $\mathrm{Diam}$ is the network diameter.
Likewise, a gradient bound~$G$ can be taken as $G=\max_{i=1,...,N}G_i$,
where $G_i$ is a gradient bound for the~$f_i.$ The quantity $\mu(W)$ (equal to the second largest eigenvalue of $W$) can be computed in a distributed way, e.g., by algorithm {DECENTRALOI}, proposed in~\cite{decentraloi} and adapted to the problem like ours in~[\cite{BoydGossip}, Subsection~{IV}-A, p.~2519].
With DECENTRALOI, node $i$ obtains $q^{\mu}_i$, the $i$-th coordinate of
the $N \times 1$ eigenvector $q^{\mu}$ of $W$ that corresponds to $\mu(W)$, (up to $\epsilon$-accuracy) after
$O\left(\frac{\log^2(N/\epsilon) \log N}{1-\mu(W)} \right)$
 per-node communicated scalars~\cite{decentraloi}; then, node $i$
  obtains $\mu(W)$ as:~$\frac{\sum_{j \in O_i} W_{ij} q^{\mu}_j}{q^{\mu}_i}$.

Consider now D--NC when nodes do not have
available their local gradient Lipschitz constants~$L_i$.
Nodes can take a diminishing step size~$\alpha_k=1/(k+1)^p$, $p \in (0,1]$,
 and still guarantee convergence, with a deteriorated rate~$O\left( \frac{1}{\mathcal{K}^{2-p-\xi}}\right).$
 In alternative, it may be possible to employ
 a ``distributed line search,'' similarly to~\cite{distributedLineSearch}. Namely, in the absence of knowledge of the gradient's Lipschitz constant~$L$,
 the centralized Nesterov gradient method
 with a backtracking line search achieves the same rate~$O(1/k^2)$,
  with an additional computational cost per iteration~$k$; see~\cite{Vandenberghe,NesterovLineSearch}.
 It is an interesting research direction to develop
 a variant of distributed line search for D--NC type methods
  and explore the amount of incurred additional communications/computations per
  outer iteration~$k$; due to lack of space, this is left for future work.
  %Finally, we remark that distributed line search methods have been
  %already proposed in a different context in~\cite{distributedLineSearch}.

\textbf{The $\Omega(1/k^{2/3})$ lower bound on the worst-case optimality gap for~\cite{nedic_T-AC}}.
%\label{subsection-lower-bound}
 %
 %
We focus on the dependence on $k$ and $\mathcal{K}$ only (assuming a finite, fixed~$1/(1-\mu(W))$.)
%It is important to note that the optimality gap $O(\log k/k^{1/2})$ and $O(\log \mathcal{K}/\mathcal{K}^{1/2})$ in~\cite{duchi}
%holds for the non-differentiable convex $f_i$'s
% that are Lipschitz continuous (have bounded gradients) --  a wider class of functions than what we consider.
 %To our knowledge, a detailed study of
% the algorithms in~\cite{nedic_T-AC,duchi}
% under Assumptions~\ref{assumption-f-i-s} and~\ref{assumption-bdd-gradients} does not exist.
 We demonstrate that D--NG
has a strictly better worst-case convergence rate in $k$ (and~$\mathcal{K}$) than~\cite{nedic_T-AC},
  when applied to the $f_i$'s defined by Assumptions~\ref{assumption-f-i-s} and~\ref{assumption-bdd-gradients}.
  Thus, D--NC also has a better rate.
 % The method in~\cite{duchi}
%  performs very similarly to~\cite{nedic_T-AC} (or slightly worse) in our simulations.
%We do so by constructing a hard problem for the algorithms in~\cite{nedic_T-AC,duchi} in the class of interest. Subsection~\ref{subsection-set-up} sets up the analysis,
% Subsection~\ref{subsection-hard-example-nedic} considers the algorithm in~\cite{nedic_T-AC},
% and subsection~\ref{subsection-hard-example-duchi} considers the algorithm in~\cite{duchi}.
%
%
%
%
%
%\subsubsection{Set-up and statement of the result}
%\label{subsection-set-up}
%
%

%We clarify mathematically the claim that we make.
Fix a generic, connected network $\mathcal{G}$ with $N$ nodes and $W$ that
obeys Assumption~\ref{assumption-network}.
Let $\mathcal{F}=\mathcal{F}(L,G)$ be the class
of all $N$-element sets of functions $\{f_i\}_{i=1}^N$, such that: 1)
each $f_i: {\mathbb R}^d \rightarrow \mathbb R$ is convex, has Lipschitz continuous derivative with constant $L$, and
 bounded gradient with bound $G$; and 2) Assumption~\ref{assumption-f-i-s}~(a) holds.
 Consider~\eqref{eqn-opt-prob-original} with
 $\{f_i\}_{i=1}^N \in \mathcal{F}$, for all $i$; consider D--NG
  with the step-size $\alpha_k=\frac{c}{(k+1)}$, $k=0,1,...$, $c\leq 1/(2L)$. Denote by:
\begin{eqnarray*}
&\,&\mathcal{E}^{\mathrm{D-NG}}\left( k,R\right) =\\
&\,&\sup_{\{f_i\}_{i=1}^N \in \mathcal{F}} \sup_{\{\overline{x}(0):\,\|\overline{x}(0)-{{x^\star}}\|\leq R\}} \max_{i=1,...,N}\left\{
f(x_i(k))-f^\star\right\}
\end{eqnarray*}
the optimality gap at the $k$-th iteration
of D--NG for the worst~$\{f_i\}_{i=1}^N \in \mathcal{F}$,
and the worst $\overline{x}(0)$ (provided~$\|\overline{x}(0)-{{x^\star}}\|\leq R$.)
  From Theorem~\ref{theorem-basic-convergence-result}~(a), \emph{for any $k=1,2,...$}:
$
  \mathcal{E}^{\mathrm{D-NG}}\left( k,R\right) \leq \mathcal{C} \frac{\log k}{k} = O(\log k/k),
$
 with $\mathcal{C}$ in~\eqref{eqn-constant-scaling}.
 %;
%  also, when $\mathcal{E}^{\mathrm{D-NG}}\left( k,R\right)$ is viewed as a
%  sequence in $k$:
%  \[
%  \mathcal{E}^{\mathrm{Nes}}\left( k,R\right) = O(\log k/k).
%  \]
 Now, consider the algorithm in~\cite{nedic_T-AC} with
 the step-size $\alpha_k=\frac{c}{(k+1)^\tau}$, $k=0,1,...,$ where $c \in [\,c_0,\,1/(2 L)\,]$, $\tau \geq 0$
  are the degrees of freedom, and $c_0$ is an arbitrarily small positive number.
  %(We constrain $c \leq 1/(2L)$, similarly as with~D--NG.)
   With this algorithm,
   $k=\mathcal{K}.$ We show that, for the $N=2$-node connected network,
   the weight matrix $W$ with $W_{ii}=7/8$, $i=1,2$, and $W_{ij}=1/8$, $i \neq j$ (which
    satisfies Assumption~\ref{assumption-network}), and $R=\sqrt{2}$, $L=\sqrt{2}$ and $G=10$, with~\cite{nedic_T-AC}:
    \begin{equation}
    \label{eqn-claim-nedic-hard}
    \inf_{\tau \geq 0, \,c \,\in\, [\,c_0,\,1/(2L)\,]} \mathcal{E} \left( k,R;\tau,c\right) = \Omega \left( \frac{1}{k^{2/3}}\right),
    \end{equation}
where
\begin{eqnarray*}
&\,&\mathcal{E} \left( k,R; \tau,c\right) \\
&=&
\sup_{\{f_i\}_{i=1}^N \in \mathcal{F}} \sup_{\{\overline{x}(0):\,\|\overline{x}(0)-{{x^\star}}\|\leq R\}} \max_{i=1,...,N}\left\{
f(x_i(k))-f^\star\right\}
\end{eqnarray*}
is the worst-case optimality gap when the step-size $\alpha_k=\frac{c}{(k+1)^\tau}$ is used.
%We omit the proof due to the lack of~space.
We perform the proof by constructing
a ``hard'' example of the functions $f_{i} \in \mathcal{F}(L,G)$ and
a ``hard'' initial condition to
upper bound $\mathcal{E} \left( k,R; \tau,c\right)$; for any fixed $k,c,\tau$, we set: $x_i(0)=:(1,0)^\top$, $i=1,2;$
 $f_{i}=:f_{i}^{{\theta}_k}$, where:
 {
 \small{
\begin{equation}
\label{eqn-f-i-s-hard}
 f^{{\theta}}_i(x) = \left\{ \begin{array}{llll}
 \frac{{{\theta}}(x^{(1)}+(-1)^i)^2}{2} +\frac{(x^{(2)}+(-1)^i)^2}{2} \\
 \mbox{ if ${{\theta}}(x^{(1)}+(-1)^i)^2 +
 (x^{(2)}+(-1)^i)^2 \leq \overline{\chi}^2$} \\
 \overline{\chi} \left( \left[ {{\theta}} (x^{(1)}+(-1)^i)^2 + (x^{(2)}+(-1)^i)^2 \right]^{1/2}-\frac{\overline{\chi}}{2}\right)\\
 \mbox{ else;}
       \end{array} \right.
\end{equation}}}
${\theta}_k=\frac{1}{\sum_{t=0}^{k-1}(t+1)^{-\tau}}$; and $\overline{\chi}=6.$
The proof of~\eqref{eqn-claim-nedic-hard} is in the Appendix. We convey here the underlying intuition.
When $\tau$ is $\epsilon$-smaller (away) from one, we show:
\[
\max_{i=1,2}(f^{{\theta}_k}(x_i(k))-f^{\star,{\theta}_k}) \geq \Omega\left( \frac{1}{k^{1-\tau}} + \frac{1}{k^{2\tau}} \right).
\]
The first summand is the ``optimization term,''
 for which a counterpart exists in the centralized gradient method also.
 The second, ``distributed problem'' term, arises because
 the gradients $\nabla f_i(x^\star)$ of the individual nodes functions are non-zero at the solution $x^\star.$
 Note the two opposing effects with respect to $\tau$: $\frac{1}{k^{1-\tau}}$ (the smaller $\tau \geq 0$, the better)
  and $\frac{1}{k^{2 \tau}}$ (the larger $\tau \geq 0$, the better.)
 To balance the opposing effects of the two summands,
 one needs to take a diminishing step-size; $\tau=1/3$ strikes the needed balance to give the $\Omega(1/k^{2/3})$ bound.
 \vspace{-0mm}
\subsection{Discussion on Assumptions}
\label{subsection-discussion-bounded-gradients}
 We now discuss what may occur if we drop each of the Assumptions made in our main results--Theorems~\ref{theorem-bound-distance-to-consensus} and~\ref{theorem-basic-convergence-result} for D--NG,
 and Theorems~\ref{theorem-distance-cons-consensus-alg} and~\ref{theorem-conv-rate-cons-alg-new} for D--NC.

\textbf{Assumption~\ref{assumption-network}~(a)}. Consider Theorems~\ref{theorem-bound-distance-to-consensus} and~\ref{theorem-distance-cons-consensus-alg}.
 If Assumption~\ref{assumption-network}~(a) is relaxed, then $\widetilde{x}(k)$
 with both methods may not converge to zero. Similarly, consider Theorems~\ref{theorem-basic-convergence-result} and~\ref{theorem-conv-rate-cons-alg-new}. Without Assumption~\ref{assumption-network}~(a),
  $f(x_i(k))$ may not converge to $f^\star$ at any node; e.g., take
  $N=2$, $W=I$, and $f_i$, $i=1,2,$ in the next paragraph.

%  Consider the following simple example: $N=2$, $W=I$;
% let $f_i:\mathbb R \rightarrow \mathbb R$, $i=1,2,$ obey Assumptions~\ref{assumption-f-i-s} and~\ref{assumption-bdd-gradients},
%  with the three quantities: $x_i^\star:=\mathrm{arg\,min}_{x \in \mathbb R} f_i(x)$, $i=1,2,$ and $x^\star=\mathrm{arg\,min}_{x \in \mathbb R} \left[\,f(x)=f_1(x)+f_2(x)\,\right]$ all unique and mutually different; set arbitrary intialization~$x(0)=y(0)$. Then, $x_i(k)$ converges to $x_i^\star$ (by convergence of the centralized Nesterov gradient method,) while $\widetilde{x}(k)$ and $f(x_i(k))-f^\star$, $i=1,2,$ converge respectively to the non-zero values: $(\frac{x_1^\star -x_2^\star}{2},\,\frac{x_2^\star-x_1^\star}{2})^\top$
%  and $f(x_i^\star)-f^\star$,~$i=1,2$.

\textbf{Assumption~\ref{assumption-network}~(b)} is imposed only for D--NG -- Theorems~\ref{theorem-bound-distance-to-consensus} and~\ref{theorem-basic-convergence-result}. We show by simulation that, if relaxed,
$\|\widetilde{x}(k)\|$ and $f(x_i(k))-f^\star$ may grow unbounded. Take $N=2$
 and $W_{11}=W_{22}=1/10$, $W_{12}=W_{21}=9/10$; the Huber losses $f_i:\mathbb R \rightarrow \mathbb R$, $f_i(x) =\frac{1}{2}(x-a_i)^2$ if
 $\|x-a_i\|\leq 1$ and $f_i(x)=\|x-a_i\|-1/2$ else, $a_i=(-1)^{i+1}$; $c=1$, and $x(0)=y(0)=(0,0)^\top$. Then, we verify by simulation~\cite{arxivVersion} that~$\|\widetilde{x}(k)\|$ and $\min_{i=1,2}(\,f(x_i(k))-f^\star\,)$ grow unbounded.

\textbf{Assumption~\ref{assumption-f-i-s}} is not needed for consensus with D--NG and D--NC (Theorems~\ref{theorem-bound-distance-to-consensus} and~\ref{theorem-distance-cons-consensus-alg}), but we impose it for
Theorems~\ref{theorem-basic-convergence-result} and~\ref{theorem-conv-rate-cons-alg-new} (convergence rates of D--NG and D--NC).
This Assumption is standard and widely present in the convergence analysis of gradient methods, e.g.,~\cite{Nesterov-Gradient}.
Nonetheless, we consider what may occur if we relax the requirement on the Lipschitz continuity of the gradient of the $f_i$'s.
For both D--NG and D--NC, we borrow the example functions $f_i: \mathbb R \rightarrow \mathbb R$,
$i=1,2,$ from~\cite{AnnieChenThesis},~pages~{29--31}: $f_1(x)=4 x^3+\frac{3 x^2}{2}$,
$x \geq 1$; $f_1(x) = \frac{15 x^2}{2}-2$, $x<1$; and $f_2(x):=f_1(-x).$
 Then, for D--NG with $W_{11}=W_{22}=1-W_{12}=1-W_{21}=9/10$, $c=1$,
 and $x(0)=y(0)=(-1,1)^\top,$ simulations show that $\|x(k)\|$ and $f(x_i(k))-f^\star$, $i=1,2,$ grow unbounded. Similarly, with D--NC,
 for the same $W$, $\alpha=0.1,$ and $x(0)=y(0)=(-1,1)^\top$,
  simulations show that $f(x_i(k))-f^\star$, $i=1,2$, stays away from zero when $k$ grows~\cite{arxivVersion}.

\textbf{Assumption~\ref{assumption-bdd-gradients}}.
First consider Theorems~\ref{theorem-basic-convergence-result} and~\ref{theorem-conv-rate-cons-alg-new}
 on the convergence rates of D--NG and D--NC. %
%As noted, for the convergence rate analysis of the distributed
%D--NG and D--NC methods, compared with
%the analysis of the centralized Nesterov gradient method~\cite{Nesterov-Gradient},
%we require additionally bounded gradients (Assumption~\ref{assumption-bdd-gradients}.)
%We show that, if Assumption~\ref{assumption-bdd-gradients} is dropped,
%the worst-case convergence rate of D--NG and D--NC is arbitrarily slow.
%
%Consider a hypothetical fusion center that knows all the $f_i$'s, $i=1,...,N$,
%and runs the centralized Nesterov gradient method. Define the
%class $\overline{\mathcal F}(L)$ to be the collection
%of all $N$-element sets of functions $\{f_i\}_{i=1}^N$,
%where each $f_i: {\mathbb R}^d \rightarrow \mathbb R$
% has Lipschitz  continuous gradient with constant $L$,
% and problem~\eqref{assumption-f-i-s}
% is solvable in the sense of Assumption~\ref{assumption-f-i-s}~(a).
%  Then, with the centralized Nesterov gradient method:
%  \[
%  \sup_{\{f_i\}_{i=1}^N \in \overline{\mathcal{F}}(L)}\:
%  \sup_{x(0):\|x(0)\| \leq R} \:\left( f(x(k)) -f^\star\right) \leq \frac{L R^2}{k^2}=O\left( \frac{1}{k^2}\right).
%  \]
%In words, the worst case optimality gap is $O(1/k^2)$,
%where the worst case is taken over all choices of the $\{f_i\}_{i=1}^N$
% and all initial conditions $x(0)$, provided that $\|x(0)\|\leq R.$
%
Define the class $\overline{\mathcal F}(L)$ to be the collection
of all $N$-element sets of convex functions $\{f_i\}_{i=1}^N$,
where each $f_i: {\mathbb R}^d \rightarrow \mathbb R$
 has Lipschitz continuous gradient with constant $L$,
 and problem~\eqref{assumption-f-i-s}
 is solvable in the sense of Assumption~\ref{assumption-f-i-s}~(a). (Assumption~\ref{assumption-bdd-gradients} relaxed.)
With the D--NC for the $2$-node connected network,
arbitrary weight matrix $W$ obeying Assumption~\ref{assumption-network}~(a),
 and the step-size $\alpha=1/(2 L)$, we show for $L=1$, $R \geq 0$, that, for any $k \geq 10$ and arbitrarily
 large $M >0$:
 \begin{eqnarray}
 \label{eqn-D-NC-worst-case-lower-bnd}
 &\,&\mathcal{E}(k;R; \alpha=1/(2L)) =\\
 &\,& \hspace{-3mm}
 \sup_{\{f_i\} \in \overline{\mathcal{F}}(L=1)}\:
  \sup_{\overline{x}(0):\|\overline{x}(0)-x^\star\| \leq R}\:\max_{i=1,2} \left( f(x_i(k)) -f^\star\right) \geq M. \nonumber
 \end{eqnarray}
 Note that the above means $\mathcal{E}(k;R; \alpha=1/(2L))=+\infty$, $\forall k \geq 10$, $\forall R \geq 0.$ That is, no matter how large the (outer) iteration number $k$ is,
 the worst case optimality gap is still arbitrarily large.

We conduct the proof by making a ``hard'' instance
for $\{f_i\}_{i=1}^N$: for
a fixed $k,M$, we set $x_i(0)=y_i(0)=0$, $i=1,2$, $f_i:{\mathbb R} \rightarrow {\mathbb R}$,
to $f_i=f_i^{\theta(k,M)}$, where $\theta=\theta(k,M)=8 \sqrt{M} \,k^2$ and:
\begin{equation}
\label{eqn-hard-f-i-s-D-NC-Lb}
f_i^{\theta}(x)=\frac{1}{2}\left( x+(-1)^i\theta \right)^2,\:\:\:\:i=1,2,\:\:\:\:\theta>0.
\end{equation}
Similarly to D--NC,
with D--NG we show in~\cite{arxivVersion} that~\eqref{eqn-D-NC-worst-case-lower-bnd} also holds for the $2$-node connected
network, the symmetric $W$ with $W_{12}=W_{21}=1-W_{11}=1-W_{22} = \frac{1}{2}\left( 1-10^{-6}\right)$ (this $W$
 obeys Assumption~\ref{assumption-network}), $\alpha_k=c/(k+1)$, and $c = \frac{1}{4} \times 10^{-6}$.
The candidate functions are in~\eqref{eqn-hard-f-i-s-D-NC-Lb},
where, for fixed $k \geq 5$, $M>0$, $\theta(k,M)=8 \times 10^6\,k\,\sqrt{M}.$

We convey here the intuition why~\eqref{eqn-D-NC-worst-case-lower-bnd} holds
for D--NG and D--NC, while the proof is in the Appendix.
Note that the solution to~\eqref{eqn-opt-prob-original}
with the $f_i$'s in~\eqref{eqn-hard-f-i-s-D-NC-Lb} is $x^\star=0$,
while $x_i^\star:=\mathrm{arg\,min}_{x \in \mathbb R} f_i(x) = (-1)^{i+1}\theta$, $i=1,2.$
Making $x_1^\star$ and $x_2^\star$
to be far apart (by taking a large $\theta$), problem~\eqref{eqn-opt-prob-original}
 for D--NG and D--NC becomes ``increasingly difficult.''
 This is because the inputs
 to the disagreement dynamics~\eqref{eqn-recursion}
   $(I-J)\nabla F(y(k-1))=(I-J)y(k-1)-(-\theta,\theta)^\top$
    are arbitrarily large, even
    when $y(k-1)$ is close the solution~$y(k-1)\approx (0,0)^\top.$

Finally, we consider what occurs if we drop Assumption~\ref{assumption-bdd-gradients} with Theorems~\ref{theorem-bound-distance-to-consensus} and~\ref{theorem-distance-cons-consensus-alg}. We show with D--NG and the
above ``hard'' examples that $\|\widetilde{x}(k)\| \geq \frac{\sqrt{2}\,c\,\theta }{2\,k}$,
$\forall k \geq 5.$ Hence, $\|\widetilde{x}(k)\|$
 is arbitrarily large by choosing $\theta$ large enough.~(see~\cite{arxivVersion}.) Similarly, with D--NC: $\|\widetilde{x}(k)\| \geq
\frac{\alpha\,\theta\sqrt{2}}{4\,k^2}$, $\forall k \geq 10.$~(see Appendix~C and~\cite{arxivVersion}.)

\vspace{-0mm}
\section{Simulations}
\label{section-simulations}
We compare the proposed D--NG and D--NC algorithms with~\cite{nedic_T-AC,duchi,AnnieChen} on the logistic loss.
 Simulations confirm the increased convergence rates of D--NG and D--NC
with respect to~\cite{nedic_T-AC,duchi} and show a comparable performance with respect to~\cite{AnnieChen}.
More precisely, D--NG achieves an accuracy~$\epsilon$ faster than~\cite{nedic_T-AC,duchi}
 for all~$\epsilon$, while D--NC is faster than~\cite{nedic_T-AC,duchi} at least for $\epsilon \leq 10^{-2}$.
 With respect to~\cite{AnnieChen}, D--NG is faster for lower accuracies ($\epsilon$ in the range $10^{-1}$ to $10^{-4}-10^{-5}$),
 while~\cite{AnnieChen} becomes faster for high accuracies ($10^{-4}-10^{-5}$ and finer); D--NC
 performs slower than~\cite{AnnieChen}.
 %existing methods
 %for practical accuracies ($10^{-5}$ and coarser.) We provide a further comparison between D--NG and D--NC on the Huber loss example.
%

%
\textbf{Simulation setup}. We consider distributed learning via the logistic loss; see, e.g.,~\cite{BoydADMoM} for further details.
Nodes minimize the logistic loss:
$
f(x)=\sum_{i=1}^N f_i(x)=\sum_{i=1}^N \log\left(  1+e^{-b_{i} (a_{i}^\top x_1 + x_0)}\right),
$
where $x=(x_1^\top,x_2^\top)^\top$, $a_{i} \in {\mathbb R}^{2}$ is the node $i$'s feature vector, and
$b_{i} \in \{-1,+1\}$ is its class label. The functions $f_i: {\mathbb R}^d \mapsto {\mathbb R}$, $d=3$,
satisfy Assumptions~\ref{assumption-f-i-s} and~\ref{assumption-bdd-gradients}.
The Hessian~$\nabla^2 f(x) = \sum_{i=1}^N \frac{e^{-c_i^\top x}}{(1+e^{-c_i^\top x})^2} c_i c_i^\top$,
where $c_i=(b_i a_i^\top, b_i)^\top \in {\mathbb R}^3$.
A Lipschitz constant $L$ should satisfy $\|\nabla^2 f(x)\| \leq N L$,
$\forall x \in {\mathbb R}^d.$
Note that $\nabla^2 f(x) \preceq \frac{1}{4}\sum_{i=1}^N c_i c_i^\top$,
because $\frac{e^{-c_i^\top x}}{(1+e^{-c_i^\top x})^2}\leq 1/4$, $\forall x$.
We thus choose $L = \frac{1}{4 N} \left\| \sum_{i=1}^N c_i c_i^\top \right\| \approx 0.3053$.
 %None of the $f_i$'s is strongly convex on ${\mathbb R}^d$;
% $f$ is not strongly convex on ${\mathbb R}^d$ either;
% there exists a sequence $x_n$ with $\lim_{n \rightarrow \infty}\|x_n\| \rightarrow \infty$, such that
%  $\lim_{n \rightarrow \infty} \|\nabla^2 f(x_n)\|=0.$
%  However, for the numerical values that we considered,
%  $f$ was strongly convex in a neighborhood of the solution~$x^\star.$
%
%
%Problem~\eqref{eqn-sim-prob} fits the framework in~\eqref{eqn-opt-prob-original}.
% with
% \[
% f(y)=\sum_{i=1}^N f_i(y),\qquad\textrm{and}\qquad
% f_i(y)=\sum_{j=1}^J \mathcal{J}_{\mathrm{logis}} \left(  b_{ij} H_y(a_{ij})\right) + \frac{\mathcal{R}}{N} \|y_1\|^2
% \]
%
%\textbf{Data}.
We generate $a_{i}$ independently over $i$; each entry is drawn from the standard normal distribution. We generate the ``true'' vector ${{x^\star}}=({x_1^\star}^\top, {x_0^\star})^\top$ by drawing its entries independently from the standard normal distribution. The labels are
$
b_{i}=\mathrm{sign}
\left( {{{x^\star}}_1}^\top a_{i}+{{x^\star}}_0+\epsilon_{i}\right),
$
where the $\epsilon_{i}$'s are drawn independently from a normal distribution with zero mean and variance~3.
%
%\textbf{Network}.
The network is a geometric network: nodes are
placed uniformly randomly on a unit square and the nodes whose distance is
less than a radius are connected by an edge. There are $N=100$ nodes, and the relative degree~$\left(=\frac{\mathrm{number\,\,of\,\,links}}{N(N-1)/2}\right) \approx 10\%$.
%
%\textbf{Algorithm parameters and metrics}.
We initialize
all nodes by $x_i(0)=0$ (and $y_i(0)=0$ with D--NG, D--NC, and~\cite{AnnieChen}). With all algorithms except D--NG,
we use the Metropolis weights~$W$~\cite{BoydFusion}; with D--NG,
we use $W^\prime=\frac{1+{\eta}}{2} I + \frac{1-{\eta}}{2} W$, with ${\eta}=0.1.$
%'s based on the neighbors' degrees as mentioned in Section~\ref{section-problem-model}.
%With
%D--NC,~\cite{nedic_T-AC},
% and~\cite{duchi}, we use the Metropolis weights (that require the same knowledge to be set.)
 The step-size $\alpha_k$ is: $\alpha_k=1/(k+1)$, with~D--NG;
$\alpha=1/(2L)$ and $1/L$, with D--NC;
%\footnote{Our theoretical analysis allows for
%$\alpha \leq 1/(2L)$, but we also simulate D--NG with $\alpha=1/L$.};
 $1/L$, with~\cite{AnnieChen} (both the 1st and 2nd algorithm variants -- see Subsection~\ref{subsection-comparisons}); and $1/(k+1)^{1/2}$, with~\cite{nedic_T-AC} and~\cite{duchi}.
\footnote{With~\cite{nedic_T-AC,duchi}, $\alpha_k=1/(k+1)^{p}$ and $p=1/2$,
gave the best simulation performance among the choices~$p \in \{1/3,1/2,1\}$.}
%We
%Note that we do not optimize the constant in the step-sizes with~\cite{nedic_T-AC,duchi} nor with D--NG
%   -- they operate without the prior knowledge of $L,\mu(N)$.
    We simulate the normalized (average) error~$\frac{1}{N} \sum_{i=1}^N \frac{f(x_i)-f^\star}{ f(x_i(0))-f^\star}$ versus the total number of communications at all nodes ($=N \mathcal{K}$.)

\textbf{Results}. Figure~1~(top) compares D--NG, D--NC (with step-sizes $\alpha=1/(2L)$ and $1/L$),~\cite{nedic_T-AC,duchi},~\cite{AnnieChen} (both 1st and 2nd variant with $\alpha=1/L$.) We can see that D--NG converges faster than other methods for
accuracies~$\epsilon$ in the range $10^{-1}$ to~$3\cdot 10^{-5}$.
For example, for~$\epsilon = 10^{-2}$, D--NG requires about $10^4$ transmissions;~\cite{AnnieChen}~(2nd variant) $\approx  3.16 \cdot 10^4$;
 D--NC ($\alpha=1/L$) $\approx 4.65 \cdot 10^4$,
 and D--NC with $\alpha=1/(2L)$ $\approx 1.1 \cdot 10^5$;
 and~\cite{AnnieChen}~(1st variant),~\cite{nedic_T-AC},~and~\cite{duchi} -- at least $\approx 1.3 \cdot 10^5$. For high accuracies, $2 \cdot 10^{-5}$ and finer,~\cite{AnnieChen} (2nd variant)
 becomes faster than D--NG. %; likewise, for $10^{-6}-10^{-7}$ and finer, D--NC becomes faster than D--NG.
  Finally,~\cite{AnnieChen}~(2nd) converges faster than D--NC, while~\cite{AnnieChen}~(1st) is slower than D--NC.
%We also note that~\cite{AnnieChen}~(2nd) is faster than D--NC; this is partly due to the smm

\textbf{Further comparisons of D--NG and D--NC: Huber loss}. We provide an additional
 experiment to further compare
the D--NG and D--NC methods. We show
that the relative performance of D--NC with respect to D--NG
improves when the instance of~\eqref{eqn-opt-prob-original} becomes easier (in the sense explained below.)
We consider a $N=20$-node geometric network with~$\frac{\mathrm{number\,of\,links}}{N(N-1)/2}\approx 32\%$ and Huber losses $f_i:\mathbb R \rightarrow \mathbb R$,
$f_i(x)= \frac{1}{2}\|x-a_i\|^2$ if $\|x-a_i\|\leq 1$,
 and $f_i(x) = \|x-a_i\|-1/2$, else, with $a_i \in \mathbb R$.
  We divide the set of nodes in two groups. For the first group,
  $i=1,...,6$, we generate the $a_i$'s as $a_i=\theta+\nu_i$, where $\theta>0$ is the ``signal''
   and $\nu_i$ is the uniform noise on $[-0.1 \theta,\,0.1\theta].$
   For the second group, $i=7,...,20$, we set $a_i=-\theta+\nu_i$, with the $\nu_i$'s from
   the same uniform distribution. Note that any $x_1^\star \in \mathrm{arg\,min}_{x \in \mathbb R}\sum_{i=1}^6 f_i(x)$ is in $[0.9\theta,\,1.1\theta]$, while
    any $x_2^\star \in \mathrm{arg\,min}_{x \in \mathbb R}\sum_{i=7}^{20} f_i(x)$
    lies in $[-1.1\theta,\,-0.9\theta]$. Intuitively, by making
    $\theta>0$ large, we increase the problem difficulty.
    For a small $\theta$, we are in the ``easy problem'' regime,
    because the solutions $x_1^\star$ and $x_2^\star$
     of the two nodes' groups are close; for a large $\theta$, we are in the ``difficult problem'' regime.
%    Our goal is to see the relative comparison
%    between D--NG and D--NC when we vary $\theta$.
    Figure~1~(bottom)
     plots the normalized average error versus $N \mathcal K$ for
     $\theta \in \{0.01;10;1000\}$ for
     D--NG with $\alpha_k=1/(k+1)$, D--NC with $\alpha=1/L$,
     while both algorithms are initialized by $x_i(0)=y_i(0)=0$.
     We can see that, with D--NC,
     the decrease of $\theta$ makes the convergence faster,
     as expected. (With D--NG, it is not a clear ``monotonic'' behavior.)
      Also, as $\theta$ decreases (``easier problem''),
      the performance of D--NC relative do D--NG improves.
      For $\theta=0.01$, D--NG is initially better, but the curves of D--NG and D--NC
      intersect at the value about~$4 \cdot 10^{-3}$,
      while for $\theta=1000$, D--NG is better for all accuracies as fine as (at least)~$10^{-7}$.
      \vspace{-0mm}
    \begin{figure}[thpb]
      \centering
       \includegraphics[height=2.4 in,width=3.55 in]{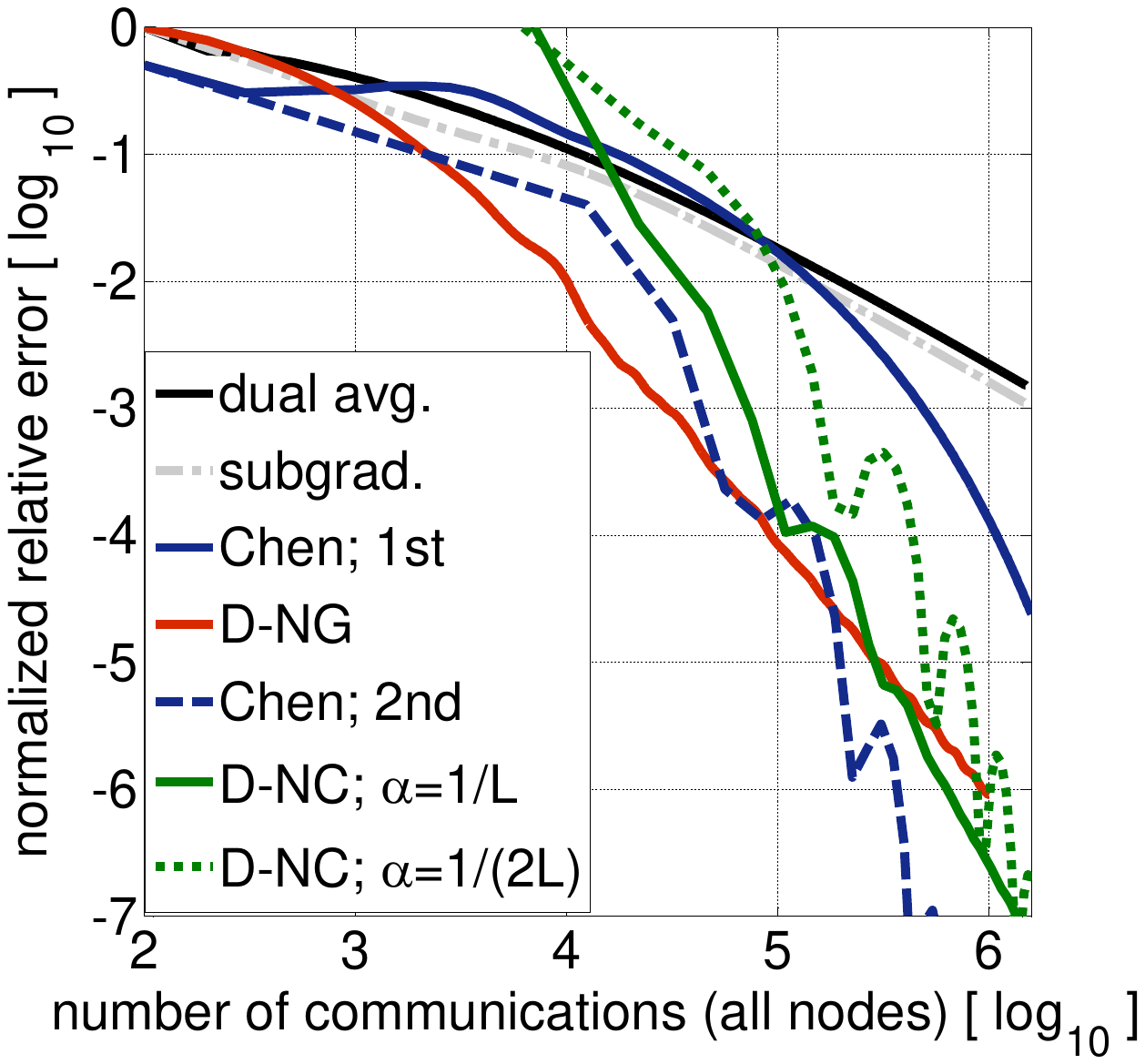}
       \includegraphics[height=2.4 in,width=3.1 in]{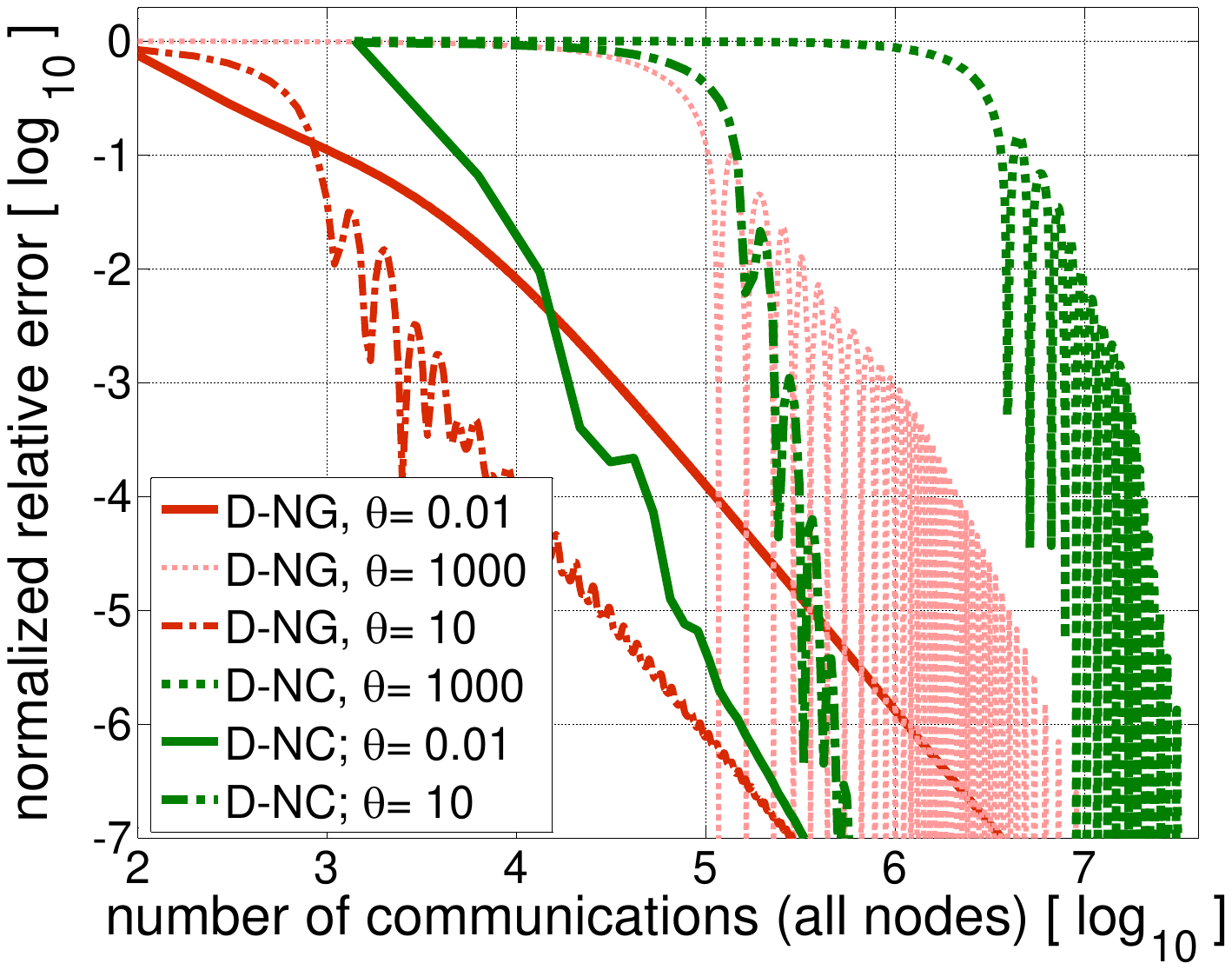}
       \vspace{-0mm}
       \caption{Normalized (average) relative error
       $\frac{1}{N}\sum_{i=1}^N\frac{f(x_i)-f^\star }{f(x_i(0))-f^\star}$
      versus the number of communications (all nodes) $N \mathcal{K}$; \textbf{Left}: Logistic loss;
       \textbf{Right}: Huber loss.}
       \vspace{-0mm}
      \label{figure-1}
\end{figure}

 We give an intuition on the observed behavior.
        Consider an ``easy'' problem with very similar local costs (small $\theta$). In such scenario,
         D--NC over outer iterations~$k$ behaves very similarly to the exact centralized Nesterov gradient method
          \emph{with a constant step-size}~$\alpha$. However, during each~$k$, D--NC uses $\tau_x(k)+\tau_y(k)$ per-node communications which, for the ``easy'' problem, are unnecessary and ``waste'' resources. (These communications are necessary for ``difficult'' problems.)
           Hence, D--NC behaves here as the centralized Nesterov gradient method slowed (re-scaled) through (unnecessary) multiple consensus rounds.
             From the above, it may seem intuitive that the relative performance of D--NC over D--NG
            is poorer for ``easy'' problems due to ``wastes'' in communications; but this does not
            occur in simulations.
            To explain why, consider now D--NG for the same ``easy'' problem. It behaves over~$k$ similarly to the exact centralized Nesterov gradient method \emph{with a diminishing step-size}~$1/k$. Hence,
            not only D--NC behaves as a suboptimal centralized gradient method (due to multiple consensus rounds),
            but also D--NG does, with the source of sub-optimality being the diminishing step-size~$1/k$.
              An intuitive  comparison of these two suboptimal methods on ``easy'' problems is the following.
             For a given network (given~$\mu(W)$), it is natural to expect
            that D--NC converges at a faster rate (steeper slope) than D--NG, but with the curve ``shifted''
             upwards due to the effect of~$\tau_x(k)+\tau_y(k)$. We indeed observe such behavior in Figure~1, bottom, case $\theta=0.01$.
              On the other hand, for ``difficult'' problems (large~$\theta$), the dynamics of disagreements play a significant role and cannot be neglected. Hence, it is much harder to intuitively understand the behavior. As our simulation example indicates, for more ``difficult'' problems (larger $\theta$), the performance of D--NC relative to D--NG actually deteriorates.
              We also performed a simulation with a deteriorated $\mu(W)$,
             while all other parameters are the same as in the above simulation. We increase $\mu(W)$
              by setting, with both D--NG and D--NC, $W^{\prime \prime}=0.9 I + 0.1 W$,
              where $W$ is the Metropolis matrix.
             The relative behavior of D--NC with respect to D--NG still deteriorates with the increase of~$\theta$. (Figure omitted due to lack of space.)

 % The $f_i$'s
 %obey Assumptions~\ref{assumption-f-i-s} and~\ref{assumption-bdd-gradients},
 %and no $f_i$ is strongly convex, nor is~$f(x)=\sum_{i=1}^N f_i(x).$

\vspace{-0mm}
\section{Conclusion}
\label{section-conclusion}
We propose fast distributed gradient algorithms when the nodes in a network minimize
the sum of their individual cost functions. Existing literature has presented distributed gradient based
 algorithms to solve this problem and
 has studied their convergence rates, for a class
 of convex, non-differentiable costs, with bounded
 gradients.
 %Further, for constrained problems,
% (Lipschitz~$f_i$'s over the constraint set -- for constrained problems.)
  We asked whether faster convergence rates than the
 rates established in the literature can be achieved for more structured costs -- convex, with Lipschitz continuous gradient (with constant $L$) and
 bounded gradient. Building from the centralized Nesterov gradient method,
 we answer affirmatively this question by proposing two distributed gradient algorithms.
Our algorithm D--NG achieves the rates~$O\left(\frac{\log \mathcal{K}}{\mathcal{K}}\right)$
 and~$O\left( \frac{\log k}{k}\right)$. Our algorithm D--NC operates only if~$L$ and~$\mu(W)$ are available and achieves
rates~$O\left(  \frac{1}{\mathcal{K}^{2-\xi}}\right)$ and~$O\left( \frac{1}{k^2}\right)$.
We also found convergence constants in terms of the network parameters.
%, and that,
%under a fixed accuracy $\epsilon$, the communication cost with our methods scales linearly with the inverse of the spectral gap~$1/(1-\mu(W))$.
Simulations illustrate the performance of the proposed methods.% Simulations illustrate our findings.

\section*{Acknowledgment}
We thank an anonymous reviewer whose instructive comments led us to develop algorithm D--NC. We also thank the anonymous reviewers and the associate editor for several useful suggestions regarding the presentation and organization of the paper. We thank as well Jo\~ao F. C. Mota for pointing us to relevant references and for useful discussions.

\vspace{-0mm}
\section*{Appendix}
\vspace{-0mm}
\subsection{Proof of Theorem~\ref{theorem-bound-distance-to-consensus}}
\label{subsection-appendix-proof-consensus}
For notational simplicity, we let $d=1$, but the proof extends to $d>1.$
We outline the main steps in the proof. First, we unwind
the recursion~\eqref{eqn-recursion} and
calculate the underlying time varying system matrices.
Second, we upper bound the norms of the time varying system matrices.
 Finally, we use these bounds and a summation argument to complete the proof of the~Theorem.
\vspace{-0mm}
%\subsubsection{Recursion for $(\widetilde{x}(k)^\top, \widetilde{x}(k-1)^\top)^\top$}
%Denote by:
%\begin{eqnarray}
%u(k) := -(I-J)\nabla F(y(k)) ,\:k=0,1,2,...
%\end{eqnarray}
%
%
\subsubsection{Unwinding~\eqref{eqn-recursion} and calculating the system matrices}
Define the $2N \times 2N$ system matrices:
\begin{equation}
\label{eqn-phi-k-t}
\Phi(k,t) := \Pi_{s=2}^{k-t+1}  \left[ \begin{array}{cc}
(1+\beta_{k-s}) \widetilde{W} & -\beta_{k-s} \widetilde{W} \\
I & 0
 \end{array} \right],\:k>t,
\end{equation}
and $\Phi(k,k)=I.$
Unwinding~\eqref{eqn-recursion}, the solution to~\eqref{eqn-recursion} is:
\begin{eqnarray}
\label{eqn-ol2}
&\,&(\widetilde{x}^\top(k),\,\widetilde{x}^\top(k-1))^\top = \sum_{t=0}^{k-1} \Phi(k,t+1) \, \alpha_{t}\, \\
&\,&\hspace{11mm}\times(\, (\,-\nabla F(y(t))\,)^\top (I-J),0\,)^\top, \:k=1,2,...\nonumber
\end{eqnarray}
%
%where
%\[
%u(t) =.
%\]
%where the system matrices $\Phi(k,t)$ are, recursively, for a fixed $t=0,1,...,k-1,$ defined by:
%
%
%\begin{eqnarray}
%\label{eqn-P-k-t-def}
%\Phi(k,t) &=&  (1+\beta_{k-1}) \widetilde{W} \Phi(k-1,t) - \beta_{k-1} \widetilde{W} \Phi(k-2,t),\:k\geq t+1 \\
%\Phi(t,t) &=& I,\:\:  \Phi(t-1,t) = 0. \nonumber
%\end{eqnarray}
%
%
We now show the interesting structure of the matrix $\Phi(k,t)$ in~\eqref{eqn-phi-k-t}
by decomposing it into the product of an
orthonormal matrix $U$, a block-diagonal matrix,
and $U^\top$. While $U$ is independent of $k$ and $t$, the block diagonal matrix
depends on $k$ and $t$, and has $2 \times 2$ diagonal blocks. Consider the matrix in~\eqref{eqn-recursion} with $k-2=t$, for a generic $t=-1,0,1,...$
Using~$\widetilde{W} =Q \widetilde{\Lambda} Q^\top$:
\begin{eqnarray}
&\,&\left[ \begin{array}{cc}
(1+\beta_{t}) \widetilde{W} & -\beta_{t} \widetilde{W} \\
I & 0
 \end{array} \right] \nonumber \\
 &=&
%  (Q \oplus Q)\, \left[ \begin{array}{cc}
%(1+\beta_{t}) \widetilde{\Lambda} & -\beta_{t} \widetilde{\Lambda} \\
%I & 0
% \end{array} \right]\,(Q \oplus Q)^\top \nonumber \\
%\]
%%
%%
%We can further express the matrix above as:
%\begin{equation}
\label{eqn-using-ova}
%&=&
(Q \oplus Q) \, P \left(\oplus_{i=1}^N \Sigma_i(t)\right) P^\top (Q \oplus Q)^\top,
\end{eqnarray}
where $P$ is the $2 N \times 2 N$ permutation matrix ($e_i$ here is the
$i$--th column of the $2 N \times 2 N$ identity matrix)
 $
 P = \left[ e_1, e_{N+1}, e_2, e_{N+2},..., e_{N}, e_{2N}\right]^\top,
 $
and $\Sigma_i(t)$ is a
 $2 \times 2$ matrix with
 $=(1+\beta_{t}) \lambda_i(\widetilde{W})$,
  $\left[ \Sigma_i(t)\right]_{12}=-\beta_{t} \lambda_i(\widetilde{W})$,
   $\left[ \Sigma_i(t)\right]_{21}=1, $ and $\left[ \Sigma_i(t)\right]_{22}=0.$
Using \eqref{eqn-using-ova},
and the fact that $(Q \oplus Q) P$ is orthonormal:
$\left((Q \oplus Q)P \right) \cdot \left((Q \oplus Q) P\right)^\top = (Q \oplus Q) P P^\top (Q\oplus Q)^\top
= (Q Q^\top) \oplus (QQ^\top)=I$, we can express $\Phi(k,t)$ in~\eqref{eqn-phi-k-t} as:
\begin{eqnarray}
\label{eqn-Phi-k-t-simplified}
&\,&\Phi(k,t) :=
\\
&\,&(Q \oplus Q) P \left(\oplus_{i=1}^N \Pi_{s=2}^{k-t+1} \Sigma_i(k-s) \right)P^\top (Q \oplus Q)^\top \nonumber \\
&\,&
 \mathrm{for}\:\:k>t;\:\:\Phi(k,k)=I. \nonumber
\end{eqnarray}

\subsubsection{Bounding the norm of $\Phi(k,t)$} As $(Q \oplus Q) P$ is orthonormal, $\Phi(k,t)$ has the same singular values as $\oplus_{i=1}^N \Pi_{s=2}^{k-t+1} \Sigma_i(k-s)$,
  and so these two matrices also share the same spectral norm (maximal singular value.)
  Further, the matrix $\oplus_{i=1}^N \Pi_{s=2}^{k-t+1} \Sigma_i(k-s)$
  is block diagonal (with $2\times 2$ blocks $\Pi_{s=2}^{k-t+1} \Sigma_i(k-s)$), and so:
  $
  \|\Phi(k,t)\| = \max_{i=1,...,N} \left\|  \Pi_{s=2}^{k-t+1} \Sigma_i(k-s)    \right\|.
  $
  We proceed by calculating $\left\|  \Pi_{s=2}^{k-t+1} \Sigma_i(k-s)    \right\|$.
  We distinguish two cases: $i=1$, and $i>1.$
   % We distinguish two cases: $k>t+1$, and $k\leq t+1$.

\textbf{Case $i=1$}. As $\lambda_1(\widetilde{W})=0$, for all $t$,
 $\Sigma_1(t)=\Sigma_1$ is a constant matrix, with $
    \left[\Sigma_1\right]_{21} = 1,
$ and the entries $(1,1)$, $(1,2)$ and $(2,2)$ of $\Sigma_1$ are zero.
Note that $\|\Sigma_1\|=1$, and $(\Sigma_1)^s=0$, $s \geq 2$. Thus, as long as $k>t+1$, the product $ \Pi_{s=2}^{k-t+1} \Sigma_i(k-s) = 0$, and so:
\begin{equation}
\label{case-i=1}
\left\|  \Pi_{s=2}^{k-t+1} \Sigma_1(k-s)    \right\| =
\left\{ \begin{array}{rl}
 1 &\mbox{ if $k=t+1$} \\
 \phantom{-} 0 &\mbox{ if $k>t+1.$}
       \end{array} \right.
\end{equation}
%
% $\left\|  \Pi_{s=2}^{k-t+1} \Sigma_1(k-s)    \right\|=0,$
%  for $k>t+1$; and $\left\|  \Pi_{s=2}^{k-t+1} \Sigma_1(k-s)    \right\| = 1$, for $k=t+1$.

\textbf{Case $i>1$}. To simplify notation, let $\lambda_i: = \lambda_i(\widetilde{W})$, and recall $\lambda_i \in (0,1);$ $\Sigma_i(t)$~is:
$
   \Sigma_i(t) = \widehat{\Sigma}_i - \frac{3}{t+3} \Delta_i$, where: 1)
    $[\widehat{\Sigma}_i]_{11}=2 \lambda_2$, $[\widehat{\Sigma}_i]_{12}=-\lambda_i,$
     $[\widehat{\Sigma}_i]_{21}=1$, and $[\widehat{\Sigma}_i]_{22}=0;$
      and 2) $[\Delta_i]_{11}=-[\Delta_i]_{12}=\lambda_i$,
      and $[\Delta_i]_{21}=[\Delta_i]_{22}=0.$
; $\widehat{\Sigma}_i$ is diagonalizable, with $\widehat{\Sigma}_i =  \widehat{\mathcal{Q}_i}  \widehat{\mathcal{D}}_i \widehat{\mathcal{Q}_i}^{-1}$, and:
{
{
\small{
\begin{eqnarray*}
\widehat{\mathcal{Q}_i} &=& \left[ \begin{array}{cc}
    {{\lambda_i}} + \mathbf{j} \sqrt{ {{\lambda_i}}(1- {{\lambda_i}})} & {{\lambda_i}}-\mathbf{j} \sqrt{ {{\lambda_i}}(1- {{\lambda_i}})} \\
    1 & 1
    \end{array} \right] \\
    \widehat{\mathcal{D}}_i & =&
    \left[ \begin{array}{cc}
    {{\lambda_i}} + \mathbf{j} \sqrt{ {{\lambda_i}}(1- {{\lambda_i}})} & 0 \\
    0 &  {{\lambda_i}} - \mathbf{j} \sqrt{ {{\lambda_i}}(1- {{\lambda_i}})}
    \end{array} \right].
\end{eqnarray*}
}
}
}
%Also, $[\widehat{\mathcal{Q}_i}^{-1}]_{11}=
%-[\widehat{\mathcal{Q}_i}^{-1}]_{21}=(2 \mathbf j \sqrt{\lambda_i(1-\lambda_i)})^{-1};$
%$[\widehat{\mathcal{Q}_i}^{-1}]_{12}=$
% \frac{1}{2 \mathbf{j} \sqrt{\lambda_i(1-\lambda_i)}}
%\left[ \begin{array}{cc}
%    1 & -{{\lambda_i}}+\mathbf{j} \sqrt{ {{\lambda_i}}(1- {{\lambda_i}})} \\
%    -1 & {{\lambda_i}}+\mathbf{j} \sqrt{ {{\lambda_i}}(1- {{\lambda_i}})}
%    \end{array} \right]}}$.
(Note that the matrices $\widehat{\mathcal{Q}}_i $ and $ \widehat{\mathcal{D}}_i$ are complex.)
Denote by
$
{{\mathcal{D}_i}}(t) = \widehat{\mathcal{D}}_i - \frac{3}{t+3} \widehat{\mathcal{Q}_i}^{-1} \Delta_i \widehat{\mathcal{Q}_i}.
$
%Then, $\mu(k)$ is re-written as:
%\[
%\mu(k) = \widehat{\mathcal{Q}_i} \mathcal{D}(k-1) \mathcal{D}(k-2)...\mathcal{D}(t) \widehat{\mathcal{Q}_i}^{-1}(1,0)^\top.
%\]
%Applying sub-multiplicative property of norms:
Then,
$
\Sigma_i(t) = \widehat{\mathcal{Q}}_i \left( \widehat{\mathcal{D}}_i -\frac{3}{t+3} \widehat{\mathcal{Q}}_i^{-1}  \Delta_i \widehat{\mathcal{Q}}_i\right)
 \widehat{\mathcal{Q}}_i^{-1} = \widehat{\mathcal{Q}_i} {{\mathcal{D}_i}}(t) \widehat{\mathcal{Q}_i}^{-1}.
$
By the sub-multiplicative property of norms, and using
$\left\|\widehat{\mathcal{Q}_i} \right\| \leq    \sqrt{2} \left\| \widehat{\mathcal{Q}_i} \right\|_{\infty}  = 2\sqrt{2}$,
$\left\|\widehat{\mathcal{Q}_i}^{-1} \right\| \leq \sqrt{2} \left\| \widehat{\mathcal{Q}_i}^{-1} \right\|_{\infty} = \frac{2 \sqrt{2}}{\sqrt{{{\lambda_i}}(1-{{\lambda_i}})}}$:
\begin{equation}
\label{eqn-app}
\| \Pi_{s=2}^{k-t+1} \Sigma_i(k-s) \| \leq \frac{8}{\sqrt{\lambda_i(1-\lambda_i)}} \, \Pi_{s=2}^{k-t+1}
\|{{\mathcal{D}_i}}(k-s)\| .
\end{equation}
%
%Further, the norms of $\widehat{\mathcal{Q}_i}$ and $\widehat{\mathcal{Q}_i}^{-1}$ are upper bounded as:
%%
%%
%\begin{align}
%\label{eqn-bnd-Q-o}
%\left\|\widehat{\mathcal{Q}_i} \right\| {}&\leq    \sqrt{2} \left\| \widehat{\mathcal{Q}_i} \right\|_{\infty}  = 2\sqrt{2},\:
%\left\|\widehat{\mathcal{Q}_i}^{-1} \right\| \leq \sqrt{2} \left\| \widehat{\mathcal{Q}_i}^{-1} \right\|_{\infty} = \frac{2 \sqrt{2}}{\sqrt{{{\lambda_i}}(1-{{\lambda_i}})}}.
%\end{align}
%
%
% the matrix $\widehat{\Sigma}_i -\frac{3}{k+2} \Delta$ equals:
%%
%\begin{eqnarray*}
%%\label{eqn-phi-prime}
%\widehat{\Sigma}_i -\frac{3}{k+2} \Delta = \mathcal{Q}  \left(  \mathcal{D} - \frac{3}{k+2}  \mathcal{Q}^{-1} \Delta \mathcal{Q} \right)  \mathcal{Q}^{-1},
%\end{eqnarray*}
%where:
%
%
%
It remains to upper bound $\|{{\mathcal{D}_i}}(t)\|$, for all $t=-1,0,1,...$ We will show that
\begin{equation}
\label{eqn-bnd-D-t}
\|{{\mathcal{D}_i}}(t)\| \leq \sqrt{\lambda_i}, \:\:\forall t=-1,0,1,...
\end{equation}
 Denote by $a_t=\frac{3}{t+3}$, $t=0,1,...,$ and $a_{-1}=1$.
 After some algebra, the entries of $\mathcal{D}_i(t)$ are:
 $
 [{{\mathcal{D}_i}}(t)]_{11} = \left( [{{\mathcal{D}_i}}(t)]_{22} \right)^H
 =
 \frac{1}{2}(2-a_t)(\lambda_i+\mathbf{j}\sqrt{\lambda_i(1-\lambda_i)})$,
 $[{{\mathcal{D}_i}}(t)]_{12} = \left(
 [{{\mathcal{D}_i}}(t)]_{21} \right)^H = a_t(\lambda_i + \mathbf{j} \sqrt{\lambda_i(1-\lambda_i)}),$
 %
%  &
% a_t(\lambda_i - \mathbf{j} \sqrt{\lambda_i(1-\lambda_i)})\\
% a_t(\lambda_i + \mathbf{j} \sqrt{\lambda_i(1-\lambda_i)}) & (2-a_t)(\lambda_i-\mathbf{j} \sqrt{\lambda_i(1-\lambda_i)}) \end{array} \right],
% \]
%
%
%
which gives:
$[{\mathcal{D}_i(t)}^H {\mathcal{D}_i}(t)]_{11}=
[{\mathcal{D}_i(t)}^H {\mathcal{D}_i}(t)]_{22}= \frac{a_t^2 + (2-a_t)^2}{4} \lambda_i$,
 and $[{\mathcal{D}_i(t)}^H {\mathcal{D}_i}(t)]_{12}=\left( [{\mathcal{D}_i(t)}^H {\mathcal{D}_i}(t)]_{21}\right)^H
  =  \frac{a_t(2-a_t)}{2} \left( 2\lambda_i^2-\lambda_i - 2 \mathbf{j} \lambda_i \sqrt{\lambda_i(1-\lambda_i)}\right).$
%
%\left[ \begin{array}{cc}
% \frac{a_t^2 + (2-a_t)^2}{4} \lambda_i&
% \frac{a_t(2-a_t)}{2} \left( 2\lambda_i^2-\lambda_i - 2 \mathbf{j} \lambda_i \sqrt{\lambda_i(1-\lambda_i)}\right)\\
% \frac{a_t(2-a_t)}{2} \left( 2\lambda_i^2 -\lambda_i^2+ 2\mathbf{j} \lambda_i \sqrt{\lambda_i(1-\lambda_i)}\right) &
% \frac{a_t^2 + (2-a_t)^2}{4} \lambda_i
% \end{array} \right].
% \]
% }
% }
%
%
%
Next, very interestingly:
$
\|{{\mathcal{D}_i}}^H(t)  {{\mathcal{D}_i}}(t)\|_1 = \left\| \left[ {{\mathcal{D}_i}}^H(t)  {{\mathcal{D}_i}}(t)\right]_{11} \right\|
 + \left\| \left[ {{\mathcal{D}_i}}^H(t)  {{\mathcal{D}_i}}(t)\right]_{12} \right\|
 = \frac{1}{4} (a_t^2+(2-a_t)^2) \lambda_i + \frac{1}{2} a_t(2-a_t) \lambda_i ,
 = \lambda_i.
$
 for any $a_t \in [0,2]$, which is the case here
 because $a_t=3/(t+3)$, $t=0,1,...$, and $a_{-1}=0.$
Thus, as $\|A\| \leq \|A\|_1$ for a Hermitean matrix $A$:
$
\|{{\mathcal{D}_i}}(t)\| = \sqrt{\|{{\mathcal{D}_i}}^H(t)  {{\mathcal{D}_i}}(t)\|}
\leq
\sqrt{\|{{\mathcal{D}_i}}^H(t)  {{\mathcal{D}_i}}(t)\|_1}=\sqrt{\lambda_i}.
$
Applying the last equation and \eqref{eqn-bnd-D-t} to \eqref{eqn-app},
we get, for $i \neq 1$:
$
\| \Pi_{s=2}^{k-t+1} \Sigma_i(k-s) \|
\leq
%\left\|\widehat{\mathcal{Q}_i}\right\| \, \left\|\widehat{\mathcal{Q}_i}^{-1}\right\|  \, \left( \sqrt{\lambda}_i\right)^{k-t}
%=
%\label{eqn-bound-Sigma-i}
\frac{8}{\sqrt{\lambda_i(1-\lambda_i)}} \left(\sqrt{\lambda_i}\right)^{k-t} ,\:k\geq t+1.
$
%\end{eqnarray}
%
%
Combine the latter with~\eqref{case-i=1}, and use
$\|\Phi(k,t)\| = \max_{i=1,...,N} \| \Pi_{s=2}^{k-t+1} \Sigma_i(k-s) \|$,
 Assumption~\ref{assumption-network}~(b) and $\lambda_N(\widetilde{W})=\mu(W)$, to~obtain:
\begin{eqnarray}
\| \Phi(k,t) \| &\leq& \frac{8 \left( \sqrt{\mu(W)}\right)^{k-t}}{\min_{i \in \{2,N\}} \sqrt{\lambda_i(\widetilde{W})(1-\lambda_i(\widetilde{W}))}} \nonumber \\
\label{eqn-bound-norm-P-k-t}
&\leq& \hspace{-2mm}\frac{8}{\sqrt{{\eta}(1-\mu(W))}}\left( \sqrt{\mu(W)}\right)^{k-t}\hspace{-2mm},\,k \geq t.
% \nonumber \\
%
%&=:& C_P \,  \left( \frac{k}{t}\right)^3 \, \left( \sqrt{\mu(W)}\right)^{k-t}.
\end{eqnarray}
%
%
%where the last inequality uses $\lambda_N(\widetilde{W})=\mu(W)$ and Assumption~\ref{assumption-network}~(b).
%
%
%
\subsubsection{Summation} We apply \eqref{eqn-bound-norm-P-k-t}
 to~\eqref{eqn-ol2}. Using the sub-multiplicative and sub-additive properties of norms, expression $\alpha_t=c/(t+1)$,
 and the inequalities
 $\|\widetilde{x}(k)\| \leq \left\|(\widetilde{x}(k)^\top, \widetilde{x}(k-1)^\top)^\top\right\|$,
 $\left\|(-(I-J)\nabla F(y(t))^\top,\,0^\top)^\top\right\| \leq \sqrt{N}\, {G}$:%, we have, for $t=0,1,...,k-1$, with $k \leq K$:
 {\allowdisplaybreaks{
\begin{eqnarray}
 \|\widetilde{x}(k)\|
% &\leq& \sum_{t=0}^{k-1} \frac{8 }{\min_{i \in \{2,N\}} \sqrt{\lambda_i(\widetilde{W})(1-\lambda_i(\widetilde{W}))}} \,
%  \left( \sqrt{\mu(W)}\right)^{k-(t+1)} \alpha_t \|u(t)\|
% \nonumber \\
 \label{eqn-x-tilde-bnd}
 &\leq&
\frac{8\, \sqrt{N}\, c \, {G}}{\sqrt{{\eta}(1-\mu(W))}} \\
&\times&
 \sum_{t=0}^{k-1}  \, \left(\sqrt{\mu(W)}\right)^{k-(t+1)} \frac{1}{(t+1)}. \nonumber
 \end{eqnarray}}}
We now denote by $r:=\sqrt{\mu(W)}\in (0,1)$. To complete the proof of the Lemma, we upper bound the sum $\sum_{t=0}^{k-1}r^{k-(t+1)}\frac{1}{(t+1)}$
 by splitting it into two sums. With the first sum, $t$ runs from zero to $\lceil k/2 \rceil$, while with the second sum, $t$ runs
 from $\lceil k/2 \rceil+1$ to $k:$
%\[
%\sum_{t=0}^{k-1}q^{k-(t+1)}\frac{1}{(t+1)^4}  =  \sum_{t=1}^{k}q^{k-t}\frac{1}{t^4}
%\]
%
%
%
%
%
%
%
%
{\allowdisplaybreaks{
\begin{eqnarray}
&\,&\sum_{t=0}^{k-1} \frac{r^{k-(t+1)}}{t+1} \hspace{-1mm} = \hspace{-1mm}
\left(r^{k-1} + r^{k-2}\frac{1}{2}+...+r^{\lceil k/2 \rceil }\frac{1}{\lceil k/2 \rceil } \right) \nonumber\\
&+& \left( r^{k - (\lceil k/2 \rceil +1)} \frac{1}{\lceil k/2 \rceil +1}  +...+\frac{1}{k} \right) \nonumber \\
&\leq&
%q^{k-1} + q^{k-2}\frac{1}{2^4}+...+q^{\lceil k/2 \rceil }\frac{1}{(\lceil k/2 \rceil )^4} +...+\frac{1}{k^4}\\
%&\leq&
r^{k/2} \left(1+\frac{1}{2}+...+\frac{1}{k/2}+ \frac{1}{(k+1)/2}\right) \nonumber \\
&+& \frac{1}{(k/2)} \left( 1+r+...+r^{k}\right) \nonumber \\
%\end{eqnarray*}
%
%
%
%Further:
%
%
%\begin{eqnarray}
\label{eqn-chain-1}
%\sum_{t=0}^{k-1}r^{k-(t+1)}\frac{1}{(t+1)}
&\leq&
r^{k/2} \left( \log(1+k/2) + 2 \right) + \frac{2}{k}\frac{1}{1-r}\\
\label{eqn-chain-2}
&=&
2 \, \left\{ r^{k/2} \log(1+k/2) (k/2) \right\}\frac{1}{k} +
\left\{ 4 r^{k/2} (k/2) \right\} \frac{1}{k}  \nonumber\\
&+&   \frac{2}{k}\frac{1}{1-r}\\
\label{eqn-chain-3}
&\leq&
2 \, \sup_{z\geq 1/2}\left\{ r^{z} \log(1+z) z \right\}\frac{1}{k} +
4 \sup_{z \geq 1/2} \left\{ r^{z} z \right\} \frac{1}{k}  \nonumber\\
&+&   \frac{2}{k}\frac{1}{1-r}\\
&\leq&
\label{eqn-chain-4}
\left( 2 \mathcal{B}(r) + \frac{4}{e(-\log r)}+\frac{2}{1-r}\right)\frac{1}{k}\\
&\leq&
\left( 2\mathcal{B}(r) + \frac{7}{1-r^2} \right)\frac{1}{k}.\nonumber
\end{eqnarray}}}
Inequality~\eqref{eqn-chain-1} uses
the inequality $1+\frac{1}{2}+...+\frac{1}{t} \leq \log t + 1,\,\:t=1,2,...,$
 and $1+r+...+r^k \leq \frac{1}{1-r}$;~\eqref{eqn-chain-2}
  multiplies and divides the first summand on the right hand side of~\eqref{eqn-chain-1} by $k/2$;~\eqref{eqn-chain-3}
   uses $r^{k/2} \log(1+k/2) (k/2) \leq \sup_{z \geq 1/2} r^z \log(1+z)z$, for all $k=1,2,...$,
   and a similar bound for the second summand in~\eqref{eqn-chain-2}; the left inequality in~\eqref{eqn-chain-4}
    uses $\mathcal{B}(r):=\sup_{z \geq 1/2}r^z \log(1+z) z$
     and $\sup_{z \geq 1/2} r^z \,z \leq \frac{1}{e\,(-\log r)}$ (note that $r^z\,z$ is convex in~$z$; we take the derivative
     of $r^z \,z$ with respect to $z$ and set it to zero);
      and the right inequality in~\eqref{eqn-chain-4} uses $-1/\log r \leq 1/(1-r)$, $\forall r\in [0,1)$;
$1/(1-r) \leq 2/(1-r^2)$, $\forall r \in [0,1)$, and $e=2.71...$
%
%
%
%$
%%
%%
%we have:
%
%The last inequality uses: $-1/\log r \leq 1/(1-r)$, $\forall r\in [0,1)$;
%$1/(1-r) \leq 2/(1-r^2)$, $\forall r \in [0,1)$, and $e=2.71...$
Applying the last to~\eqref{eqn-x-tilde-bnd}, and using the~$C_{\mathrm{cons}}$ in~\eqref{eqn-C-cons},
Theorem~\ref{theorem-bound-distance-to-consensus} for $\|\widetilde{x}(k)\|$ follows.
 Then, as $\widetilde{y}(k)=\widetilde{x}(k)+\frac{k-1}{k+2}(\widetilde{x}(k)-\widetilde{x}(k-1))$, we have that $\|\widetilde{y}(k)\| \leq 2 \|\widetilde{x}(k)\| + \|\widetilde{x}(k-1)\|$.
  Further, by Theorem~\ref{theorem-bound-distance-to-consensus}: $\|\widetilde{x}(k-1)\| \leq c \sqrt{N} G C_{\mathrm{cons}} \frac{1}{k-1}\frac{k}{k}
   \leq 2 c \sqrt{N} G C_{\mathrm{cons}} \frac{1}{k}$, $k \geq 2$,
   and $\|\widetilde{x}(0)\| =0$ (by assumption). Thus,
   $\|\widetilde{x}(k-1)\| \leq 2 c \sqrt{N} G C_{\mathrm{cons}} \frac{1}{k}$, $\forall k \geq 1.$
    Thus, $\|\widetilde{y}(k)\| \leq 2 \|\widetilde{x}(k)\| + \|\widetilde{x}(k-1)\|
    \leq 4 c \sqrt{N} G C_{\mathrm{cons}} \frac{1}{k}$, $\forall k \geq 1.$
\vspace{-0mm}
\subsection{Proof of the lower bound in~\eqref{eqn-claim-nedic-hard} on the worst-case optimality gap for~\cite{nedic_T-AC}}
\label{subsection-appendix-lower-bound-nedic}
Consider the $f_i$'s in~\eqref{eqn-f-i-s-hard}, the initialization $x_i(0)=(1,0)^\top$, $i=1,2,$ and
$W_{12}=W_{21}=1-W_{11}=1-W_{22}=w=1/8$, as we set in Subsection~\ref{subsection-comparisons}.
We divide the proof in four steps. First, we prove certain properties of~\eqref{eqn-opt-prob-original}
 and the $f_i$'s in~\eqref{eqn-f-i-s-hard}; second,
 we solve for the state $x(k)=(x_1(k)^\top,x_2(k)^\top)^\top$ with the algorithm in~\cite{nedic_T-AC};
 third, we upper bound $\|x(k)\|$; finally, we use the latter bound to
 derive the $\Omega(1/k^{2/3})$ worst-case optimality gap.

\subsubsection*{Step~1: Properties of the $f_i^{\theta}$'s}
Consider the $f_i^{\theta}$'s in~\eqref{eqn-f-i-s-hard} for a fixed $\theta \in [0,1]$.
 The solution to~\eqref{eqn-opt-prob-original}, with $f(x)=f_1^{\theta}(x)+f_2^{\theta}(x)$, is
${x^\star}=(0,0)^\top$, and the corresponding optimal value is
 $f^\star = {\theta}+1.$
%
% \begin{equation}
% \label{eqn-region}
% \mathcal{R}_i = \left\{x \in {\mathbb R}^2:\, {\theta}(x^{(1)}+(-1)^i)^2 +(x^{(2)}+(-1)^i)^2 \leq \overline{\chi}^2 \right\},
% \end{equation}
%  and outside of this region it behaves as the norm $x \rightarrow \|x\|.$
%
%
%
%
%
%
%
%
Further, the $f_i^{\theta}$'s
belong to the class $\mathcal{F}(L=\sqrt{2},G=10).$ (Proof is in~\cite{arxivVersion}.)

\subsubsection*{Step~2: Solving for $x(k)$ with the algorithm in~\cite{nedic_T-AC}}
Now, consider the algorithm in~\cite{nedic_T-AC},
and consider $x_i(k)$--the solution estimate
 at node $i$ and time $k$.
 Denote by $x^{l}(k)=(x_1^{(l)}(k),x_2^{(l)}(k))^\top$--the
 vector with the $l$-th coordinate
  of the estimate of both nodes, $l=1,2$; and
  $d^l(k) =
  \left( \frac{ \partial  f_1 (x_1(k))} {\partial x^{(l)}},\,   \frac{ \partial f_2 (x_2(k))} {\partial x^{(l)}}\right)^\top$, $l=1,2$.
 Then, the update rule of~\cite{nedic_T-AC} is, for the
 $f_1^{\theta},f_2^{\theta}$ in~\eqref{eqn-f-i-s-hard}:
 \begin{eqnarray}
 \label{eqn-nedic-update-gen}
 x^l(k) &=& W x^l(k-1) - \alpha_{k-1} d^l(k-1) \\
 k&=&1,2,...,\:\:\:l=1,2. \nonumber
 \end{eqnarray}
 %
 %
 %
 %\textbf{Performance evaluation of~\cite{nedic_T-AC}}.
 Recall the ``hard'' initialization $x^{\mathrm{I}}(0)=(1,1)^\top$, $x^{\mathrm{II}}(0)=(0,0)^\top.$
 Under this initialization:
 {\allowdisplaybreaks{
 \begin{eqnarray}
 \label{eqn-x-i-k-belongs-R-i}
 &\,&x_i(k)  \in \mathcal{R}_i:=\\
 &\,&
 \left\{x \in {\mathbb R}^2:\, {\theta}(x^{(1)}+(-1)^i)^2 +(x^{(2)}+(-1)^i)^2 \leq {\overline \chi}^2 \right\},\nonumber
 \end{eqnarray}}}
for all $k$, for both nodes $i=1,2$ (proof in~\cite{arxivVersion}.)
Note that $\mathcal{R}_i$ is the region where the $f_i^{\theta}$ in~\eqref{eqn-f-i-s-hard} is quadratic.
   Thus, evaluating $\nabla f_i^{\theta}$'s in the quadratic region:
%
%We fist prove that, if $\|x^{\mathrm{I}}(k)\|\leq 2\sqrt{2}$, and $\|x^{\mathrm{II}}\| \leq 2\sqrt{2}$, then
%  $x_i(k) \in \mathcal{R}_i$, $i=1,2.$ Consider
%  node 1'$s$ estimate $x_1(k)$. If
%   $\|x^{\mathrm{I}}(k)\|\leq 2\sqrt{2}$, and $\|x^{\mathrm{II}}\| \leq 2\sqrt{2}$, then
%  $\|x_1^{(l)}(k)\| \leq 2\sqrt{2}$, $l=1,2$, and:
%  \[
%  {\theta} (x_1^{(1)}(k)-1)^2 + (x_1^{(2)}(k)-1)^2 \leq 2 (2\sqrt{2}+1)^2 <  2 (2\frac{3}{2}+1)^2 < 32 < \overline{\chi}^2=36,
%  \]
%  which means $x_1(k) \in \mathcal{R}_1.$ (Analogously,
%  we can show $x_2(k) \in \mathcal{R}_2.$)
%
%We next prove that $\|x^{l}(k)\|\leq 2\sqrt{2}$, $l=1,2,$ for all $k$;
%we do this by induction. For $k=0$, $\|x^l(0)\|\leq 2\sqrt{2}$, $l=1,2.$
% Now, suppose that, for some $k\geq 1$,
%$\|x^l(k-1)\|\leq 2\sqrt{2}$, $l=1,2.$
% Then, the update equations~\eqref{eqn-nedic-update-gen}
% to get $x^{\mathrm{I}}(k)$ and $x^{\mathrm{II}}(k)$, using
% the derivatives of the $f_i^{\theta}$'s in the quadratic region in~\eqref{eqn-gradient}, become:
 %
 %
 \begin{eqnarray}
 \label{eqn-nedic-alg}
 x^{l}(k) = \left( W-\alpha_{k-1} {\kappa^l} I\right) x^{\mathrm{I}}(k-1) - \alpha_{k-1} {\kappa^l} \left( -1,1\right)^\top,
% \label{eqn-nedic-alg-2}
% x^{\mathrm{II}}(k) &=& \left( W-\alpha_{k-1}  I\right) x^{\mathrm{II}}(k-1) - \alpha_{k-1}  \left( -1,1\right)^\top,\:\forall k=1,2,...
 \end{eqnarray}
 $l=1,2,$ where $\kappa^{\mathrm{I}}=\theta$ and $\kappa^{\mathrm{II}}=1.$
 %
%
%
% From~\eqref{eqn-nedic-alg}--\eqref{eqn-nedic-alg-2}, using the sub-additive and
% sub-multiplicative properties of norms, and the step-size
% value $\alpha_{k-1}=c/(k^\tau)$:
% \begin{eqnarray*}
% \|x^{\mathrm{I}}(k)\| &\leq& \left( 1-\frac{c {\theta}}{k^\tau} \right) \|x^{\mathrm{I}}(k-1)\| + \frac{ c {\theta}}{k^\tau} \sqrt{2} \\
%            &=&  \|x^{\mathrm{I}}(k-1)\| - \frac{{\theta} c}{k^\tau} \left( \|x^{\mathrm{I}}(k-1)\| - \sqrt{2}\right).
%\end{eqnarray*}
%Now, we distinguish two cases: 1) $\|x^{\mathrm{I}}(k)\|\in [0, \sqrt{2}]$; and 2) $\|x^{\mathrm{I}}(k)\|\in (\sqrt{2}, 2\sqrt{2}]$.
% In case~1, from the last equation:
% \[
%  \|x^{\mathrm{I}}(k)\| \leq \|x^{\mathrm{I}}(k-1)\| + \frac{\sqrt{2}{\theta} c}{k^\tau} \leq 2 \sqrt{2},
% \]
% where we used $0\leq c\leq 1/(2\sqrt{2}) = 1/(2L)$ and $0 \leq {\theta} \leq 1.$
% In case~2:
% \[
% \|x^{\mathrm{I}}(k)\| < \|x^{\mathrm{I}}(k-1)\| \leq 2\sqrt{2}.
% \]
% Thus, we have shown that $\|x^{\mathrm{I}}(k)\| \leq 2\sqrt{2}$.
%  Similarly, we can show that $\|x^{\mathrm{II}}(k)\| \leq 2\sqrt{2}$.
%  Thus, by induction, $\|x^l(k)\| \leq 2\sqrt{2}$, $l=1,2,$ for all $k$, and so
%  $x_i(k) \in \mathcal{R}_i$, $i=1,2,$ for all $k$.
%
%
%
  We now evaluate $\sum_{i=1}^2 \left( f(x_i(k))-f^\star \right),$ $f(x)=f_1^{\theta}(x)+f_2^{\theta}(x)$.
     Because $x_i(k) \in \mathcal{R}_i$, $i=1,2,$
     verify, using~\eqref{eqn-f-i-s-hard}, and
     $f^\star=1+{\theta}$, that:
     \begin{equation}
     \label{eqn-sum-opt-gaps}
     \sum_{i=1}^2 \left(f(x_i(k))-f^\star \right) = {\theta} \|x^{\mathrm{I}}(k)\|^2 + \|x^{\mathrm{II}}(k)\|^2.
     \end{equation}
   %
   %
%   Because $x_i(k) \in \mathcal{R}_i$, for all $k$, the recursions~\eqref{eqn-nedic-alg}--\eqref{eqn-nedic-alg-2}
%     hold for all $k=0,1,2,...$, and we can write the solution
%      for $x^{\mathrm{I}}(k)$ and $x^{\mathrm{II}}(k)$.
      By unwinding~\eqref{eqn-nedic-alg},
      and using $x^{\mathrm{I}}(0)=(1,1)^\top$, $x^{\mathrm{II}}(0)=(0,0)^\top$:
      {\allowdisplaybreaks{
     \begin{eqnarray*}
     %\label{eqn-x-1-sol}
     &x^{\mathrm{I}}(k)& = \\
     &\,& \hspace{-7mm}\left( W- \alpha_{k-1} {\theta} I\right)\left( W - \alpha_{k-2} {\theta}I\right)...
     \left( W - \alpha_{0} {\theta} I\right)( 1,1)^\top \\
     &+&
     \theta\,( \sum_{t=0}^{k-2} (W - \alpha_{k-1}{\theta} I)(W - \alpha_{k-2}{\theta} I) ...\\
     &\times& (W - \alpha_{t+1}{\theta} I) \alpha_t + \alpha_{k-1} I) (1,-1)^\top \nonumber \\
     %\label{eqn-x-2-xol}
     &x^{\mathrm{II}}(k)& =
     \\
     &\,&(\sum_{t=0}^{k-2} (W-\alpha_{k-1}I)(W - \alpha_{k-2}I) ...\\
     &\times& (W - \alpha_{t+1}I) \alpha_t
      + \alpha_{k-1} I) (1,-1)^\top .
     \end{eqnarray*}
     }}
      Consider the eigenvalue decomposition $W=Q \Lambda Q^\top$,
      where $Q=[q_1,q_2]$, $q_1=\frac{1}{\sqrt{2}}(-1,1)^\top$, $q_2=\frac{1}{\sqrt{2}}(1,1)^\top$,
and $\Lambda$ is diagonal with the eigenvalues $\Lambda_{11}=\lambda_1 = 1-2w = 3/4$, $\Lambda_{22}=\lambda_2=1.$
       The matrix $W- \alpha_{k-1} {\theta} I$ decomposes
        as $W- \alpha_{k-1} {\theta} I = Q(\Lambda - \alpha_{k-1}{\theta} I)Q^\top$; likewise,
        $W- \alpha_{k-1}  I = Q(\Lambda - \alpha_{k-1} I)Q^\top$.
         Then, $(W - \alpha_{k-1}{\theta} I)(W - \alpha_{k-2}{\theta} I) ...(W - \alpha_{t+1}{\theta} I)
         = Q (\Lambda - \alpha_{k-1}{\theta} I)...(\Lambda - \alpha_{t+1}{\theta} I) Q^\top$,
          and $(W - \alpha_{k-1} I)...(W - \alpha_{t+1} I)
         = Q (\Lambda - \alpha_{k-1} I)...(\Lambda - \alpha_{t+1} I) Q^\top$.
      Using these decompositions, and the orthogonality: $q_1^\top (1,1)^\top = 0$, and $q_2^\top (-1,1)^\top = 0$:
      {\allowdisplaybreaks{
     \begin{eqnarray}
     \label{eqn-x-1-sol}
     x^{\mathrm{I}}(k)\hspace{-2.5mm} &=& \hspace{-2.7mm}\left( 1 - \alpha_{k-1} {\theta}\right)\left( 1 - \alpha_{k-2} {\theta}\right)...
     \left( 1 - \alpha_{0} {\theta}\right)(1,1)^\top \\
     &+&
     {\theta} (1,-1)^\top ( \sum_{t=0}^{k-2} (\lambda_1 - \alpha_{k-1}{\theta})(\lambda_1 - \alpha_{k-2}{\theta}) ...\nonumber\\
     &\times& (\lambda_1 - \alpha_{t+1}{\theta}) \alpha_t + \alpha_{k-1} )\nonumber \\
     \label{eqn-x-2-xol}
     x^{\mathrm{II}}(k) &=&
     (1,-1)^\top (\sum_{t=0}^{k-2} (\lambda_1-\alpha_{k-1})(\lambda_1 - \alpha_{k-2}) ...\\
     &\times&(\lambda_1 - \alpha_{t+1}) \alpha_t
      + \alpha_{k-1}).\nonumber
     \end{eqnarray}}}

\subsubsection*{Step~3: Upper bounding $\|x(k)\|$}
     %We upper bound $\|x^{\mathrm{I}}(k)\|$.
     Note that $\lambda_1 - \alpha_{k-1} {\theta} =
     3/4 - \frac{c {\theta} } {k^\tau} \geq  1/4$, for all $k,\tau,c.$
      Also, $\lambda_1 - \alpha_{k-1} {\theta} \leq \lambda_1 = 3/4$,
      for all $k,\tau,c.$ Similarly, we can show $1 - \alpha_{k-1} {\theta} \in [1/2,1];$ then,
        $(1-\alpha_{k-1}{\theta})...(1-\alpha_0 {\theta}) \geq 0$,
        $(\lambda_1-\alpha_{k-1}{\theta})...(\lambda_1-\alpha_{t+1} {\theta}) \geq 0$,
         and $(\lambda_1-\alpha_{k-1})...(\lambda_1-\alpha_{t+1}) \geq 0$, $\forall t$. Thus:
       $
       \|x^{\mathrm{I}}(k)\| \geq
        \left( 1 - \alpha_{k-1} {\theta}\right)\left( 1 - \alpha_{k-2} {\theta}\right) ... \left( 1 - \alpha_{0} {\theta}\right).
       $
       Set $\theta=\theta_k=1/(s_k(\tau)) \leq 1$, where $s_k(\tau):=\sum_{t=0}^{k-1} (t+1)^{-\tau};$ use $(1-a_1)(1-a_2)...(1-a_n) \geq 1-(a_1+a_2+...+a_n)$, $a_i \in [0,1)$, $\forall i$; and
       $\alpha_k = \frac{c}{(k+1)^\tau}$. We obtain:
        $\|x^{\mathrm{I}}(k)\| \geq 1-c\,\theta_k\,s_k(\tau)$, and so:
       %
       %
       %\begin{eqnarray}
%       \label{eqn-to-square}
%       \|x^{\mathrm{I}}(k)\| \geq
%        \left( 1 -  {\theta}\,c\,s_k(\tau)\right) ,\:\:\:\:\:s_k(\tau):=\sum_{t=0}^{k-1} (t+1)^{-\tau}.
%       \end{eqnarray}
%       %
%       We use in~\eqref{eqn-to-square}
%       the inequality $\sum_{t=0}^{k-1} \lambda_1^{k-t-1}\frac{1}{(t+1)^\tau} \leq C(\lambda_1) \frac{1}{k^\tau}$,
%       $\forall k \geq 1$, $\forall \tau \in [0,1]$,
%       where $C(\lambda_1) \in (0,\infty)$
%        is a constant that depends only on $\lambda_1$
%         and is independent of~$k,\tau.$
       %
       %
       %We now set $\theta=\theta_k=\frac{1}{2\,c s_k (\tau)}$
%        in~\eqref{eqn-to-square} and obtain:
%       %
       %
       %
       $
       \theta_k \|x^{\mathrm{I}}(k)\|^2
       %
%       \theta_k\, \left\{ \left( 1 -  2{\theta}_k\,c\,s_k(\tau)\right) -
%        2\,\left({\theta} c C(\lambda_1)\frac{1}{k^\tau}\right)\left( 1 -  {\theta}_k\,c\,s_k(\tau)\right) \right\}%\nonumber
%%        \\
%%               \label{eqn-x-1-lower-bnd}
%%       &\geq& \theta_k \left( \frac{1}{2}-\frac{3 C(\lambda_1)}{8 \log k} \right)
       \geq \frac{(1-c_{\mathrm{max}})^2}{s_k(\tau)},
       $
       where we denote $c_{\mathrm{min}}:=c_0$ and $c_{\mathrm{max}}:=1/(2 L)=1/(2 \sqrt{2}).$
       Further, from~\eqref{eqn-x-2-xol}:
      % \[
%       \|x^{\mathrm{II}}(k)\|^2 \geq \alpha_{k-1}^2 = \frac{c^2}{k^{2 \tau}}.
%       \]
%       and so:
       $
       \|x^{\mathrm{II}}(k)\|^2 \geq \alpha_{k-1}^2 \geq \frac{c_{\mathrm{min}}^2}{k^{2 \tau}},
       $
       and we obtain:
       \vspace{-0mm}
       \begin{equation}
       \label{eqn-x-1-x-2-ower-bnd}
       \theta_k\,\|x^{\mathrm{I}}(k)\|^2+\|x^{\mathrm{II}}(k)\|^2 \geq \frac{(1-c_{\mathrm{max}})^2}{s_k(\tau)}+\frac{c_{\mathrm{min}}^2}{k^{2\tau}}.
       \end{equation}
       \subsubsection*{Step~4: Upper bounding the optimality gap from~\eqref{eqn-x-1-x-2-ower-bnd}}
       From~\eqref{eqn-x-1-x-2-ower-bnd},
       and using~\eqref{eqn-sum-opt-gaps}:
       \begin{eqnarray}
       &\,&\max_{i=1,2} \left( f(x_i(k)) -f^\star \right) \geq
       \frac{1}{2} \sum_{i=1}^2 \left( f(x_i(k)) -f^\star \right) \nonumber \\
       \label{eqn-inf-e-k}
       &\geq& \frac{(1-c_{\mathrm{max}})^2}{2 s_k(\tau)} + \frac{c_{\mathrm{min}}^2}{2\,k^{2 \tau}}=:e_k(\tau),
       \end{eqnarray}
       $\forall k \geq 1$, $\forall \tau \geq 0.$
       %
       %
       %
       %Let~$\mathcal{S}:=\left\{ (c,\tau) \in {\mathbb R}^2:\,c\in (0,1/(2L)],\,\tau\in [0,\infty)\right\}$.
        We further upper bound the right hand side in~\eqref{eqn-inf-e-k} by taking the infimum
        of $e_k(\tau)$ over $\tau \in [0,\infty)$; we split
         the interval $[0,\infty)$ into $[0,3/4]$; $[3/4,1],$
           and $[1,\infty),$
           so that
           \begin{equation}
           \label{eqn-min-infima}
            \inf_{[0,\infty)} e_k(\tau)
            \geq \min \left\{ \inf_{[0,3/4]}e_k(\tau),\,\inf_{[3/4,1]}e_k(\tau),\,\inf_{[1,\infty)}e_k(\tau) \right\}.
           \end{equation}
       It is easy to prove that: 1)~$\inf_{[0,3/4)}e_k(\tau)=\Omega(1/k^{2/3})$;
        2)~using $
        s_k(\tau)\leq 3 (\log k) (k+1)^{1-\tau},
       $ $\forall k \geq 3$, $\forall \tau \in [0,1]$,
        that $\inf_{[3/4,1]}e_k(\tau)=\Omega\left(\frac{1}{(\log k) k^{1/4}}\right)$;
        and 3)~$\inf_{[1,\infty)}e_k(\tau)=\Omega\left( \frac{1}{\log k}\right)$. (see~\cite{arxivVersion}.)
        Combining the latter bounds with~\eqref{eqn-min-infima} completes the proof
       of~\eqref{eqn-claim-nedic-hard}.

\subsection{Relaxing bounded gradients: Proof of~\eqref{eqn-D-NC-worst-case-lower-bnd}
 for D--NC}
\label{subsection-lower-bound-without-bdd-grad}
We prove~\eqref{eqn-D-NC-worst-case-lower-bnd} for D--NC
while the proof of D--NG is similar and is in~\cite{arxivVersion}.
%\textbf{Algorithm D--NC}.
Fix arbitrary $\theta>0$ and take the $f_i$'s in~\eqref{eqn-hard-f-i-s-D-NC-Lb}.
From~\eqref{eqn-our-alg-cons}--\eqref{eqn-our-alg-cons-druga}, evaluating the~$\nabla f_i$'s:
 {\small{
 \begin{eqnarray}
 \label{eqn-update-x-y-d-nc-proof}
 x(k) \hspace{-0.5mm} &=& \hspace{-0.5mm} (1-\alpha)W^{\tau_x(k)}y(k-1) +\alpha\,\theta \,W^{\tau_x(k)}(1,-1)^\top\\
 y(k)\hspace{-0.5mm} &=& \hspace{-0.5mm} W^{\tau_y(k)}\left( x(k)+\beta_{k-1}(x(k)-x(k-1))\right), \nonumber
 \end{eqnarray}
 }}
for $k=1,2,...$ We take the initialization at the solution $x(0)=y(0)=(0,0)^\top.$
 Consider the eigenvalue decomposition $W=Q \Lambda Q^\top$, with
 $Q=[q_1,q_2]$, $q_1=\frac{1}{\sqrt{2}}(1,-1)^\top$,
$q_2=\frac{1}{\sqrt{2}}(1,1)^\top$, and $\Lambda$
 is diagonal with $\Lambda_{11}=\lambda_1$, $\Lambda_{22}=\lambda_2=1.$
 Define ${z}(k)=Q^\top x(k)$ and ${w} (k) = Q^\top y(k)$.
 Multiplying~\eqref{eqn-update-x-y-d-nc-proof} from the left by $Q^\top$,
 and using $Q^\top(1,-1)^\top = (\sqrt{2},0)^\top$:
 \begin{eqnarray}
 {z}(k) &=& (1-\alpha)\Lambda^{\tau_x(k)}{w}(k-1)
 +\alpha\,\theta \Lambda^{\tau_x(k)}(\sqrt{2},0)^\top \nonumber \\
  \label{eqn-update-x-y-d-nc-proof-2}
 {w}(k) &=& \Lambda^{\tau_y(k)}\left[ {z}(k)+
 \beta_{k-1}({z}(k)-{z}(k-1))\right],
 \end{eqnarray}
$k=1,2,...$, and ${z}(0)={w}(0)=(0,0)^\top.$
Next, note that
\vspace{-3mm}
\begin{eqnarray}
\max_{i=1,2}(f(x_i(k))-f^\star) &\geq& \frac{1}{2}\sum_{i=1}^2(f(x_i(k))-f^\star)= \frac{\|x(k)\|^2}{2} \nonumber \\
\label{eqn-x-1-prime-to-bnd-opt-gap}
&=& \frac{\|{z}(k)\|^2}{2}\geq \frac{({z}^{(1)}(k))^2}{2}. %\nonumber
\end{eqnarray}
Further, from~\eqref{eqn-update-x-y-d-nc-proof-2} for the first coordinate ${z}^{(1)}(k),w^{(1)}(k)$,
recalling that $\mu:=\lambda_1$:
{
{
 \begin{eqnarray}
 \label{eqn-update-x-y-d-nc-proof-3}
 \|{z}^{(1)}(k)\| &\leq& \mu^{\tau_x(k)}\|{w}^{(1)}(k-1)\|
 +\sqrt{2}\,\alpha\,\theta \,\mu^{\tau_x(k)}\\
 \|{w}^{(1)}(k)\| &\leq&
 \mu^{\tau_y(k)} \left(2\|{z}^{(1)}(k)\|+
 \|z^{(1)}(k-1)\|\right), \nonumber
 \end{eqnarray}}}
$k=1,2,...$ Note that~\eqref{eqn-update-x-y-d-nc-proof-3}
is analogous to~\eqref{eqn-111}--\eqref{eqn-222}
with the identification $\widetilde{x}(k) \equiv z^{(1)}(k)$,
 $\widetilde{y}(k) \equiv w^{(1)}(k)$,
 $\sqrt{N}G \equiv \sqrt{2}\theta$; hence,
analogously to the proof of Theorem~\ref{theorem-distance-cons-consensus-alg},
from~\eqref{eqn-update-x-y-d-nc-proof-3}:
$
\|w^{(1)}(k-1)\| \leq \frac{2\,\sqrt{2}\,\alpha\, \theta}{(k-1)^2},\:\:k=2,3...
$%\end{equation}
Using the latter, \eqref{eqn-update-x-y-d-nc-proof-2}, and $\frac{1}{k^2} \geq \mu^{\tau_x(k)} \geq \frac{1}{e\,k^2}$~(see~\eqref{eqn-consensus}):
$
\|z^{(1)}(k)\| \geq \alpha\,\theta \sqrt{2}\mu^{\tau_x(k)} - \mu^{\tau_x(k)}\|w^{(1)}(k-1)\| \geq
\frac{\alpha \theta \sqrt{2}}{e\,k^2}\left(1-\frac{2\,e}{(k-1)^2}\right)
\geq \frac{\alpha\,\theta\sqrt{2}}{4\,k^2} >0,\:\forall k \geq 10.
$
Thus, from~\eqref{eqn-x-1-prime-to-bnd-opt-gap} and the latter inequality,
$\max_{i=1,2}(f(x_i(k))-f^\star) \geq \frac{\alpha^2 \theta^2}{16\, k^4}$,
which is, for $\alpha=1/(2L)=1/2$,
greater or equal $M$ for $\theta=\theta(k,M)=8\, \sqrt{M} \,k^2.$

\bibliographystyle{IEEEtran}
\bibliography{IEEEabrv,bibliographyNesterovNew}
\end{document}